\RequirePackage{fix-cm}

\documentclass[twocolumn,nospthms]{svjour3}
\smartqed
\usepackage[square,numbers,sort&compress]{natbib}
\usepackage[cp1251]{inputenc}
\usepackage{graphicx}
\usepackage{amsmath,amssymb}
\usepackage{amsthm}
\usepackage{mathptmx}
\DeclareMathAlphabet{\mathcal}{OMS}{cmsy}{m}{n}

\usepackage{tabularx,multirow}
\usepackage{booktabs}
\usepackage{hyperref}
\hypersetup{pdfpagemode=UseOutlines,pdfstartview=FitH,bookmarks=true,bookmarksnumbered=true,bookmarksopen=true,colorlinks=true,linkcolor=blue,citecolor=blue,urlcolor=blue}
\usepackage{bookmark}
\usepackage{cuted}
\journalname{Applied Physics A}

\theoremstyle{definition}
\newtheorem{prop}{Proposition}
\renewcommand{\qedsymbol}{$\blacksquare$}

\begin{document}

\title{Calculation of electron transport in branched semiconductor nanostructures using quantum network model}
\author{ D. E. Tsurikov }
\institute{	D. E. Tsurikov \at
				Spin Optics Laboratory, St. Petersburg State University,\\
				Ulyanovskaya 1, Peterhof, St. Petersburg, 198504, Russia\\
				\email{DavydTsurikov@mail.ru}\\
				ORCID ID: 0000-0001-6677-7427
}

\date{\today}

\maketitle

\begin{abstract}
Electron transport in branched semiconductor nanostructures provides many possibilities for creating fundamentally new devices. We solve the problem of its calculation using a quantum network model. The proposed scheme consists of three computational parts: S-matrix of the network junction, S-matrix of the network in terms of its junctions' S-matrices, electric currents through the network based on its S-matrix. To calculate the S-matrix of the network junction, we propose scattering boundary conditions in a clear integro-differential form. As an alternative, we also consider the Dirichlet-to-Neumann and Neumann-to-Dirichlet map methods. To calculate the S-matrix of the network in terms of its junctions' S-matrices, we obtain a network combining formula. We find electrical currents through the network in the framework of the Landauer--B\"uttiker formalism. Everywhere for calculations, we use extended scattering matrices, which allows taking into account correctly the contribution of tunnel effects between junctions. We demonstrate the proposed calculation scheme by modeling nanostructure based on two-dimensional electron gas. For this purpose we offer a model of a network formed by smooth junctions with one, two and three adjacent branches. We calculate the electrical properties of such a network (by the example of GaAs), formed by four junctions, depending on the temperature.
\end{abstract}

\section{Introduction}	\label{sec:1}

An electron transport in branched semiconductor nanostructures provides many possibilities for creating fundamentally new devices \cite{Bib001,Bib002,Bib003,Bib004,Bib005,Bib006,Bib007,Bib008,Bib009,Bib010,Bib011,Bib012,Bib013,Bib014}. Modern experiments have shown that it is always a combination of diffusion transport and ballistic one \cite{Bib005,Bib006,Bib007}. Ballistic transport corresponds to the coherent motion of electrons and, therefore, relates to quantum effects in the system. At the same time, it also occurs in structures whose size is an order of magnitude greater than the free length of charge carriers \cite{Bib005,Bib006}, and in materials with their low mobility \cite{Bib007}. Therefore, ballistic transport is relevant for a much wider range of structures than “purely” quantum ones, and its accounting is fundamentally important for their correct modeling.

A mathematical model describing ballistic electron transport in branched semiconductor nanostructures is a quantum network \cite{Bib015,Bib016}. Electrical currents through the network are calculated in the framework of the Landauer--B\"uttiker formalism \cite{Bib017,Bib018} based on its transparency. In turn, the transparency is determined by the scattering matrix of network.

The scattering matrix of the quantum network junction can be found by solving the Schr\"odinger equation in it with scattering boundary conditions (SBC) at the boundaries with branches \cite{Bib019,Bib020,Bib021}. SBC provide matching of wave function in junction with incident and scattered waves in branches. Therefore, this method is universal and is also used for quantum self-consistent calculation of current through nanostructures \cite{Bib021}. In the literature, the form of SBC is usually defined by a concrete problem. Such approach reduces their clearness. Also, their “non-standard” form makes them difficult to use in computing packages.

It is more convenient to apply the methods based on Dirichlet and Neumann boundary conditions: Dirichlet-to-Neumann and Neumann-to-Dirichlet map (DN- and ND-map) methods \cite{Bib015,Bib016,Bib022,Bib023}. To calculate the scattering matrix of semiconductor nanostructures, the formalism of the R-matrix \cite{Bib024,Bib025,Bib026} is also widespread. This method is similar to the ND-map method, since the definition of the R-matrix is identical to the definition of the ND-map, and the expression for the scattering matrix has the same form \cite{Bib026}. The disadvantage of the DN- and ND-map methods is their low clearness, since additional mathematical considerations are required to solve the initial scattering problem.

Any quantum network can be considered as one junction with a complex structure, and one of the mentioned above methods \cite{Bib015,Bib016,Bib017,Bib018,Bib019,Bib020,Bib021,Bib022,Bib023,Bib024,Bib025,Bib026} can be used to calculate its scattering matrix. However, for calculations, it may be more convenient to find a scattering matrix of network using the rule of combining scattering matrices of its junctions \cite{Bib017,Bib027}. This approach becomes particularly effective when the network contains the same junctions. In its use, as a rule, the boundaries of the junctions are selected so that branches connecting junctions are excluded from consideration. This formal technique can simplify calculations, but makes it difficult to study the transport properties of a quantum network depending on the lengths of the branches. The disadvantage of this method is also the absence of an explicit formula for the scattering matrix of the entire network in terms of its junctions' scattering matrices.

In this work, we will develop the methods discussed above and propose on their basis a scheme for calculating electron transport in branched semiconductor nanostructures. The scheme will be formulated in terms of an extended scattering matrix of the quantum network \cite[p.~155]{Bib027}. This approach will allow correctly taking into account tunnel effects in system. For an arbitrary network junction, we will obtain scattering boundary conditions in an integro-differential form. Using them, we will write the procedure for calculating its extended scattering matrix, as well as the calculation scheme using the DN- and ND-map methods. We will obtain a formula for an extended scattering matrix of a quantum network in terms of extended scattering matrices of its junctions, taking into account the lengths of the branches connecting them. Finally, the Landauer--B\"uttiker formalism, along with all other methods, will be written in the framework of a single notation system. We will demonstrate the effectiveness of the proposed calculation scheme by modeling a nanostructure based on two-dimensional electron gas using a model of a quantum network of smooth junctions with one, two and three adjacent branches.

\section{Scattering of charge carrier in quantum network}	\label{sec:2}

\subsection{Agreements and notation}	\label{sec:2_1}

\pagebreak
\subsubsection{Agreements}	\label{sec:2_1_1}

We use the following agreements to simplify the presentation.

I. \label{agr:I} \textit{Computations}. In computations, we use the symbols of implication $\Rightarrow $ and equivalence $\Leftrightarrow $, equality $:=$ and equivalence $:\Leftrightarrow $ by definition. The ellipsis symbol is used to exclude obvious computations in the implication chain; the semicolon symbol is used to write multi-level computations to a string:
\begin{equation}	\label{Eq001}
	A\Rightarrow B;C\Rightarrow...\Rightarrow D
	\quad :\Leftrightarrow \quad
	\left. \begin{aligned}
		&A\Rightarrow B \\
		&C
	\end{aligned} \right\}
	\Rightarrow...\Rightarrow D
\end{equation}
where $A$, $B$, $C$ and $D$ are statements or references to them (definitions, formulas, etc.). In some cases, we use references to statements without naming them

II. \label{agr:II} \textit{Indices and ranges}. The ranges of the superscripts and subscripts values in the enumerations and sums are set at the appropriate levels, for example: $\{d_n^l\}_{n\in \mathbb{B}}^{l\in \mathbb{A}}$, $\sum\nolimits_{n\in \mathbb{B}}^{l\in \mathbb{A}}{d_n^l}$. The implicitly defined range of index values is determined by its location relative to the letter and its semantics. By default, any enumeration in this work is ordered~--- \textit{tuple} \cite[p.~33]{Bib028}.

III. \label{agr:III} \textit{Vectors and matrices}. Enumeration by one index (for example, $d = \{d^l\}^l$) is a column vector, by two indexes at the same level (for example, $D = \{D^{kl}\}^{kl}$) is a matrix: the first index is a row number, the second one is a column number. Objects with tuple indexes outside square brackets are enumerations by all tuple elements. For example, $d^{\mathbb{A}} = \{{d^k}\}^{k\in \mathbb{A}}$ is vector, $D^{\mathbb{A}\mathbb{B}} = \{D^{kl}\}^{k\in \mathbb{A},l\in \mathbb{B}}$ is matrix,
\begin{equation}	\label{Eq002}
	\begin{aligned}
		&\mathbb{A} := \{1,2\},\quad \mathbb{B} := \{3,4,5\}\quad \Rightarrow \\
		&D^{\mathbb{A}\mathbb{B}} = {D^{\{1,2\}\{3,4,5\}}} =
		\left[ \begin{matrix}
			{D^{13}} & {D^{14}} & {D^{15}} \\
			{D^{23}} & {D^{24}} & {D^{25}} \\
		\end{matrix} \right]
	\end{aligned}
\end{equation}
The following rules are also executed: $d^{\varnothing} = \{d^n\}^{n\in \varnothing} = \varnothing$ is empty vector, $D^{\mathbb{A}\varnothing} = \{D^{pq}\}^{p\in \mathbb{A},q\in \varnothing} = \varnothing = \{D^{pq}\}^{p\in \varnothing,q\in \mathbb{B}} = D^{\varnothing \mathbb{B}}$ is empty matrix. Empty rows and columns are excluded from the matrices. $O$ is zero matrix, $I$ is identity matrix.

IV. \label{agr:IV} \textit{Two-level tuples in indexes}. Two-level tuples are tuples of the form $\{_m^k\}_m^k$. For example,
\begin{equation}	\label{Eq003}
	\begin{aligned}
		&\mathbb{A} := \{_1^0,_3^2\},\quad \mathbb{B} := \{_5^4,_7^6,_9^8\}\quad \Rightarrow \\
		&D_{\mathbb{A}\mathbb{B}} = D^{\mathbb{A}\mathbb{B}} = D\{_1^0,_3^2\}\{_5^4,_7^6,_9^8\} =
		\left[ \begin{matrix}
			D_{15}^{04} & D_{17}^{06} & D_{19}^{08} \\
			D_{35}^{24} & D_{37}^{26} & D_{39}^{28} \\
		\end{matrix} \right]
	\end{aligned}
\end{equation}

V. \label{agr:V} \textit{Dummy symbol} $\square $. A dummy symbol $\square $ is used to reduce definitions and statements. In enumerations, it is replaced by the values specified for it. For example,
\begin{equation}	\label{Eq004}
	\begin{aligned}
		&\{f(\square) = \ \square \ |\ \square \ = a,b,c,d,e\}\quad \Leftrightarrow	\\
		&\{f(a) = a,f(b) = b,f(c) = c,f(d) = d,f(e) = e\}
	\end{aligned}
\end{equation}

VI. \label{agr:VI} \textit{Differentiation}. ${\partial}_{\square}$ is the operator of partial differentiation, where $\square $ is the name of the variable (\textit{named differentiation}) or the argument number (\textit{numbered differentiation}) of the corresponding function. Reduction:
\begin{equation}	\label{Eq005}
	\forall f\quad \dot{f} := {{\partial}_1}f
\end{equation}

VII. \label{agr:VII} \textit{Iverson notation}  \cite[p.~24]{Bib029}. The brackets with the statement are 1 if it is true and 0 if it is false.

VIII. \label{agr:VIII} \textit{Line over the symbol}. The semantics of the line over the symbol is determined by the semantics of the symbol (Table~\ref{tab:1}).

\begin{table}[h!]
\caption{\label{tab:1}Semantics of line over symbols}
\begin{tabularx}{\linewidth}{>{\raggedright}l>{\raggedright}X}
\toprule
	semantics of symbol	&	semantics of symbol with line	\tabularnewline
\midrule
	number		&	complex conjugate number	\tabularnewline
	function	&	function with complex conjugate values	\tabularnewline
	operator	&	Hermitian conjugate operator	\tabularnewline
	domain		&	union of domain with its boundary	\tabularnewline
	tuple		&	complement tuple	\tabularnewline
\bottomrule
\end{tabularx}
\end{table}

IX. \label{agr:IX} \textit{Precedence rule}. Indexes at the symbol are taken last from top to bottom and from left to right: $A_{cd}^{ab} := (A^{ab})_{cd} = \{[(A^a)^b]_c\}_d$. The reverse precedence rule sets the index reduction rule. Arguments in brackets at functions are also taken last, when they are partially represented, ellipses are placed: $f(x,...)$.

X. \label{agr:X} \textit{Reduction of arguments and identifiers}. If it is necessary to specify which set the reduced arguments belong to, the words “in” and “on” (usually for domain boundaries) are used. For example,
\begin{equation}	\label{Eq006}
	\begin{aligned}
		&\left\{ \begin{aligned}
			&[-\Delta + \upsilon ]f = \varepsilon f && \operatorname{in}\ \ \Omega \\
			&f = 0 && \operatorname{on}\ \partial \Omega
		\end{aligned} \right.\quad \Leftrightarrow \\
		&\left\{ \begin{aligned}
			&[-\Delta +\upsilon ({\bf r})]f({\bf r}) = \varepsilon f({\bf r}), && {\bf r}\in \Omega \\
			&f({\bf r}) = 0, && {\bf r}\in \partial \Omega
		\end{aligned} \right.
	\end{aligned}
\end{equation}
When reducing the identifier in section, also uses the word “in”. For example, the reduction for symbol $\square $ with the top identifier is:
\begin{equation}	\label{Eq007}
	\{\ {{\square}^{\text{Identifier}}}\ \mapsto \ \square \ |\ \square \ = S,C,...\}\quad \operatorname{in}\ \text{Section}
\end{equation}

XI. \label{agr:XI} \textit{Bra-ket symbolism}. The semantics of objects in the bra-ket symbolism is determined by the semantics of objects in brackets. In particular, the scalar product for the functions of a discrete and continuous argument has the form
\begin{equation}	\label{Eq008}
	\begin{aligned}
	\langle a | b \rangle & = \sum\nolimits_n{[n\in \mathcal{D}(a){\textstyle \bigcap} \mathcal{D}(b)]{{{\bar{a}}}_n}{b_n}} \\
		\langle f | g \rangle & = \int{dx[x\in \mathcal{D}(f){\textstyle \bigcap} \mathcal{D}(g)]\bar{f}(x)g(x)}
	\end{aligned}
\end{equation}
respectively, where $\mathcal{D}(\square)$ is the function $\square $ domain.

XII. \label{agr:XII} \textit{Integration}. When integrating, the interval in the Iverson brackets is the 1D analogue of the oriented curve that defines the integration direction:
\begin{equation}	\label{Eq009}
	\begin{aligned}
		\int{dx[x\in (a,b)]\left\{... \right\}} &= \int_a^b{dx\left\{... \right\}} =	\\
		-\int_b^a{dx\left\{... \right\}} &= -\int{dx[x\in (b,a)]\left\{... \right\}}
	\end{aligned}
\end{equation}

\subsubsection{Notation}	\label{sec:2_1_2}

Physically, a quantum network consists of quantum dots connected to each other by quantum wires. Since the charge carrier can scatter both in the point and in the wire, it is convenient to introduce a functional definition: a \textit{quantum network} is a set of junctions and branches. An \textit{internal junction} is a network element in which the charge carrier scatters, an \textit{external junction} is a source or drain of charge carriers, a \textit{branch} is a network element in which the charge carrier does not scatter, an \textit{internal branch} connects two internal junctions, an \textit{external branch} connects the internal junction to the external one. \textit{Channels} are energy levels that arose due to size-quantization across the branch, with the associated movement of the charge carrier along the branch.

To unambiguously identify the elements of the network, we number its branches. In notation, the superscript is the branch number, the subscript is the channel number, the superscript in square brackets is the \textit{structural identifier} of the junction, the contents of the brackets is the tuple of the numbers of the branches adjacent to the junction (Fig.~\ref{fig:1}). In this work, we also use the notations: $\mathbb{I} := ({{\bigcup}^{\mathbb{A}\in \mathcal{N}}}\mathbb{A})\backslash \mathbb{E}$ is tuple of internal branches' numbers, $\mathbb{E} := {{\ominus}^{\mathbb{A}\in \mathcal{N}}}\mathbb{A}$ is tuple of external branches' numbers ($\mathbb{A}\ominus \mathbb{B} := (\mathbb{A}\backslash \mathbb{B})\bigcup (\mathbb{B}\backslash \mathbb{A})$ is symmetric difference of tuples $\mathbb{A}$ and $\mathbb{B}$), $\mathcal{N}$ is tuple of structural identifiers of internal junctions, containing information about all connections of junctions in the network. The concept of a junction is conditional: depending on its structural identifier $\mathbb{A}$, it can be either any section of the network or the entire network if $\mathbb{A} = \mathbb{E}$.

\begin{figure}[htb]\center
	\includegraphics{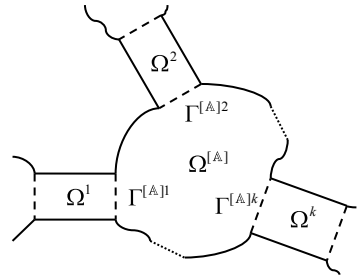}
	\caption{Section of quantum network: ${\Omega}^{[\mathbb{A}]}$ is junction of network, $\{{\Omega}^k\}^{k\in \mathbb{A}}$ are branches adjacent to junction ${\Omega}^{[\mathbb{A}]}$, ${\Gamma}^{[\mathbb{A}]k} := \partial {\Omega}^{[\mathbb{A}]} \bigcap \partial {\Omega}^k$ is boundary of junction ${\Omega}^{[\mathbb{A}]}$ with branch ${\Omega}^k$, ${\Gamma}^{[\mathbb{A}]} := {\bigcup}^{k\in \mathbb{A}}{{\Gamma}^{[\mathbb{A}]k}}$ is boundary of junction ${\Omega}^{[\mathbb{A}]}$ with adjacent branches.}\label{fig:1}
\end{figure}

In this work, we also use a \textit{type identifier}. It is placed in angle brackets and can be a number, letter, or word. A typical identifier is convenient, in particular, when specifying the structure of the network by the type of its internal junctions, for example:
\begin{equation}	\label{Eq010}
	\{{\Omega}^{[\mathbb{K}]} = {\Omega}^{\langle \text{A} \rangle} |\ a^{\langle \text{A} \rangle} = a^{[\mathbb{K}]}, b^{\langle \text{A} \rangle} = b^{[\mathbb{K}]}, ...\}^{\mathbb{K}\in \mathcal{A}}
\end{equation}
where ${\Omega}^{\langle \text{A} \rangle}$ is a junction of type “A” (A-junction), ${a^{\langle \text{A} \rangle}},{b^{\langle \text{A} \rangle}},...$ are parameters of A-junction, $\{a^{[\mathbb{K}]}, b^{[\mathbb{K}]}, ...\}^{\mathbb{K}\in \mathcal{A}}$ are values of parameters of A-junctions in the network, $\mathcal{A}$ is a tuple of structural identifiers of A-junctions.

\subsection{Problem definition}	\label{sec:2_2}

\subsubsection{Charge carrier in network}	\label{sec:2_2_1}

Let us consider a three-dimensional quantum network. The results obtained for it are easy to use for two-dimensional and one-dimensional networks. Without loss of generality, for simplicity of presentation, we consider an isotropic semiconductor. The motion of the charge carrier is described by the Schr\"odinger equation:
\begin{equation}	\label{Eq011}
	\left( -\frac{{{\hbar}^2}}{2m}\Delta +V \right)\varphi = E\varphi 
\end{equation}
where $\hbar $ is the Planck constant, $m$ is the effective mass of the charge carrier, $\Delta = \partial _1^2+\partial _2^2+\partial _3^2$ is the Laplace operator, $V$ is the potential, $E$ is the energy of the charge carrier. We write it in dimensionless form:
\begin{equation}	\label{Eq012}
	( -\Delta +\upsilon )\Psi = \varepsilon \Psi 
\end{equation}
\vspace{-17pt}
\begin{equation}	\label{Eq013}
	\upsilon ({\bf r}) := 2m{\hbar}^{-2}{L^2}V( L {\bf r} ),\quad \varepsilon := 2m{\hbar}^{-2}{L^2}E
\end{equation}
depending on the problem dimension, we have
\begin{equation}	\label{Eq014}
	\begin{aligned}
		\text{1D}:\quad & \Psi ({\bf r}) := L^{1/2}\varphi (L{\bf r}) \\
		\text{2D}:\quad & \Psi ({\bf r}) := L^{2/2}\varphi (L{\bf r}) \\
		\text{3D}:\quad & \Psi ({\bf r}) := L^{3/2}\varphi (L{\bf r})
	\end{aligned}
\end{equation}
where $L$ is the characteristic length (one can choose any for convenience in a particular problem). Then the motion of the charge carrier is described by boundary problems for dimensionless Schr\"odinger equations in the branches
\begin{equation}	\label{Eq015}
	\left\{ \begin{aligned}
		&[-\Delta + {\upsilon}^k]{\Psi}^k = \varepsilon {\Psi}^k && \operatorname{in}\ \ {\Omega}^k \\
		&{\Psi}^k = 0 && \operatorname{on}\ \partial {\Omega}^k \backslash {\Gamma}^k
	\end{aligned} \right.,\quad {\textstyle k\in \mathbb{I} \bigcup \mathbb{E}}
\end{equation}
where ${\Gamma}^k$ is the boundary of the branch ${\Omega}^k$ with adjacent junctions, in internal junctions of network
\begin{equation}	\label{Eq016}
	\left\{ \begin{aligned}
		&[-\Delta + {\upsilon}^{[\mathbb{A}]}]{\Psi}^{[\mathbb{A}]} = \varepsilon {\Psi}^{[\mathbb{A}]} && \operatorname{in}\ \ {\Omega}^{[\mathbb{A}]} \\
		&{\Psi}^{[\mathbb{A}]} = 0 && \operatorname{on}\ \partial {\Omega}^{[\mathbb{A}]}\backslash {\Gamma}^{[\mathbb{A}]}
	\end{aligned} \right.,\quad \mathbb{A}\in \mathcal{N}
\end{equation}
and matching conditions on the boundaries of the internal junctions with the branches:
\begin{equation}	\label{Eq017}
	\left\{ \begin{aligned}
		&{\Psi}^{[\mathbb{A}]} = {{\Psi}^k} \\
		&{{\partial}_n}{\Psi}^{[\mathbb{A}]} = {{\partial}_n}{\Psi}^k
	\end{aligned} \right.\quad \operatorname{on}\ {{\Gamma}^{[\mathbb{A}]k}},\quad k\in \mathbb{A}\in \mathcal{N}
\end{equation}

Equations (\ref{Eq015}), (\ref{Eq016}) and conditions (\ref{Eq017}) are written in the \textit{global coordinate frame} (\textit{GCF})~--- the coordinate frame associated with the network. Its location one can set in a specific problem for convenience reasons.

\subsubsection{Charge carrier in branch}	\label{sec:2_2_2}

To solve the problems (\ref{Eq015}), we introduce the concept of the \textit{local coordinate frame} (\textit{LCF}) at the boundary of the junction with the branch (Fig.~\ref{fig:2}). Coordinate frames are changed using the following operators: $W^{[\mathbb{A}]k}$ translates functions specified in GCF into functions specified in LCF $[XYZ]^{[\mathbb{A}]k}$, $w^{[\mathbb{A}]k}$ expresses global coordinates in terms of local ones. Operator $W^{[\mathbb{A}]k}$ is unitary:
\begin{equation}	\label{Eq018}
	{W^{[\mathbb{A}]k}}{{\bar{W}}^{[\mathbb{A}]k}} = {I^{[\mathbb{A}]k}} = {{\bar{W}}^{[\mathbb{A}]k}}{W^{[\mathbb{A}]k}}
\end{equation}

\begin{figure}[htb]\center
	\includegraphics{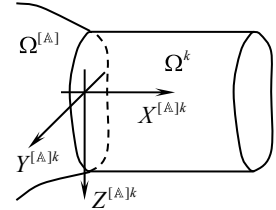}
	\caption{Local coordinate frame $[XYZ]^{[\mathbb{A}]k}$ at boundary ${\Gamma}^{[\mathbb{A}]k}$ of junction ${\Omega}^{[\mathbb{A}]}$ with branch ${\Omega}^k$.}\label{fig:2}
\end{figure}

Since branches are quantum wires, they have a simple geometry:
\begin{equation}	\label{Eq019}
	{\omega}^{[\mathbb{A}]k} := \{{\bf r}\in {\mathbb{R}}^3 | w^{[\mathbb{A}]k} {\bf r} \in {\Omega}^k \} = {\alpha}^{[\mathbb{A}]k} \times {\beta}^{[\mathbb{A}]k}
\end{equation}
where ${\beta}^{[\mathbb{A}]k}$ is the cross section of the branch ${\Omega}^k$ in LCF $[XYZ]^{[\mathbb{A}]k}$. The coordinate origins of all LCFs are at the starts of the branches:
\begin{equation}	\label{Eq020}
	{{\alpha}^{[\mathbb{A}]k}} = (0,{a^k})
\end{equation}
where $a^k$ is the length of the branch ${\Omega}^k$, and boundaries of the junction with branches have the form:
\begin{equation}	\label{Eq021}
	{\gamma}^{[\mathbb{A}]k} := \{ {\bf r} \in {\mathbb{R}}^3 | w^{[\mathbb{A}]k} {\bf r} \in {\Gamma}^k \} = \{0\} \times {\beta}^{[\mathbb{A}]k}
\end{equation}

The motion of the charge carrier in the domain ${\Omega}^k$ in GCF is described by equation (\ref{Eq015}). To separate variables in problem (\ref{Eq015}), we change GCF to LCF:
\begin{equation}	\label{Eq022}
	\begin{aligned}
		{\psi}^{[\mathbb{A}]k}({\bf r}) &:= W^{[\mathbb{A}]k}{{\Psi}^k}({\bf r}) = {\Psi}^k({w^{[\mathbb{A}]k}}{\bf r}) \\
		u^{[\mathbb{A}]k}({\bf r}) &:= W^{[\mathbb{A}]k}{{\upsilon}^k}({\bf r}) \ = {\upsilon}^k({w^{[\mathbb{A}]k}}{\bf r})
	\end{aligned}
\end{equation}
and assume that the potential ${\upsilon}^k$ does not change along the branch. Since the branches have a simple geometry (\ref{Eq019}), we obtain the following analogue of problem (\ref{Eq015}) in LCF:
\begin{equation}	\label{Eq023}
	\left\{ \begin{aligned}
		&[ -\Delta + u^{[\mathbb{A}]k}( y,z) ]{\psi}^{[\mathbb{A}]k}&&\!\!\!\!\!\!(x,y,z) = \varepsilon {\psi}^{[\mathbb{A}]k}(x,y,z),\\
		&&&\!\!\!\!\!\!\{x,y,z\} \in {{\alpha}^{[\mathbb{A}]k}}\times {{\beta}^{[\mathbb{A}]k}} \\
		&{\psi}^{[\mathbb{A}]k}(x,y,z) = 0,  &&\!\!\!\!\!\!\{x,y,z\} \in {{\alpha}^{[\mathbb{A}]k}}\times \partial {{\beta}^{[\mathbb{A}]k}}
	\end{aligned} \right.
\end{equation}

We are looking for a solution of equation (\ref{Eq023}) in the form:
\begin{equation}	\label{Eq024}
	{\psi}^{[\mathbb{A}]k}(x,y,z) =: g^{[\mathbb{A}]k} (x)h^{[\mathbb{A}]k}(y,z)
\end{equation}
Substituting expression (\ref{Eq024}) into problem (\ref{Eq023}), we get
\begin{equation}	\label{Eq025}
	\left\{ \begin{aligned}
		&[-\partial _y^2-\partial _z^2 + u^{[\mathbb{A}]k}]h_m^{[\mathbb{A}]k} = \lambda _m^{[\mathbb{A}]k}h_m^{[\mathbb{A}]k} && \operatorname{in}\ \ {{\beta}^{[\mathbb{A}]k}} \\
		&h_m^{[\mathbb{A}]k} = 0 && \operatorname{on}\ \partial {\beta}^{[\mathbb{A}]k}
	\end{aligned} \right.
\end{equation}
\begin{equation}	\label{Eq026}
	-\partial _x^2g_m^{[\mathbb{A}]k} = (\varepsilon -\lambda _m^{[\mathbb{A}]k})g_m^{[\mathbb{A}]k}\quad \operatorname{in}\ {{\alpha}^{[\mathbb{A}]k}}
\end{equation}
From here follows
\begin{equation}	\label{Eq027}
	\begin{aligned}
		{\psi}^{[\mathbb{A}]k}(x,y,z) = &\sum\nolimits_m{c_m^{[\mathbb{A}]\triangleleft k}\exp \left( -i\kappa _m^{[\mathbb{A}]k}x \right)h_m^{[\mathbb{A}]k}(y,z)} \\
		&+ \sum\nolimits_m{c_m^{[\mathbb{A}]\triangleright k}\exp \left( +i\kappa _m^{[\mathbb{A}]k}x \right)h_m^{[\mathbb{A}]k}\left( y,z \right)}
	\end{aligned}
\end{equation}
\begin{equation}	\label{Eq028}
	\kappa _m^{[\mathbb{A}]k} := \sqrt{\varepsilon -\lambda _m^{[\mathbb{A}]k}}
\end{equation}
where $c^{[\mathbb{A}]\triangleleft k}$ are amplitudes of waves incident on the junction ${\Omega}^{[\mathbb{A}]}$ from branch ${\Omega}^k$ in LCF $[XYZ]^{[\mathbb{A}]k}$, $c^{[\mathbb{A}]\triangleright k}$ are amplitudes of waves scattered by the junction ${\Omega}^{[\mathbb{A}]}$ to branch ${\Omega}^k$ in LCF $[XYZ]^{[\mathbb{A}]k}$ ($c^{[\mathbb{A}]\triangleleft k}$ and $c^{[\mathbb{A}]\triangleright k}$ are \textit{wave amplitudes}), ${\lambda}^{[\mathbb{A}]k}$ are \textit{energies of channels} in branch ${\Omega}^k$, $h^{[\mathbb{A}]k}$ are \textit{transverse modes} in branch ${\Omega}^k$ in LCF $[XYZ]^{[\mathbb{A}]k}$. \textit{Wavenumber} $\kappa _m^{[\mathbb{A}]k}$ determines type of wave depending on type of channel for given energy of charge carrier:
\begin{equation}	\label{Eq029}
	\begin{aligned}
		\lambda _m^{[\mathbb{A}]k}<\varepsilon \ \ -\ \ & \textit{open channel} \quad \ \ \Rightarrow \\
		& \operatorname{Im}(\kappa _m^{[\mathbb{A}]k}) = 0\ \ -\ \ \textit{running wave} \\
		\lambda _m^{[\mathbb{A}]k}\ge \varepsilon \ \ -\ \ & \textit{closed channel} \;\;\; \Rightarrow \\
		& \operatorname{Im}(\kappa _m^{[\mathbb{A}]k})\ne 0\ \ -\ \ \textit{evanescent wave}
	\end{aligned}
\end{equation}

Everywhere in this work, we suppose that functions $h^{[\mathbb{A}]}$ are orthonormal:
\begin{equation}	\label{Eq030}
	I_{mn}^{[\mathbb{A}]kk} = \langle h_m^{[\mathbb{A}]k} | h_n^{[\mathbb{A}]k} \rangle,\quad {I^{[\mathbb{A}]kk}} = \sum\nolimits_n{| h_n^{[\mathbb{A}]k} \rangle \langle h_n^{[\mathbb{A}]k} |}
\end{equation}
The expression (\ref{Eq027}) is also conveniently written in terms of the \textit{wavenumber operator} $K^{[\mathbb{A}]}$ in two equivalent representations:
\begin{equation}	\label{Eq031}
	\begin{aligned}
	{{\psi}^{[\mathbb{A}]k}}\left( x,y,z \right) =& \sum\nolimits_m{\left[ \exp \left( -i{K^{[\mathbb{A}]}}x \right){c^{[\mathbb{A}]\triangleleft}} \right]_m^kh_m^{[\mathbb{A}]k}\left( y,z \right)} \\
		&+ \sum\nolimits_m{\left[ \exp \left( +i{K^{[\mathbb{A}]}}x \right){c^{[\mathbb{A}]\triangleright}} \right]_m^kh_m^{[\mathbb{A}]k}\left( y,z \right)}
	\end{aligned}
\end{equation}
\vspace{-20pt}
\begin{equation}	\label{Eq032}
	K_{mn}^{[\mathbb{A}]kl} := I_{mn}^{[\mathbb{A}]kl}\kappa _m^{[\mathbb{A}]k}
\end{equation}
\vspace{-20pt}
\begin{equation}	\label{Eq033}
	\begin{aligned}
	{{\psi}^{[\mathbb{A}]k}}\left( x,y,z \right) =& \sum\nolimits_m{\left[ \exp \left( -i{K^{[\mathbb{A}]}}x \right){h^{[\mathbb{A}]}}\left( y,z \right) \right]_m^kc_m^{[\mathbb{A}]\triangleleft k}} \\
		&+ \sum\nolimits_m{\left[ \exp \left( +i{K^{[\mathbb{A}]}}x \right){h^{[\mathbb{A}]}}\left( y,z \right) \right]_m^kc_m^{[\mathbb{A}]\triangleright k}}
	\end{aligned}
\end{equation}
\vspace{-15pt}
\begin{equation}	\label{Eq034}
	{K^{[\mathbb{A}]kl}} := {I^{[\mathbb{A}]kl}}\sum\nolimits_n{| h_n^{[\mathbb{A}]k} \rangle \kappa _n^{[\mathbb{A}]k}\langle h_n^{[\mathbb{A}]k} |}
\end{equation}
We group the incident and scattered waves in the expression (\ref{Eq027}):
\begin{equation}	\label{Eq035}
	{{\psi}^{[\mathbb{A}]k}} = {{\psi}^{[\mathbb{A}]\triangleleft k}}+{{\psi}^{[\mathbb{A}]\triangleright k}}
\end{equation}
\vspace{-15pt}
\begin{equation}	\label{Eq036}
	{{\psi}^{[\mathbb{A}]\triangleleft k}} := \sum\nolimits_m{\psi _m^{[\mathbb{A}]\triangleleft k}},\quad {{\psi}^{[\mathbb{A}]\triangleright k}} := \sum\nolimits_m{\psi _m^{[\mathbb{A}]\triangleright k}}
\end{equation}
\vspace{-15pt}
\begin{equation}	\label{Eq037}
	\begin{aligned}
		\psi _m^{[\mathbb{A}]\triangleleft k}\left( x,y,z \right) :&\!\!= \left[ \exp \left( -i{K^{[\mathbb{A}]}}x \right){c^{[\mathbb{A}]\triangleleft}} \right]_m^kh_m^{[\mathbb{A}]k}\left( y,z \right) \\
		&\!\!= \left[ \exp \left( -i{K^{[\mathbb{A}]}}x \right){h^{[\mathbb{A}]}}\left( y,z \right) \right]_m^kc_m^{[\mathbb{A}]\triangleleft k} \\
		\psi _m^{[\mathbb{A}]\triangleright k}\left( x,y,z \right) :&\!\!= \left[ \exp \left( +i{K^{[\mathbb{A}]}}x \right){c^{[\mathbb{A}]\triangleright}} \right]_m^kh_m^{[\mathbb{A}]k}\left( y,z \right) \\
		&\!\!= \left[ \exp \left( +i{K^{[\mathbb{A}]}}x \right){h^{[\mathbb{A}]}}\left( y,z \right) \right]_m^kc_m^{[\mathbb{A}]\triangleright k}
	\end{aligned}
\end{equation}

\subsubsection{Scattering problem}	\label{sec:2_2_3}

Amplitudes of the incident and scattered waves in all canals are related with each other by the \textit{extended scattering matrix} $S^{[\mathbb{A}]}$ of junction ${\Omega}^{[\mathbb{A}]}$ \cite[p.~155]{Bib027}:
\begin{equation}	\label{Eq038}
	c_m^{[\mathbb{A}]\triangleright k} = \sum\nolimits_n^l{S_{mn}^{[\mathbb{A}]kl}c_n^{[\mathbb{A}]\triangleleft l}}
\end{equation}
\begin{equation}	\label{Eq039}
	c_m^{[\mathbb{A}]\triangleleft k} = \sum\nolimits_n^l{[(S^{[\mathbb{A}]})^{-1}]_{mn}^{kl}c_n^{[\mathbb{A}]\triangleright l}}
\end{equation}
To interpret expressions, mnemonic rules are convenient:
\begin{equation}	\label{Eq040}
	c_{\text{“from”}}^{[\mathbb{A}]\triangleleft \text{“from”}},\quad c_{\text{“to”}}^{[\mathbb{A}]\triangleright \text{“to”}},\quad S_{\text{“to”“from”}}^{[\mathbb{A}]\text{“to”“from”}}
\end{equation}
Taking into account our notation, everywhere in this work, we call the extended scattering matrix briefly: \textit{S-matrix}.

We can formulate two problems: the \textit{incoming scattering problem} is the search for scattered waves from known incident waves, the \textit{outgoing scattering problem} is the search for incident waves from known scattered waves. By solving the incoming and outgoing scattering problems, one can find $S^{[\mathbb{A}]}$ and $(S^{[\mathbb{A}]})^{-1}$ matrices, respectively. As a rule, only the first problem is considered in the literature. It is also the main one for this work.

\subsection{S-matrix of quantum network junction}	\label{sec:2_3}

\begin{equation}	\label{Eq041}
	\{\ {{\square}^{[\mathbb{A}]}}\ \mapsto \ \square \ |\ \square \ = K,S,...\}\quad \operatorname{in}\ \text{\ref{sec:2_3}}
\end{equation}

\subsubsection{Scattering boundary conditions}	\label{sec:2_3_1}

Let us write a direct method for solving the scattering problems (section \ref{sec:2_2_3}). To do this, one should obtain scattering boundary conditions \cite{Bib019,Bib020,Bib021} for an arbitrary junction of a quantum network (Fig.~\ref{fig:1}).

\begin{prop}	\label{prop:1}
 The incoming scattering problem is a problem (\ref{Eq016}) supplemented by \textit{incoming scattering boundary conditions} (\textit{incoming SBC}):
\begin{equation}	\label{Eq042}
	[K+i{{\partial}_1}]W\Psi = 2K{{\psi}^{\triangleleft}}\quad \operatorname{on}\ \gamma 
\end{equation}
The outgoing scattering problem is a problem (\ref{Eq016}) supplemented by \textit{outgoing scattering boundary conditions} (\textit{outgoing SBC}):
\begin{equation}	\label{Eq043}
	[K-i{{\partial}_1}]W\Psi = 2K{{\psi}^{\triangleright}}\quad \operatorname{on}\ \gamma 
\end{equation}
\end{prop}
\begin{proof}[\textsc{Proof~\ref{prop:1}}]
 We rewrite the matching conditions (\ref{Eq017}) in LCF taking into account the form (\ref{Eq021}):
\begin{equation}	\label{Eq044}
	\left\{ \begin{aligned}
		&W\Psi = \psi \\
		&{{\partial}_1}W\Psi = {{\partial}_1}\psi
	\end{aligned} \right.\quad \operatorname{on}\ \gamma 
\end{equation}
Here, the branch index is reduced, and all expressions with vector and matrix objects ($W$, $\psi $, $\omega $, $\gamma $, $K$ etc.) are vector (section \ref{sec:2_1_1}). Using the expression (\ref{Eq035}), we get a derivative of the function in the branches
\begin{equation}	\label{Eq045}
	{{\partial}_1}\psi = -iK{{\psi}^{\triangleleft}}+iK{{\psi}^{\triangleright}}\quad \operatorname{in}\ \omega 
\end{equation}
From the formulas (\ref{Eq031}), (\ref{Eq044}) and (\ref{Eq045}) follows
\begin{equation}	\label{Eq046}
	\left\{ \begin{aligned}
		&W\Psi = {\psi}^{\triangleleft} + {\psi}^{\triangleright} \\
		&{{\partial}_1}W\Psi = -iK{\psi}^{\triangleleft} +iK{\psi}^{\triangleright}
	\end{aligned} \right.\quad \operatorname{on}\ \gamma 
\end{equation}

Excluding from expressions (\ref{Eq046}) functions ${\psi}^{\triangleright}$, we have incoming SBC (\ref{Eq042}). Excluding from expressions (\ref{Eq046}) functions ${\psi}^{\triangleleft}$, we have outgoing SBC (\ref{Eq043}). Solving the problem (\ref{Eq016}) with incoming SBC (\ref{Eq042}), it is possible to find scattered waves ${\psi}^{\triangleright}$ according to the conditions (\ref{Eq044}) by known incident ones ${\psi}^{\triangleleft}$:
\begin{equation}	\label{Eq047}
	{{\psi}^{\triangleright}} = W\Psi -{{\psi}^{\triangleleft}}\quad \operatorname{on}\ \gamma 
\end{equation}
and hence the matrix $S$ (\ref{Eq038}). Solving the problem (\ref{Eq016}) with outgoing SBC (\ref{Eq043}), it is possible to find incident waves ${\psi}^{\triangleleft}$ according to the conditions (\ref{Eq044}) by known scattered ones ${\psi}^{\triangleright}$:
\begin{equation}	\label{Eq048}
	{{\psi}^{\triangleleft}} = W\Psi -{{\psi}^{\triangleright}}\quad \operatorname{on}\ \gamma 
\end{equation}
and hence the matrix $S^{-1}$ (\ref{Eq039}).	
\end{proof}

SBC is also convenient to write in GCF. For incoming SBC (\ref{Eq042}) and outgoing SBC (\ref{Eq043}) we have:
\begin{equation}	\label{Eq049}
	[K+i{{\partial}_n}]\Psi = 2K{{\Psi}^{\triangleleft}}\quad \operatorname{on}\ \Gamma 
\end{equation}
\begin{equation}	\label{Eq050}
	[K-i{{\partial}_n}]\Psi = 2K{{\Psi}^{\triangleright}}\quad \operatorname{on}\ \Gamma 
\end{equation}
respectively, where the operator $K$ acts on functions in the GCF as follows:
\begin{equation}	\label{Eq051}
	{\{{{[K\square ]}^k} := {{\bar{W}}^k}\sum\nolimits_{}^l{{K^{kl}}{W^l}\square}|\ \square \ = \Psi,{{\Psi}^{\triangleleft}},{{\Psi}^{\triangleright}}\}}^k
\end{equation}
Here ${{\Psi}^{\triangleleft}} := \bar{W}{{\psi}^{\triangleleft}}$, ${{\Psi}^{\triangleright}} := \bar{W}{{\psi}^{\triangleright}}$, ${\partial}_n$ is a derivative along the outer normal to the boundary. The expressions (\ref{Eq047}) and (\ref{Eq048}) take the form:
\begin{equation}	\label{Eq052}
	{{\Psi}^{\triangleright}} = \Psi -{{\Psi}^{\triangleleft}}\quad \operatorname{on}\ \Gamma 
\end{equation}
\begin{equation}	\label{Eq053}
	{{\Psi}^{\triangleleft}} = \Psi -{{\Psi}^{\triangleright}}\quad \operatorname{on}\ \Gamma 
\end{equation}

The boundary conditions (\ref{Eq042}), (\ref{Eq043}), (\ref{Eq049}) and (\ref{Eq050}) are integro-differential because they contain both the differentiation operators ${\partial}_1$, ${\partial}_n$ and the integral operator $K$ (\ref{Eq034}). At the same time, in terms of SBC, the physical meaning of the operator $K$ is \textit{isolation of junction} (Appendix \ref{sec:A}). Therefore, these boundary conditions are clear: they provide an intuitive and concise formulation of the scattering problem at the junction of the quantum network.

\subsubsection{Calculation of S-matrix by SBC method}	\label{sec:2_3_2}

In all scattering problems, the S-matrix of a junction can be found element by element using SBC. At the same time, it is convenient to write SBC in GCF for implementation in computing packages.

\begin{prop}	\label{prop:2}
 Any element of the S-matrix of a junction $\Omega $ is written as
\begin{equation}	\label{Eq054}
	S_{mq}^{kp} = \langle h_m^k | {W^k}\Psi _q^p(0,...) \rangle -I_{mq}^{kp}
\end{equation}
where the function $\Psi _q^p$ is the solution to the scattering problem
\begin{equation}	\label{Eq055}
	\left\{ \begin{aligned}
		&[-\Delta +\upsilon ]\Psi _q^p = \varepsilon \Psi _q^p && \operatorname{in}\ \ \Omega \\
		&\Psi _q^p = 0 && \operatorname{on}\ \partial \Omega \backslash \Gamma \\
		&[K+i{{\partial}_n}]\Psi _q^p = 2K\bar{W}h_q^p && \operatorname{on}\ \Gamma
	\end{aligned} \right.
\end{equation}
\end{prop}
\begin{proof}[\textsc{Proof~\ref{prop:2}}]
 The S-matrix of the junction is independent of the of incident waves. We set their amplitudes as $\{c_m^{\triangleleft k} = I_{mq}^{kp}\}_m^k$ by fixing the number of the branch-channel $_q^p$. Of (\ref{Eq016}), (\ref{Eq049}), (\ref{Eq036}), (\ref{Eq037}), we have (\ref{Eq055}). Based on function $\Psi _q^p$ from (\ref{Eq047}), (\ref{Eq036})--(\ref{Eq038}) we obtain (\ref{Eq054}).	
\end{proof}

Obviously, to calculate one column of the S-matrix $_q^p$, it is enough to solve only one boundary problem of the form (\ref{Eq055}). This is an advantage of SBC in the numerical study of individual S-matrix elements.

In some scattering problems, using SBC one can write the S-matrix of a junction explicitly.

\begin{prop}	\label{prop:3}
 The S-matrix of the junction $\Omega $ is written as
\begin{equation}	\label{Eq056}
	S = {G^{\Diamond}}{{\left[ iK{G^{\Diamond}}-{{{\dot{G}}}^{\Diamond}} \right]}^{-1}}\left( 0 \right)i2K-I
\end{equation}
in the scattering problem
\begin{equation}	\label{Eq057}
	\left\{ \begin{aligned}
		&[-\Delta +\upsilon ]\Psi = \varepsilon \Psi && \operatorname{in}\ \ \Omega \\
		&\Psi = 0 && \operatorname{on}\ \partial \Omega \backslash \Gamma \\
		&[K+i{{\partial}_1}]W\Psi = 2K{{\psi}^{\triangleleft}} && \operatorname{on}\ \gamma
	\end{aligned} \right.
\end{equation}
where one can define an operator $G^{\Diamond}$ in one of two equivalent representations (Appendix \ref{sec:B}):
\begin{equation}	\label{Eq058}
	{W^k}\Psi (x,y,z) = :\sum\nolimits_m[ G^{\Diamond}(x)c^{\Diamond}]_m^kh_m^k(y,z)
\end{equation}
\begin{equation}	\label{Eq059}
	{W^k}\Psi (x,y,z) = :\sum\nolimits_m{c_m^{\Diamond k}[G^{\Diamond}(x)h(y,z)]_m^k}
\end{equation}
\end{prop}
\begin{proof}[\textsc{Proof~\ref{prop:3}}]
 Based on (\ref{Eq058}), (\ref{Eq036}) and (\ref{Eq037}) for the problem (\ref{Eq057}), we obtain
\begin{equation}	\label{Eq060}
	{c^{\Diamond}} = {{\left[ iK{G^{\Diamond}}-{{{\dot{G}}}^{\Diamond}} \right]}^{-1}}\left( 0 \right)i2K{c^{\triangleleft}}
\end{equation}
From (\ref{Eq047}), (\ref{Eq058}), (\ref{Eq060}), (\ref{Eq036}), (\ref{Eq037}) and (\ref{Eq038}) we have (\ref{Eq056}).	
\end{proof}

\subsubsection{Calculation of S-matrix by DN- and ND-map methods}	\label{sec:2_3_3}

The SBC method is the most clear when calculating the S-matrix of a quantum network junction, because provides the wave function in a junction during scattering. However, SBC are integro-differential boundary conditions that are not supported by all compute packages. In the framework of the DN- and ND-map \cite{Bib015,Bib016,Bib022,Bib023} method, it is enough to solve the Dirichlet and Neumann problems, respectively. It is acceptable in cases where the wave function in the junction during scattering is not essential. We specify the DN- and ND-map methods for the notation system used.

\begin{prop}	\label{prop:4}
 The S-matrix of the junction $\Omega $ is written using the \textit{DN-map} operator $D$:
\begin{equation}	\label{Eq061}
	S = [iK-D]^{-1}[iK+D]
\end{equation}
\begin{equation}	\label{Eq062}
	D\chi := {{\partial}_1}W\Psi \quad \operatorname{on}\ \gamma,\quad
	\left\{ \begin{aligned}
		&[-\Delta +\upsilon ]\Psi = \varepsilon \Psi && \operatorname{in}\ \ \Omega \\
		&\Psi = 0 && \operatorname{on}\ \partial \Omega \backslash \Gamma \\
		&W\Psi = \chi && \operatorname{on}\ \gamma
	\end{aligned} \right.
\end{equation}
\end{prop}
\begin{proof}[\textsc{Proof~\ref{prop:4}}]
 See Appendix \ref{sec:C}.	
\end{proof}

\begin{prop}	\label{prop:5}
 The S-matrix of the junction $\Omega $ is written using the \textit{ND-map} operator $N$:
\begin{equation}	\label{Eq063}
	S = [NiK-I]^{-1}[NiK+I]
\end{equation}
\begin{equation}	\label{Eq064}
	N\dot{\chi} := W\Psi \quad \operatorname{on}\ \gamma,\quad
	\left\{ \begin{aligned}
		&[-\Delta +\upsilon ]\Psi = \varepsilon \Psi && \operatorname{in}\ \ \Omega \\
		&\Psi = 0 && \operatorname{on}\ \partial \Omega \backslash \Gamma \\
		&{{\partial}_1}W\Psi = \dot{\chi} && \operatorname{on}\ \gamma
	\end{aligned} \right.
\end{equation}
\end{prop}
\begin{proof}[\textsc{Proof~\ref{prop:5}}]
 See Appendix \ref{sec:C}.	
\end{proof}

Expressions (\ref{Eq061}) and (\ref{Eq063}) are equivalent taking into account the property: $D = {N^{-1}}$. By analogy with (\ref{Eq055}) the matrices $D$ and $N$ can be found element by element from (\ref{Eq222}), (\ref{Eq224}) and (\ref{Eq030}):
\begin{equation}	\label{Eq065}
	\begin{aligned}
		D_{mq}^{kp} &= \langle h_m^k | {{\partial}_1}{W^k}\Psi _q^p(0,...) \rangle,	\\
		&\left\{ \begin{aligned}
			&[-\Delta +\upsilon ]\Psi _q^p = \varepsilon \Psi _q^p && \operatorname{in}\ \ \Omega \\
			&\Psi _q^p = 0 && \operatorname{on}\ \partial \Omega \backslash \Gamma \\
			&\Psi _q^p = \bar{W}h_q^p && \operatorname{on}\ \Gamma
		\end{aligned} \right.
	\end{aligned}
\end{equation}
\vspace{15pt}
\begin{equation}	\label{Eq066}
	\begin{aligned}
		N_{mq}^{kp} &= \langle h_m^k | {W^k}\Psi _q^p(0,...) \rangle,	\\
		&\left\{ \begin{aligned}
			&[-\Delta +\upsilon ]\Psi _q^p = \varepsilon \Psi _q^p && \operatorname{in}\ \ \Omega \\
			&\Psi _q^p = 0 && \operatorname{on}\ \partial \Omega \backslash \Gamma \\
			&{{\partial}_n}\Psi _q^p = \bar{W}h_q^p && \operatorname{on}\ \Gamma
		\end{aligned} \right.
	\end{aligned}
\end{equation}

To find any element of the S-matrix, according to the expressions (\ref{Eq061}) and (\ref{Eq063}), one need to find all elements of the $D$ or $N$ operator. DN- and ND-map are singular if $\varepsilon $ is the eigenvalue of the problem with zero Dirichlet and Neumann boundary conditions, respectively.

\subsection{S-matrix of quantum network in terms of its junctions' S-matrices}	\label{sec:2_4}

Based on the agreements and notation introduced in section \ref{sec:2_1}, we will obtain an expression for the S-matrix of an arbitrary quantum network in terms of its junctions' S-matrices.

\subsubsection{Elementary section of network}	\label{sec:2_4_1}

Let us consider the scattering of the charge carrier in a quantum network. Since the network consists of junctions and branches connecting them, its \textit{elementary section} is the connection of two junctions (Fig.~\ref{fig:3}). Any quantum network can be represented as a set of such connections. Therefore, the expressions for the elementary section are the basis for calculating the S-matrix of the entire network.

\begin{figure}[htb]\center
	\includegraphics{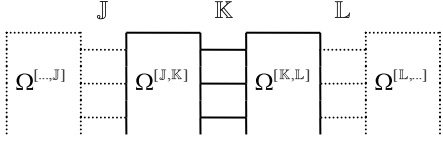}
	\caption{Scheme of elementary section of network ${\Omega}^{[\mathbb{J},\mathbb{L}]}$: solid lines are internal junctions (${\Omega}^{[\mathbb{J},\mathbb{K}]}$, ${\Omega}^{[\mathbb{K},\mathbb{L}]}$) and branches ($\mathbb{K}$ is tuple of their numbers), dotted lines are external junctions (${\Omega}^{[...,\mathbb{J}]}$, ${\Omega}^{[\mathbb{L},...]}$) and branches ($\mathbb{J}$, $\mathbb{L}$ are tuples of their numbers).}\label{fig:3}
\end{figure}

Elementary section of network (Fig.~\ref{fig:3}) is a new combined junction ${{\Omega}^{[\mathbb{J},\mathbb{L}]}} = {{\Omega}^{[\mathbb{J},\mathbb{K}]}}+\sum\nolimits_{}^{k\in \mathbb{K}}{({{\Omega}^k}+{{\Gamma}^k})}+{{\Omega}^{[\mathbb{K},\mathbb{L}]}}$. Its extended scattering matrix $S^{[\mathbb{J},\mathbb{L}]}$ relates the amplitudes of the waves incident on it $c^{[\mathbb{J},\mathbb{L}]\triangleleft}$ with the amplitudes of the waves scattered by it $c^{[\mathbb{J},\mathbb{L}]\triangleright}$:
\begin{equation}	\label{Eq067}
	\begin{aligned}
		{c^{[\mathbb{J},\mathbb{L}]\triangleright}} &=
		\left[ \begin{matrix}
			{c^{[\mathbb{J},\mathbb{L}]\triangleright \mathbb{J}}} \\
			{c^{[\mathbb{J},\mathbb{L}]\triangleright \mathbb{L}}} \\
		\end{matrix} \right]	\\
		&=\left[ \begin{matrix}
			{S^{[\mathbb{J},\mathbb{L}]\mathbb{J}\mathbb{J}}} & {S^{[\mathbb{J},\mathbb{L}]\mathbb{J}\mathbb{L}}} \\
			{S^{[\mathbb{J},\mathbb{L}]\mathbb{L}\mathbb{J}}} & {S^{[\mathbb{J},\mathbb{L}]\mathbb{L}\mathbb{L}}} \\
		\end{matrix} \right]
		\left[ \begin{matrix}
		{c^{[\mathbb{J},\mathbb{L}]\triangleleft \mathbb{J}}} \\
			{c^{[\mathbb{J},\mathbb{L}]\triangleleft \mathbb{L}}} \\
		\end{matrix} \right] = {S^{[\mathbb{J},\mathbb{L}]}}{c^{[\mathbb{J},\mathbb{L}]\triangleleft}}
	\end{aligned}
\end{equation}
where $c^{[\mathbb{J},\mathbb{L}]\triangleleft \mathbb{J}}$ are amplitudes of waves incident on the junction ${\Omega}^{[\mathbb{J},\mathbb{L}]}$ from external branches ${\Omega}^{\mathbb{J}}$, $c^{[\mathbb{J},\mathbb{L}]\triangleright \mathbb{J}}$ are amplitudes of waves scattered by the junction ${\Omega}^{[\mathbb{J},\mathbb{L}]}$ to external branches ${\Omega}^{\mathbb{J}}$. Amplitudes $c^{[\mathbb{J},\mathbb{L}]\triangleleft \mathbb{L}}$ and $c^{[\mathbb{J},\mathbb{L}]\triangleright \mathbb{L}}$ also have a similar meaning to the accuracy of replacing indices outside square brackets $\mathbb{J}\leftrightarrow \mathbb{L}$. Matrix $S^{[\mathbb{J},\mathbb{L}]}$ is separated in expression (\ref{Eq067}) by corresponding submatrices.

From the shown scheme of the elementary section of the network (Fig.~\ref{fig:3}) it follows that
\pagebreak
\begin{equation}	\label{Eq068}
	\left[ \begin{matrix}
	{c^{[\mathbb{J},\mathbb{L}]\triangleright \mathbb{J}}} \\
		{c^{[\mathbb{J},\mathbb{L}]\triangleright \mathbb{L}}} \\
		\end{matrix} \right] = \left[ \begin{matrix}
	{c^{[\mathbb{J},\mathbb{K}]\triangleright \mathbb{J}}} \\
		{c^{[\mathbb{K},\mathbb{L}]\triangleright \mathbb{L}}} \\
		\end{matrix} \right],\quad \left[ \begin{matrix}
	{c^{[\mathbb{J},\mathbb{L}]\triangleleft \mathbb{J}}} \\
		{c^{[\mathbb{J},\mathbb{L}]\triangleleft \mathbb{L}}} \\
		\end{matrix} \right] = \left[ \begin{matrix}
	{c^{[\mathbb{J},\mathbb{K}]\triangleleft \mathbb{J}}} \\
		{c^{[\mathbb{K},\mathbb{L}]\triangleleft \mathbb{L}}} \\
		\end{matrix} \right]
\end{equation}
Therefore, it is possible to write $c^{[\mathbb{J},\mathbb{L}]\triangleright}$ in terms of S-matrices of junctions forming a network section $S^{[\mathbb{J},\mathbb{K}]}$ and $S^{[\mathbb{K},\mathbb{L}]}$:
\begin{equation}	\label{Eq069}
	{c^{[\mathbb{J},\mathbb{L}]\triangleright}} = \left[ \begin{matrix}
	{c^{[\mathbb{J},\mathbb{K}]\triangleright \mathbb{J}}} \\
		{c^{[\mathbb{K},\mathbb{L}]\triangleright \mathbb{L}}} \\
		\end{matrix} \right] = \left[ \begin{matrix}
	\sum\nolimits_{}^l{{S^{[\mathbb{J},\mathbb{K}]\mathbb{J}l}}{c^{[\mathbb{J},\mathbb{K}]\triangleleft l}}} \\
		\sum\nolimits_{}^l{{S^{[\mathbb{K},\mathbb{L}]\mathbb{L}l}}{c^{[\mathbb{K},\mathbb{L}]\triangleleft l}}} \\
		\end{matrix} \right]
\end{equation}
In expression (\ref{Eq069}), let us separate the first sum as $\sum\nolimits^l{} = \sum\nolimits^{l \in \mathbb{J}}{} + \sum\nolimits^{l \in \mathbb{K}}{}$ and the second one as $\sum\nolimits^l{} = \sum\nolimits^{l \in \mathbb{K}}{} + \sum\nolimits^{l \in \mathbb{L}}{}$. Taking into account this separation and the agreements of section \ref{sec:2_1_1}, from expression (\ref{Eq069}) follows
\begin{equation}	\label{Eq070}
	\begin{aligned}
		\left[ \begin{matrix}
			{c^{[\mathbb{J},\mathbb{K}]\triangleright \mathbb{J}}} \\
			{c^{[\mathbb{K},\mathbb{L}]\triangleright \mathbb{L}}} \\
		\end{matrix} \right] 
		=& \left[ \begin{matrix}
		{S^{[\mathbb{J},\mathbb{K}]\mathbb{J}\mathbb{J}}} & {O^{\mathbb{J}\mathbb{L}}} \\
			{O^{\mathbb{L}\mathbb{J}}} & {S^{[\mathbb{K},\mathbb{L}]\mathbb{L}\mathbb{L}}} \\
			\end{matrix} \right]\left[ \begin{matrix}
		{c^{[\mathbb{J},\mathbb{K}]\triangleleft \mathbb{J}}} \\
			{c^{[\mathbb{K},\mathbb{L}]\triangleleft \mathbb{L}}} \\
		\end{matrix} \right]	\\
		&+\left[ \begin{matrix}
		{S^{[\mathbb{J},\mathbb{K}]\mathbb{J}\mathbb{K}}} & {O^{\mathbb{J}\mathbb{K}}} \\
			{O^{\mathbb{L}\mathbb{K}}} & {S^{[\mathbb{K},\mathbb{L}]\mathbb{L}\mathbb{K}}} \\
			\end{matrix} \right]\left[ \begin{matrix}
		{c^{[\mathbb{J},\mathbb{K}]\triangleleft \mathbb{K}}} \\
			{c^{[\mathbb{K},\mathbb{L}]\triangleleft \mathbb{K}}} \\
		\end{matrix} \right]
	\end{aligned}
\end{equation}

Thus, to calculate the matrix $S^{[\mathbb{J},\mathbb{L}]}$, it is necessary to express the amplitudes of waves incident from the internal branches $c^{[\mathbb{J},\mathbb{K}]\triangleleft \mathbb{K}}$ and $c^{[\mathbb{K},\mathbb{L}]\triangleleft \mathbb{K}}$ in terms of the amplitudes of waves incident from the external branches $c^{[\mathbb{J},\mathbb{K}]\triangleleft \mathbb{J}}$ and $c^{[\mathbb{K},\mathbb{L}]\triangleleft \mathbb{L}}$.

\subsubsection{Combining formula}	\label{sec:2_4_2}

\begin{prop}	\label{prop:6}
 The extended scattering matrix of the elementary section of the network $S^{[\mathbb{J},\mathbb{L}]}$ (Fig.~\ref{fig:3}) is written in terms of extended scattering matrices of its internal junctions $S^{[\mathbb{J},\mathbb{K}]}$ and $S^{[\mathbb{K},\mathbb{L}]}$ using a \textit{combining formula}:
\begin{strip}
\vspace{-20pt}
\begin{equation}	\label{Eq071}
	\begin{aligned}
	{S^{[\mathbb{J},\mathbb{L}]}} =&
	\left[ \begin{matrix}
		{S^{[\mathbb{J},\mathbb{K}]\mathbb{J}\mathbb{J}}} & {O^{\mathbb{J}\mathbb{L}}} \\
		{O^{\mathbb{L}\mathbb{J}}} & {S^{[\mathbb{K},\mathbb{L}]\mathbb{L}\mathbb{L}}} \\
	\end{matrix} \right]+
	\left[ \begin{matrix}
		{S^{[\mathbb{J},\mathbb{K}]\mathbb{J}\mathbb{K}}} & {O^{\mathbb{J}\mathbb{K}}} \\
		{O^{\mathbb{L}\mathbb{K}}} & {S^{[\mathbb{K},\mathbb{L}]\mathbb{L}\mathbb{K}}} \\
	\end{matrix} \right]\\
	&\times {{\left[ \begin{matrix}
		-{S^{[\mathbb{J},\mathbb{K}]\mathbb{K}\mathbb{K}}} & {U^{[\mathbb{J},\mathbb{K}]\mathbb{K}\mathbb{K}}}\exp \left( -i{K^{\mathbb{K}\mathbb{K}}}{A^{\mathbb{K}\mathbb{K}}} \right) \\
		{U^{[\mathbb{K},\mathbb{L}]\mathbb{K}\mathbb{K}}}\exp \left( -i{K^{\mathbb{K}\mathbb{K}}}{A^{\mathbb{K}\mathbb{K}}} \right) & -{S^{[\mathbb{K},\mathbb{L}]\mathbb{K}\mathbb{K}}} \\
		\end{matrix} \right]}^{-1}}\left[ \begin{matrix}
	{S^{[\mathbb{J},\mathbb{K}]\mathbb{K}\mathbb{J}}} & {O^{\mathbb{K}\mathbb{L}}} \\
		{O^{\mathbb{K}\mathbb{J}}} & {S^{[\mathbb{K},\mathbb{L}]\mathbb{K}\mathbb{L}}} \\
	\end{matrix} \right]
	\end{aligned}
\end{equation}
\vspace{-20pt}
\end{strip}
\begin{equation}	\label{Eq072}
	\begin{aligned}
	U_{mn}^{[\mathbb{J},\mathbb{K}]kl} := \langle {{{\bar{W}}}^{[\mathbb{J},\mathbb{K}]k}}h_m^{[\mathbb{J},\mathbb{K}]k} | {{{\bar{W}}}^{[\mathbb{K},\mathbb{L}]l}}h_n^{[\mathbb{K},\mathbb{L}]l} \rangle \\
		U_{mn}^{[\mathbb{K},\mathbb{L}]kl} := \langle {{{\bar{W}}}^{[\mathbb{K},\mathbb{L}]k}}h_m^{[\mathbb{K},\mathbb{L}]k} | {{{\bar{W}}}^{[\mathbb{J},\mathbb{K}]l}}h_n^{[\mathbb{J},\mathbb{K}]l} \rangle
	\end{aligned}
\end{equation}
\begin{equation}	\label{Eq073}
	K_{mn}^{kl} := \kappa _m^kI_{mn}^{kl},\quad A_{mn}^{kl} := {a^k}I_{mn}^{kl}
\end{equation}
\end{prop}
\begin{proof}[\textsc{Proof~\ref{prop:6}}]
 Let us consider the internal branch ${\Omega}^k$ of the elementary section of the network: $k\in \mathbb{K}$. From the definitions (\ref{Eq022}) taking into account property (\ref{Eq018}) follows
\begin{equation}	\label{Eq074}
	{{\bar{W}}^{[\mathbb{J},\mathbb{K}]k}}{{\psi}^{[\mathbb{J},\mathbb{K}]k}} = {{\Psi}^k} = {{\bar{W}}^{[\mathbb{K},\mathbb{L}]k}}{{\psi}^{[\mathbb{K},\mathbb{L}]k}}
\end{equation}
According to expressions (\ref{Eq027}) and (\ref{Eq038}), wave function ${\Psi}^k$ in LCFs $[XYZ]^{[\mathbb{J},\mathbb{K}]k}$ and $[XYZ]^{[\mathbb{K},\mathbb{L}]k}$ takes the form
\begin{strip}
\begin{equation}	\label{Eq075}
	\begin{aligned}
		{{\psi}^{[\mathbb{J},\mathbb{K}]k}}\left( x,y,z \right) = \sum\nolimits_m{c_m^{[\mathbb{J},\mathbb{K}]\triangleleft k}\exp \left( -i\kappa _m^{[\mathbb{J},\mathbb{K}]k}x \right)h_m^{[\mathbb{J},\mathbb{K}]k}\left( y,z \right)}
		+\sum\nolimits_m{\left[ \sum\nolimits_n^l{S_{mn}^{[\mathbb{J},\mathbb{K}]kl}c_n^{[\mathbb{J},\mathbb{K}]\triangleleft l}} \right]\exp \left( +i\kappa _m^{[\mathbb{J},\mathbb{K}]k}x \right)h_m^{[\mathbb{J},\mathbb{K}]k}\left( y,z \right)}
	\end{aligned}
\end{equation}
\begin{equation}	\label{Eq076}
	\begin{aligned}
	{{\psi}^{[\mathbb{K},\mathbb{L}]k}}\left( x,y,z \right) = \sum\nolimits_m{c_m^{[\mathbb{K},\mathbb{L}]\triangleleft k}\exp \left( -i\kappa _m^{[\mathbb{K},\mathbb{L}]k}x \right)h_m^{[\mathbb{K},\mathbb{L}]}\left( y,z \right)}
		+\sum\nolimits_m{\left[ \sum\nolimits_n^l{S_{mn}^{[\mathbb{K},\mathbb{L}]kl}c_n^{[\mathbb{K},\mathbb{L}]\triangleleft l}} \right]\exp \left( +i\kappa _m^{[\mathbb{K},\mathbb{L}]k}x \right)h_m^{[\mathbb{K},\mathbb{L}]k}\left( y,z \right)}
	\end{aligned}
\end{equation}
\end{strip}
respectively. Let us find unknown amplitudes $c^{[\mathbb{J},\mathbb{K}]\triangleleft k}$ and $c^{[\mathbb{K},\mathbb{L}]\triangleleft k}$. For this purpose, we rewrite the identity (\ref{Eq074}) in LCF ${[XYZ]}^{[\mathbb{J},\mathbb{K}]k}$:
\begin{equation}	\label{Eq077}
	{{\psi}^{[\mathbb{J},\mathbb{K}]k}} = {W^{[\mathbb{J},\mathbb{K}]k}}{{\bar{W}}^{[\mathbb{K},\mathbb{L}]k}}{{\psi}^{[\mathbb{K},\mathbb{L}]k}}
\end{equation}
Since LCFs are related to each other by translation and rotation transformations so that ${X^{[\mathbb{J},\mathbb{K}]k}}\parallel {X^{[\mathbb{K},\mathbb{L}]k}}$, according to condition (\ref{Eq020}), we have
\begin{equation}	\label{Eq078}
	\begin{aligned}
		{W^{[\mathbb{J},\mathbb{K}]k}}&{{\bar{W}}^{[\mathbb{K},\mathbb{L}]k}}\exp \left( \pm i\kappa _m^{[\mathbb{K},\mathbb{L}]k}x \right) =	\\
		& \exp \left( \pm i\kappa _m^{[\mathbb{K},\mathbb{L}]k}\left[ {a^k}-x \right] \right){W^{[\mathbb{J},\mathbb{K}]k}}{{\bar{W}}^{[\mathbb{K},\mathbb{L}]k}}
	\end{aligned}
\end{equation}
At the same time, one can shows that
\begin{equation}	\label{Eq079}
	\begin{aligned}
		&{W^{[\mathbb{J},\mathbb{K}]k}}{{\bar{W}}^{[\mathbb{K},\mathbb{L}]k}}h_m^{[\mathbb{K},\mathbb{L}]k} = \sum\nolimits_n{h_n^{[\mathbb{J},\mathbb{K}]k}U_{nm}^{[\mathbb{J},\mathbb{K}]kk}},	\\
		&\lambda _m^{[\mathbb{J},\mathbb{K}]k} = \lambda _m^{[\mathbb{K},\mathbb{L}]k} = :\lambda _m^k
	\end{aligned}
\end{equation}
where $U^{[\mathbb{J},\mathbb{K}]}$ and $U^{[\mathbb{K},\mathbb{L}]}$ are the operators providing changing coordinate frame (\ref{Eq072}). For them from (\ref{Eq072}) we have
\begin{equation}	\label{Eq080}
	{{\bar{U}}^{[\mathbb{J},\mathbb{K}]}} = {U^{[\mathbb{K},\mathbb{L}]}}
\end{equation}
Since $\kappa _m^{[\mathbb{J},\mathbb{K}]k} = \kappa _m^{[\mathbb{K},\mathbb{L}]k} = :\kappa _m^k$, substituting expressions (\ref{Eq075}) and (\ref{Eq076}) to the identity (\ref{Eq077}) taking into account expressions (\ref{Eq078}) and (\ref{Eq079}), we get
\begin{equation}	\label{Eq081}
	\left\{ \begin{aligned}
		&c_m^{[\mathbb{J},\mathbb{K}]\triangleleft k} = U_{mm}^{[\mathbb{J},\mathbb{K}]kk}\exp \left( +i\kappa _m^k{a^k} \right)\sum\nolimits_n^l{S_{mn}^{[\mathbb{K},\mathbb{L}]kl}c_n^{[\mathbb{K},\mathbb{L}]\triangleleft l}} \\
		&c_m^{[\mathbb{K},\mathbb{L}]\triangleleft k} = \bar{U}_{mm}^{[\mathbb{J},\mathbb{K}]kk}\exp \left( +i\kappa _m^k{a^k} \right)\sum\nolimits_n^l{S_{mn}^{[\mathbb{J},\mathbb{K}]kl}c_n^{[\mathbb{J},\mathbb{K}]\triangleleft l}}
	\end{aligned} \right.
\end{equation}

Taking into account property (\ref{Eq080}) and agreements in section \ref{sec:2_1_1}, equations (\ref{Eq081}) take the form
\begin{equation}	\label{Eq082}
	\left\{ \begin{aligned}
		&c^{[\mathbb{J},\mathbb{K}]\triangleleft \mathbb{K}} = U^{[\mathbb{J},\mathbb{K}]\mathbb{K}\mathbb{K}} \exp \! \left( +iK^{\mathbb{K}\mathbb{K}}A^{\mathbb{K}\mathbb{K}} \right) \! \sum\nolimits^l \!\!{S^{[\mathbb{K},\mathbb{L}]\mathbb{K}l}c^{[\mathbb{K},\mathbb{L}]\triangleleft l}} \\
		&c^{[\mathbb{K},\mathbb{L}]\triangleleft \mathbb{K}} = U^{[\mathbb{K},\mathbb{L}]\mathbb{K}\mathbb{K}} \exp \! \left( +iK^{\mathbb{K}\mathbb{K}}A^{\mathbb{K}\mathbb{K}} \right) \! \sum\nolimits^l \!\!{S^{[\mathbb{J},\mathbb{K}]\mathbb{K}l}c^{[\mathbb{J},\mathbb{K}]\triangleleft l}}
	\end{aligned} \right.
\end{equation}
where $K^{\mathbb{K}\mathbb{K}}$ and $A^{\mathbb{K}\mathbb{K}}$ are the operators of wavenumbers and branch lengths, respectively (\ref{Eq073}). In the system (\ref{Eq082}), we separate the sums into two parts: in the first equation as $\sum\nolimits^l{} = \sum\nolimits^{l \in \mathbb{K}}{} + \sum\nolimits^{l \in \mathbb{L}}{}$, in the second one as $\sum\nolimits^l{} = \sum\nolimits^{l \in \mathbb{J}}{} + \sum\nolimits^{l \in \mathbb{K}}{}$, therefore
\begin{strip}
\begin{equation}	\label{Eq083}
	\begin{aligned}
		\left[ \begin{matrix}
			{c^{[\mathbb{J},\mathbb{K}]\triangleleft \mathbb{K}}} \\
			{c^{[\mathbb{K},\mathbb{L}]\triangleleft \mathbb{K}}} \\
		\end{matrix} \right] =
		{{\left[ \begin{matrix}
			-{S^{[\mathbb{J},\mathbb{K}]\mathbb{K}\mathbb{K}}} & {U^{[\mathbb{J},\mathbb{K}]\mathbb{K}\mathbb{K}}}\exp \left( -i{K^{\mathbb{K}\mathbb{K}}}{A^{\mathbb{K}\mathbb{K}}} \right) \\
			{U^{[\mathbb{K},\mathbb{L}]\mathbb{K}\mathbb{K}}}\exp \left( -i{K^{\mathbb{K}\mathbb{K}}}{A^{\mathbb{K}\mathbb{K}}} \right) & -{S^{[\mathbb{K},\mathbb{L}]\mathbb{K}\mathbb{K}}} \\
		\end{matrix} \right]}^{-1}}
		\left[ \begin{matrix}
			{S^{[\mathbb{J},\mathbb{K}]\mathbb{K}\mathbb{J}}} & {O^{\mathbb{K}\mathbb{L}}} \\
			{O^{\mathbb{K}\mathbb{J}}} & {S^{[\mathbb{K},\mathbb{L}]\mathbb{K}\mathbb{L}}} \\
		\end{matrix} \right]
		\left[ \begin{matrix}
			{c^{[\mathbb{J},\mathbb{K}]\triangleleft \mathbb{J}}} \\
			{c^{[\mathbb{K},\mathbb{L}]\triangleleft \mathbb{L}}} \\
		\end{matrix} \right]
	\end{aligned}
\end{equation}
\end{strip}
The equality (\ref{Eq083}) expresses the amplitudes of the waves incident from the internal branches $c^{[\mathbb{J},\mathbb{K}]\triangleleft \mathbb{K}}$ and $c^{[\mathbb{K},\mathbb{L}]\triangleleft \mathbb{K}}$ in terms of the amplitudes of the waves incident from the external branches, $c^{[\mathbb{J},\mathbb{K}]\triangleleft \mathbb{J}}$ and $c^{[\mathbb{K},\mathbb{L}]\triangleleft \mathbb{L}}$. Substituting equality (\ref{Eq083}) in expression (\ref{Eq070}), taking into account equation (\ref{Eq067}) we obtain formula (\ref{Eq071}).	
\end{proof}

For the operators used in the formula (\ref{Eq071}) $U^{[\mathbb{J},\mathbb{K}]}$ and $U^{[\mathbb{K},\mathbb{L}]}$ (\ref{Eq072}), the unitary property follows:
\begin{equation}	\label{Eq084}
	\{U^{\square}{\bar{U}}^{\square} = I^{\square} = {\bar{U}}^{\square}U^{\square}\}^{\square \ = \ [\mathbb{J},\mathbb{K}],[\mathbb{K},\mathbb{L}]}
\end{equation}
When executing the properties ${{\bar{W}}^{[\mathbb{J},\mathbb{K}]k}}h_m^{[\mathbb{J},\mathbb{K}]k}\parallel {{\bar{W}}^{[\mathbb{K},\mathbb{L}]k}}h_m^{[\mathbb{K},\mathbb{L}]k}$ and ${{\bar{W}}^{[\mathbb{J},\mathbb{K}]k}}h_m^{[\mathbb{J},\mathbb{K}]k}\underset{m\ne n}{\mathop{\bot}}\,{{\bar{W}}^{[\mathbb{K},\mathbb{L}]k}}h_n^{[\mathbb{K},\mathbb{L}]k}$, from the definition (\ref{Eq072}), it also follows that these operators are diagonal:
\begin{equation}	\label{Eq085}
	\{U^{\square} = \operatorname{diag}\}^{\square \ = \ [\mathbb{J},\mathbb{K}],[\mathbb{K},\mathbb{L}]}
\end{equation}

The combining formula (\ref{Eq071}) is universal. When external branches absent, for example, on the right in Figure~\ref{fig:3} ($\mathbb{L} = \varnothing $, ${{\Omega}^{[\mathbb{K},\varnothing ]}} = {{\Omega}^{[\mathbb{K}]}}$), formula (\ref{Eq071}) is also applicable (at the same time, ${S^{[\mathbb{J},\varnothing ]}} = {S^{[\mathbb{J}]}}$, ${S^{[\mathbb{K},\varnothing ]}} = {S^{[\mathbb{K}]}}$). Taking into account the agreements in section \ref{sec:2_1_1} it takes thes form:
\begin{strip}
\begin{equation}	\label{Eq086}
	\begin{aligned}
		S^{[\mathbb{J}]} = {S^{[\mathbb{J},\mathbb{K}]\mathbb{J}\mathbb{J}}}+
		\left[ \begin{matrix}
		{S^{[\mathbb{J},\mathbb{K}]\mathbb{J}\mathbb{K}}} & {O^{\mathbb{J}\mathbb{K}}} \\
		\end{matrix} \right]
		{{\left[ \begin{matrix}
		-{S^{[\mathbb{J},\mathbb{K}]\mathbb{K}\mathbb{K}}} & {U^{[\mathbb{J},\mathbb{K}]\mathbb{K}\mathbb{K}}}\exp \left( -i{K^{\mathbb{K}\mathbb{K}}}{A^{\mathbb{K}\mathbb{K}}} \right) \\
		{U^{[\mathbb{K}]\mathbb{K}\mathbb{K}}}\exp \left( -i{K^{\mathbb{K}\mathbb{K}}}{A^{\mathbb{K}\mathbb{K}}} \right) & -{S^{[\mathbb{K}]\mathbb{K}\mathbb{K}}} \\
		\end{matrix} \right]}^{-1}}
		\left[ \begin{matrix}
		{S^{[\mathbb{J},\mathbb{K}]\mathbb{K}\mathbb{J}}} \\
			{O^{\mathbb{K}\mathbb{J}}} \\
		\end{matrix} \right]
	\end{aligned}
\end{equation}
\end{strip}
If there is no connection between junctions ($\mathbb{K} = \varnothing $, ${{\Omega}^{[\mathbb{J},\varnothing ]}} = {{\Omega}^{[\mathbb{J}]}}$, ${{\Omega}^{[\varnothing,\mathbb{L}]}} = {{\Omega}^{[\mathbb{L}]}}$, ${S^{[\mathbb{J},\varnothing ]}} = {S^{[\mathbb{J}]}}$, ${S^{[\varnothing,\mathbb{L}]}} = {S^{[\mathbb{L}]}}$), the formula (\ref{Eq071}) is written as
\begin{equation}	\label{Eq087}
	{S^{[\mathbb{J},\mathbb{L}]}} = \left[ \begin{matrix}
	{S^{[\mathbb{J}]\mathbb{J}\mathbb{J}}} & {O^{\mathbb{J}\mathbb{L}}} \\
		{O^{\mathbb{L}\mathbb{J}}} & {S^{[\mathbb{L}]\mathbb{L}\mathbb{L}}} \\
		\end{matrix} \right]
\end{equation}
Thus, due to the agreements of section \ref{sec:2_1_1}, formula (\ref{Eq071}) is applicable for combining network junctions in all possible cases.

\subsubsection{Network combining formula}	\label{sec:2_4_3}

Based on the combining formula (\ref{Eq071}), assuming in it that
\begin{equation}	\label{Eq088}
	\begin{aligned}
		\mathbb{J} &= \ \mathbb{A}\backslash \mathbb{B} \!\!\!\!&= \{k\in \mathbb{A}|k\notin \mathbb{B}\} \\
		\mathbb{K} &= \mathbb{A}{\textstyle \bigcap} \mathbb{B} \!\!\!\!&= \{k\in \mathbb{A}|k\in \mathbb{B}\} \\
		\mathbb{L} &= \ \mathbb{B}\backslash \mathbb{A} \!\!\!\!&= \{k\in \mathbb{B}|k\notin \mathbb{A}\}
	\end{aligned}
\end{equation}
for S-matrices with tuple identifiers $\mathbb{A}$ and $\mathbb{B}$ we define a \textit{combining operation}:
\begin{equation}	\label{Eq089}
	{S^{[\mathbb{A}]}}\circledast {S^{[\mathbb{B}]}} := {S^{[\mathbb{A}\backslash \mathbb{B},\mathbb{B}\backslash \mathbb{A}]}} = {S^{[\mathbb{A}\ominus \mathbb{B}]}}
\end{equation}
At the same time ${S^{[\mathbb{J},\mathbb{K}]}}\sim {S^{[\mathbb{A}]}}$ and ${S^{[\mathbb{K},\mathbb{L}]}}\sim {S^{[\mathbb{B}]}}$ (may differ from each other by permutation of rows and columns). The operation is not binary, since, in addition to the elements of the matrices $S^{[\mathbb{J},\mathbb{K}]}$ and $S^{[\mathbb{K},\mathbb{L}]}$, (\ref{Eq071}) contains elements $U^{[\mathbb{J},\mathbb{K}]\mathbb{K}\mathbb{K}}$, $U^{[\mathbb{K},\mathbb{L}]\mathbb{K}\mathbb{K}}$ and $\exp \left( -i{K^{\mathbb{K}\mathbb{K}}}{A^{\mathbb{K}\mathbb{K}}} \right)$. In this work, they are omitted for short.

The location of the branch boundaries in the network can be chosen arbitrarily, in particular so that
\begin{equation}	\label{Eq090}
	\{a^k = 0\}^{k\in \mathbb{I}}
\end{equation}
Then according to expressions (\ref{Eq073}), (\ref{Eq071}) and (\ref{Eq090}) we have
\begin{strip}
\begin{equation}	\label{Eq091}
	\begin{aligned}
		{S^{[\mathbb{J},\mathbb{L}]}} =&
		\left[ \begin{matrix}
		{S^{[\mathbb{J},\mathbb{K}]\mathbb{J}\mathbb{J}}} & {O^{\mathbb{J}\mathbb{L}}} \\
		{O^{\mathbb{L}\mathbb{J}}} & {S^{[\mathbb{K},\mathbb{L}]\mathbb{L}\mathbb{L}}} \\
		\end{matrix} \right]+
		\left[ \begin{matrix}
		{S^{[\mathbb{J},\mathbb{K}]\mathbb{J}\mathbb{K}}} & {O^{\mathbb{J}\mathbb{K}}} \\
		{O^{\mathbb{L}\mathbb{K}}} & {S^{[\mathbb{K},\mathbb{L}]\mathbb{L}\mathbb{K}}} \\
		\end{matrix} \right]
		{{\left[ \begin{matrix}
		-{S^{[\mathbb{J},\mathbb{K}]\mathbb{K}\mathbb{K}}} & {U^{[\mathbb{J},\mathbb{K}]\mathbb{K}\mathbb{K}}} \\
		{U^{[\mathbb{K},\mathbb{L}]\mathbb{K}\mathbb{K}}} & -{S^{[\mathbb{K},\mathbb{L}]\mathbb{K}\mathbb{K}}} \\
		\end{matrix} \right]}^{-1}}
		\left[ \begin{matrix}
		{S^{[\mathbb{J},\mathbb{K}]\mathbb{K}\mathbb{J}}} & {O^{\mathbb{K}\mathbb{L}}} \\
		{O^{\mathbb{K}\mathbb{J}}} & {S^{[\mathbb{K},\mathbb{L}]\mathbb{K}\mathbb{L}}} \\
		\end{matrix} \right]
	\end{aligned}
\end{equation}
\end{strip}
There is no element $\exp \left( -i{K^{\mathbb{K}\mathbb{K}}}{A^{\mathbb{K}\mathbb{K}}} \right)$ in formula (\ref{Eq091}). Therefore, in the case of a “branchless” network (\ref{Eq090}), the combining operation (\ref{Eq089}) is simplified. Formally, it is convenient. However, branch lengths are parameters that affect the S-matrix of the network. Their explicit accounting is fundamentally important when designing networks with predetermined transport properties.

Using definition (\ref{Eq089}) we write the \textit{network combining formula} (\textit{NCF})~--- an expression for the S-matrix of an arbitrary quantum network in terms of its junctions' S-matrices:
\begin{equation}	\label{Eq092}
	{S^{[\mathbb{E}]}} = {{\circledast}^{\mathbb{A}\in \mathcal{N}}}{S^{[\mathbb{A}]}}
\end{equation}
Here, the tuple $\mathcal{N}$ also sets the order of combining network junctions by formula (\ref{Eq089}).

Since formula (\ref{Eq092}) is written in terms of extended scattering matrices, it allows correctly taking into account closed channels in the network and, therefore, tunnel effects between junctions. Due to the agreements and notation of section \ref{sec:2_1}, formula (\ref{Eq092}) is convenient both in analytical and numerical calculations. In the particular case, the result obtained using it coincides with the previously one obtained in the literature (Appendix \ref{sec:D_2}).

\section{Electron transport in quantum network}	\label{sec:3}

The approach described in section \ref{sec:2} enables efficient calculation of the scattering properties of semiconductor nanostructures using quantum network model. In section \ref{sec:3}, we will show how the Landauer--B\"uttiker formalism \cite{Bib017,Bib018} is used to find electric currents through the network as part of a calculation scheme suggested in this work.

\subsection{Probability currents and S-matrix}	\label{sec:3_1}

\begin{equation}	\label{Eq093}
	\{\ {{\square}^{[\mathbb{A}]}}\ \mapsto \ \square \ |\ \square \ = \psi,\iota,...\}\quad \operatorname{in}\ \text{\ref{sec:3_1}}
\end{equation}

The Landauer--B\"uttiker formalism uses the scattering probabilities of charge carriers. The association of scattering probabilities with the S-matrix of the quantum network can be established based on expressions for probability currents in the branches.

\subsubsection{Probability currents}	\label{sec:3_1_1}

In terms of the stationary problem (\ref{Eq011}), dimensionless incident and scattered currents in the branch-channel $_m^k$ take the form
\begin{equation}	\label{Eq094}
	\begin{aligned}
		&\iota _m^{\triangleleft k} := \frac{i}{2}\int_{{{\beta}^k}}{ds\left( \psi _m^{\triangleleft k}{{\partial}_1}\bar{\psi}_m^{\triangleleft k}-\bar{\psi}_m^{\triangleleft k}{{\partial}_1}\psi _m^{\triangleleft k} \right)},	\\
		&\iota _m^{\triangleright k} := \frac{i}{2}\int_{{{\beta}^k}}{ds\left( \psi _m^{\triangleright k}{{\partial}_1}\bar{\psi}_m^{\triangleright k}-\bar{\psi}_m^{\triangleright k}{{\partial}_1}\psi _m^{\triangleright k} \right)}
	\end{aligned}
\end{equation}
By rewriting the expressions (\ref{Eq037}) as
\begin{equation}	\label{Eq095}
	\begin{aligned}
		&\psi _m^{\triangleleft k}(x,y,z) = c_m^{\triangleleft k}\exp \left( -i\kappa _m^kx \right)h_m^k(y,z),	\\
		&\psi _m^{\triangleright k}(x,y,z) = c_m^{\triangleright k}\exp \left( +i\kappa _m^kx \right)h_m^k(y,z)
	\end{aligned}
\end{equation}
from definitions (\ref{Eq094}), we get
\begin{equation}	\label{Eq096}
	\iota _m^{\triangleleft k} = -\kappa _m^k{{| c_m^{\triangleleft k} |}^2}\cdot [_m^k\in \mathbb{O}],\quad \iota _m^{\triangleright k} = +\kappa _m^k{{| c_m^{\triangleright k} |}^2}\cdot [_m^k\in \mathbb{O}]
\end{equation}
where $\mathbb{O} := \{_m^k|\ \lambda _m^k<\varepsilon \}$ is a tuple of numbers of \textit{open branch-channels}, $\bar{\mathbb{O}} = \{_m^k|\ \lambda _m^k\ge \varepsilon \}$ is a tuple of numbers of \textit{closed branch-channels}.

Assuming that the currents in the closed channels are zero, for the total current in the branch-channel $_m^k$ we have
\begin{equation}	\label{Eq097}
	\iota _m^k = \iota _m^{\triangleleft k}+\iota _m^{\triangleright k}
\end{equation}
When calculating the currents in the branch ${\Omega}^k$, one can show that the currents in the channels are additive. Then taking into account the expressions (\ref{Eq096}) for incident and scattered currents we have
\begin{equation}	\label{Eq098}
	\begin{aligned}
	{{\iota}^{\triangleleft}} = \sum\nolimits_{}^k{{{\iota}^{\triangleleft k}}} = \sum\nolimits_m^k{\iota _m^{\triangleleft k}} = -\sum\nolimits_m^k{[_m^k\in \mathbb{O}]\cdot \kappa _m^k{{| c_m^{\triangleleft k} |}^2}} \\
		{{\iota}^{\triangleright}} = \sum\nolimits_{}^k{{{\iota}^{\triangleright k}}} = \sum\nolimits_m^k{\iota _m^{\triangleright k}} = +\sum\nolimits_m^k{[_m^k\in \mathbb{O}]\cdot \kappa _m^k{{| c_m^{\triangleright k} |}^2}}
	\end{aligned}
\end{equation}
and the total current takes a form
\begin{equation}	\label{Eq099}
	\iota := \sum\nolimits_{}^k{{{\iota}^k}} = \sum\nolimits_m^k{\iota _m^k} = \sum\nolimits_m^k{\left( \iota _m^{\triangleleft k}+\iota _m^{\triangleright k} \right)}
\end{equation}

\subsubsection{Structure of operator $K$}	\label{sec:3_1_2}

The channels separated on open and closed ones (\ref{Eq029}) sets the structure for operators introduced in section \ref{sec:2}. In particular, the operator $K$ in the representation (\ref{Eq032}) takes the form
\begin{equation}	\label{Eq100}
	K = \left[ \begin{matrix}
	{K_{\mathbb{O}\mathbb{O}}} & {O_{\mathbb{O}\bar{\mathbb{O}}}} \\
		{O_{\bar{\mathbb{O}}\mathbb{O}}} & {K_{\bar{\mathbb{O}}\bar{\mathbb{O}}}} \\
		\end{matrix} \right]
\end{equation}
The operator $K$ can also be written as a sum of Hermitian and anti-Hermitian operators:
\begin{equation}	\label{Eq101}
	K = {K_{++}}+{K_{--}}
\end{equation}
\begin{equation}	\label{Eq102}
	\begin{aligned}
	{K_{++}} := \tfrac{1}{2}\left( K+\bar{K} \right),\quad {{{\bar{K}}}_{++}} = +{K_{++}} \\
		{K_{--}} := \tfrac{1}{2}\left( K-\bar{K} \right),\quad {{{\bar{K}}}_{--}} = -{K_{--}}
	\end{aligned}
\end{equation}
According to definitions (\ref{Eq102}), in the representation (\ref{Eq034}) we get
\begin{equation}	\label{Eq103}
	\begin{aligned}
		&K_{++}^{kl} = {I^{kl}}\sum\nolimits_n{[_n^k\in \mathbb{O}]\cdot | h_n^k \rangle \kappa _n^k\langle h_n^k |},	\\
		&K_{--}^{kl} = {I^{kl}}\sum\nolimits_n{[_n^k\in \bar{\mathbb{O}}]\cdot | h_n^k \rangle \kappa _n^k\langle h_n^k |}
	\end{aligned}
\end{equation}
Then accurate to the representation and expansion from expressions (\ref{Eq100}), (\ref{Eq101}), (\ref{Eq103}) follows
\begin{equation}	\label{Eq104}
	{K_{++}} = {K_{\mathbb{O}\mathbb{O}}},\quad {K_{--}} = {K_{\bar{\mathbb{O}}\bar{\mathbb{O}}}}
\end{equation}
which is completely compliance with the notation used in the literature
\begin{equation}	\label{Eq105}
	\forall f\quad {f_+} := {f_{\mathbb{O}}},\quad {f_-} := {f_{{\bar{\mathbb{O}}}}}
\end{equation}
At that the expression (\ref{Eq100}) takes the form
\begin{equation}	\label{Eq106}
	K = \left[ \begin{matrix}
	{K_{++}} & {O_{+-}} \\
		{O_{-+}} & {K_{--}} \\
		\end{matrix} \right]
\end{equation}

\subsubsection{Current conservation and S-matrix}	\label{sec:3_1_3}

Based on the structure of the form (\ref{Eq106}) set for operators, we can write the law of current conservation in terms of the S-matrix of the quantum network.

\begin{prop}	\label{prop:7}
 The S-matrix of the quantum network has the property
\begin{equation}	\label{Eq107}
	{{\bar{S}}_{++}}{K_{++}}{S_{++}}-{K_{++}} = {O_{++}}
\end{equation}
where $S_{++}$ is the S-matrix submatrix that relates wave amplitudes in open channels with each other:
\begin{equation}	\label{Eq108}
	{c^{\triangleright}} = \left[ \begin{matrix}
	c_+^{\triangleright} \\
		c_-^{\triangleright} \\
		\end{matrix} \right] = \left[ \begin{matrix}
	{S_{++}} & {S_{+-}} \\
		{S_{-+}} & {S_{--}} \\
		\end{matrix} \right]\left[ \begin{matrix}
	c_+^{\triangleleft} \\
		c_-^{\triangleleft} \\
		\end{matrix} \right] = S{c^{\triangleleft}}
\end{equation}
\end{prop}
\begin{proof}[\textsc{Proof~\ref{prop:7}}]
 With elastic scattering, the total current (\ref{Eq099}) is zero:
\begin{equation}	\label{Eq109}
	\forall c\quad \iota = 0
\end{equation}
From expressions (\ref{Eq096}), (\ref{Eq097}), (\ref{Eq109}), (\ref{Eq106}) follows
\begin{equation}	\label{Eq110}
	\langle c_+^{\triangleright} |{K_{++}}| c_+^{\triangleright} \rangle -\langle c_+^{\triangleleft} |{K_{++}}| c_+^{\triangleleft} \rangle = 0
\end{equation}
According to notation (\ref{Eq105}), the structure of the S-matrix (\ref{Eq038}) is similar to the structure (\ref{Eq106}) and can be written as (\ref{Eq108}). Taking into account agreements \hyperref[agr:IV]{IV} and \hyperref[agr:VIII]{VIII}, we have
\begin{equation}	\label{Eq111}
	{{\bar{\square}}_{\mathbb{K}\mathbb{L}}} = \overline{{{\square}_{\mathbb{L}\mathbb{K}}}}
\end{equation}
where $\square $ is a complex matrix. Then from expressions (\ref{Eq110}), (\ref{Eq108}) and (\ref{Eq111}) we get
\begin{equation}	\label{Eq112}
	\begin{aligned}
		0 =& \langle c_+^{\triangleleft} | \!\left( {\bar{S}}_{++}K_{++}S_{++} \!-\! K_{++} \right)\! | c_+^{\triangleleft} \rangle \!+\! \langle c_+^{\triangleleft} |{\bar{S}}_{++}K_{++}S_{+-}| c_-^{\triangleleft} \rangle \\
		&+ \langle c_-^{\triangleleft} |{\bar{S}}_{-+}K_{++}S_{++}| c_+^{\triangleleft} \rangle + \langle c_-^{\triangleleft} |{\bar{S}}_{-+}K_{++}S_{+-}| c_-^{\triangleleft} \rangle
	\end{aligned}
\end{equation}
Since the S-matrix does not depend on the incident waves, we set $c_-^{\triangleleft} = 0$. Then, according to the formula (\ref{Eq112}), the property (\ref{Eq107}) is executed for the submatrix $S_{++}$.	
\end{proof}

\subsubsection{Extended current scattering matrix}	\label{sec:3_1_4}

Property (\ref{Eq107}) can be formulated in terms of a unitary matrix. We write expressions (\ref{Eq096}) as
\begin{equation}	\label{Eq113}
	\begin{aligned}
	\iota _m^{\triangleleft k} = -{{| d_m^{\triangleleft k} |}^2}\cdot [_m^k\in \mathbb{O}],\quad {d^{\triangleleft}} := {K^{1/2}}{c^{\triangleleft}} \\
		\iota _m^{\triangleright k} = +{{| d_m^{\triangleright k} |}^2}\cdot [_m^k\in \mathbb{O}],\quad {d^{\triangleright}} := {K^{1/2}}{c^{\triangleright}}
	\end{aligned}
\end{equation}
where $d^{\triangleleft}$ and $d^{\triangleright}$ are incident and scattered \textit{current amplitudes} respectively. The relation between them has the same form as between wave amplitudes (\ref{Eq038}):
\begin{equation}	\label{Eq114}
	{d^{\triangleright}} = :C{d^{\triangleleft}}\quad \Leftrightarrow \quad d_m^{\triangleright k} = :\sum\nolimits_n^l{C_{mn}^{kl}d_n^{\triangleleft l}}
\end{equation}
where $C$ is an \textit{extended current scattering matrix} (combining the concepts of extended \cite[p.~155]{Bib027} and current \cite{Bib025} scattering matrices). From expressions (\ref{Eq113}) and (\ref{Eq114}) for currents we have
\begin{equation}	\label{Eq115}
	\iota _m^{\triangleright k} = \overline{[C{d^{\triangleleft}}]_m^k}[C{d^{\triangleleft}}]_m^k\cdot [_m^k\in \mathbb{O}]
\end{equation}
From expressions (\ref{Eq113}), (\ref{Eq114}) and (\ref{Eq038}), we can find the relation of matrix $C$ with matrix $S$:
\begin{equation}	\label{Eq116}
	C = {K^{+1/2}}S{K^{-1/2}}\quad \Leftrightarrow \quad S = {K^{-1/2}}C{K^{+1/2}}
\end{equation}
The distinction between the concepts of a matrix relating wave amplitudes and a matrix relating current amplitudes avoids ambiguity in the interpretation of analytical and numerical calculations.

Taking into account structure (\ref{Eq106}) from formula (\ref{Eq116}) we have
\begin{equation}	\label{Eq117}
	\left[ \begin{matrix}
	{S_{++}} & {S_{+-}} \\
		{S_{-+}} & {S_{--}} \\
		\end{matrix} \right] = \left[ \begin{matrix}
	K_{++}^{-1/2}{C_{++}}K_{++}^{+1/2} & K_{++}^{-1/2}{C_{+-}}K_{--}^{+1/2} \\
		K_{--}^{-1/2}{C_{-+}}K_{++}^{+1/2} & K_{--}^{-1/2}{C_{--}}K_{--}^{+1/2} \\
		\end{matrix} \right]
\end{equation}
According to expression (\ref{Eq117}), the property (\ref{Eq107}) takes the form
\begin{equation}	\label{Eq118}
	{{\bar{C}}_{++}}{C_{++}} = {I_{++}} = {C_{++}}{{\bar{C}}_{++}}
\end{equation}

Thus, the unitary is a submatrix of an extended current scattering matrix that relates current amplitudes in open channels with each other~--- the \textit{current scattering matrix} $C_{++}$. According to property (\ref{Eq118}), the squares of the modules of its elements have a probabilistic interpretation.

\subsection{Electrical properties}	\label{sec:3_2}

\begin{equation}	\label{Eq119}
	\{\ {{\square}^{[\mathbb{E}]}}\ \mapsto \ \square \ |\ \square \ = J,P,...\}\quad \operatorname{in}\ \text{\ref{sec:3_2}}
\end{equation}

Based on the scattering probabilities of charge carriers in a quantum network, it is possible to find its electrical properties in the framework of the Landauer--B\"uttiker formalism. For this purpose, we will calculate electric currents in terms of average values from the ensemble of charge carriers in the external branches of the network. We will assume that they are smoothly connected to the reservoirs of charge carriers. Then when calculating, we can only take into account the scattering properties of the network (section \ref{sec:2}).

\subsubsection{Currents in branch-channels}	\label{sec:3_2_1}

The difference in voltages on the reservoirs leads to the appearance of currents through the quantum network. Since the external branches of the network are long, we neglect the effects of barrier tunneling in them. This means taking into account only open channels. For current in the $k$-th branch we have:
\begin{equation}	\label{Eq120}
	{J^k} = {J^{\triangleleft k}}+{J^{\triangleright k}}
\end{equation}
where $J^{\triangleleft k}$ is current incident on the network from the $k$-th branch, $J^{\triangleright k}$ is current scattered by the network to the $k$-th branch. According to formula (\ref{Eq098}), the currents in the channels are additive:
\begin{equation}	\label{Eq121}
	{J^{\triangleleft k}} = \sum\nolimits_m{J_m^{\triangleleft k}},\quad {J^{\triangleright k}} = \sum\nolimits_m{J_m^{\triangleright k}}
\end{equation}
where $J_m^{\triangleleft k}$ is current incident on the network from the $_m^k$-th branch-channel, $J_m^{\triangleright k}$ is current scattered by the network to the $_m^k$-th branch-channel.

\begin{prop}	\label{prop:8}
 The incident and scattered currents in the $_m^k$-th branch-channel are calculated as
\begin{equation}	\label{Eq122}
	J_m^{\triangleleft k} = \int_{-\infty}^{+\infty}{dEj_m^{\triangleleft k}\left( E \right)},\quad J_m^{\triangleright k} = \int_{-\infty}^{+\infty}{dEj_m^{\triangleright k}\left( E \right)}
\end{equation}
\begin{equation}	\label{Eq123}
	\begin{aligned}
		&j_m^{\triangleleft k}\left( E \right) := -\frac{e}{\pi \hbar}{f^k}\left( E \right)[E_{\bot m}^k<E],	\\
		&j_m^{\triangleright k}\left( E \right) := -\sum\nolimits_n^l{P_{mn}^{kl}\left( E \right)j_n^{\triangleleft l}\left( E \right)}
	\end{aligned}
\end{equation}
\begin{equation}	\label{Eq124}
	{f^k}\left( E \right) := \frac{1}{\exp [(E-E_{\text{F}}^k)/({k_0}{T^k})]+1}
\end{equation}
\begin{equation}	\label{Eq125}
	P_{mn}^{kl}\left( E \right) := [E_{\bot m}^k<E]{{| C_{mn}^{kl}\left( 2m{L^2}{{\hbar}^{-2}}E \right) |}^2}[E>E_{\bot n}^l]
\end{equation}
where $e$ is the charge of the charge carrier, $f^k$ is the distribution function in the $k$-th branch (reservoir), $E_{\text{F}}^k$ is the Fermi level in the $k$-th branch (reservoir), $P = \{P_{mn}^{kl}\}_{mn}^{kl}$ is the \textit{matrix of scattering probabilities}.

\end{prop}
\begin{proof}[\textsc{Proof~\ref{prop:8}}]
 Outside the domains of connection with reservoirs, branches are translational invariant. Since structurally the branch is a quantum wire, the dispersion law of charge carriers in the energy frame of section \ref{sec:2_2} takes the form:
\begin{equation}	\label{Eq126}
	E_M^k = E_{snp}^k = E_{\bot n}^k+\frac{p_x^2}{2m}
\end{equation}
\begin{equation}	\label{Eq127}
	E_{\bot n}^k = {{\hbar}^2}/(2m{L^2})\cdot \lambda _n^k
\end{equation}
where $M = \{snp\}$ is a multi-index (set of quantum numbers), uniquely identifying the state, $s = 1,2$ is a spin quantum number, $n$ is a quantum number corresponding to size-quantization, $p$ is a quasi-momentum, $E_{\bot n}^k$ is the dimensional energy of the $_n^k$-th branch-channel, $p_x$ is a projection of the quasi-momentum along the branch. The density of single-particle states in the $k$-th branch is written as
\begin{equation}	\label{Eq128}
	{g^k}\left( E \right) := \sum\limits_M{\delta \left( E_M^k-E \right)}
\end{equation}
where $\delta $ is the Dirac delta function. Replacing in definition (\ref{Eq128}) from summation by $p_x$ to integration taking into account formula (\ref{Eq126}), we have
\begin{equation}	\label{Eq129}
	{g^k} = \sum\nolimits_n{g_n^k}
\end{equation}
\begin{equation}	\label{Eq130}
	g_n^k\left( E \right) := \frac{L_x^k}{\pi \hbar}\int_{-\infty}^{+\infty}{d{p_x}\delta \left( E_{\bot n}^k+\frac{p_x^2}{2m}-E \right)}
\end{equation}
where $g_n^k$ is density of states in the $_n^k$-th branch-channel, $L_x^k$ is length of the $k$-th branch. Based on $g_n^k$ one can find the average value of the physical quantity $Q$ in the $_n^k$-th branch-channel:
\begin{equation}	\label{Eq131}
	\langle Q \rangle _n^k = \frac{1}{N_n^k}\int_{-\infty}^{+\infty}{dE{f^k}\left( E \right)g_n^k\left( E \right)Q}
\end{equation}
\begin{equation}	\label{Eq132}
	N_n^k = \int_{-\infty}^{+\infty}{dE{f^k}\left( E \right)g_n^k\left( E \right)}
\end{equation}
where $N_n^k$ is the number of particles in the $_n^k$-th branch-channel. We can see from definition (\ref{Eq124}), that the distribution functions in the branches differ in the positions of their Fermi levels. Fermi level $E_{\text{F}}^k$ is written as
\begin{equation}	\label{Eq133}
	E_{\text{F}}^k = {E_{\text{F}}}-e{U^k}
\end{equation}
where $E_{\text{F}}$ is the Fermi level of the quantum network in the absence of voltages on the reservoirs, $U^k$ is the voltage applied to the $k$-th reservoir.

We find incident and scattered currents in the $_m^k$-th external branch-channel as the average values (\ref{Eq131}):
\begin{equation}	\label{Eq134}
	\begin{aligned}
		&J_m^{\triangleleft k} = \frac{e}{L_x^k}\langle [{p_x}<0]{p_x}/m \rangle _m^kN_m^k,	\\
		&J_m^{\triangleright k} = \sum\nolimits_n^l{\frac{e}{L_x^l}\langle P_{mn}^{kl}[{p_x}>0]{p_x}/m \rangle _n^lN_n^l}
	\end{aligned}
\end{equation}
As shown in section \ref{sec:3_1}, the squares of the element modules of matrix $C_{++}$ have a probabilistic interpretation. Therefore, taking into account definitions (\ref{Eq013}), we have formula (\ref{Eq125}). From expressions (\ref{Eq134}), (\ref{Eq131}) and (\ref{Eq130}) we obtain formulas (\ref{Eq122}).	
\end{proof}

\subsubsection{Currents in branches and conductivity}	\label{sec:3_2_2}

Using expressions for currents in the branch-channels (\ref{Eq122}) it is possible to find currents in the branches (\ref{Eq120}).

\begin{prop}	\label{prop:9}
 Current in the $k$-th branch is calculated as
\begin{equation}	\label{Eq135}
	{J^k} = \frac{e}{\pi \hbar}\sum\nolimits_{mn}^l{\int_{-\infty}^{+\infty}{dEP_{mn}^{kl}\left( E \right)\left\{ {f^l}\left( E \right)-{f^k}\left( E \right) \right\}}}
\end{equation}
where $\{P_{mn}^{kl}\}_{mn}^{l\ne k}$ are matrix probabilities of transmission into the $k$-th branch (\textit{transparencies}).

\end{prop}
\begin{proof}[\textsc{Proof~\ref{prop:9}}]
 See Appendix \ref{sec:E_1}.	
\end{proof}

Error of current calculation by formula (\ref{Eq135}) can be estimated by value of total current
\begin{equation}	\label{Eq136}
	J := \sum\nolimits_{}^{k\in \mathbb{E}}{{J^k}}
\end{equation}
Theoretically, it should be zero due to law of charge conservation.

Let us write the expression for currents (\ref{Eq135}) in dimensionless form. It is convenient in calculations together with the proposed calculation scheme of the S-matrix of the quantum network (sections \ref{sec:2_3} and \ref{sec:2_4}). From expressions (\ref{Eq135}) and (\ref{Eq125}) taking into account definitions (\ref{Eq013}) we have
\begin{equation}	\label{Eq137}
	\begin{aligned}
		{\rm I}^k = \sum\nolimits_{mn}^l\int_{-\infty}^{+\infty}&{d\varepsilon [\lambda _m^k<\varepsilon ]{{| C_{mn}^{kl}\left( \varepsilon \right) |}^2}[\varepsilon >\lambda _n^l]}	\\
		&\!\! \times \left\{ F_{-1}( [\varepsilon _{\text{F}}^l-\varepsilon ]/{{\mu}^l} ) - F_{-1}( [\varepsilon _{\text{F}}^k-\varepsilon ]/{{\mu}^k} ) \right\}
	\end{aligned}
\end{equation}
\begin{equation}	\label{Eq138}
	{F_{-1}}\left( \eta \right) = 1/(1+{e^{-\eta}})
\end{equation}
where $F_{-1}$ is the Fermi--Dirac integral of the order $-1$, and the relations of dimensionless quantities with dimensional ones is given by the formulas
\begin{equation}	\label{Eq139}
	{J^k} = :\frac{e\hbar}{2\pi m{L^2}}{{{\rm I}}^k}
\end{equation}
\begin{equation}	\label{Eq140}
	{{\mu}^k} := 2m{L^2}{{\hbar}^{-2}}{k_0}{T^k}
\end{equation}
\begin{prop}	\label{prop:10}
 In case of low temperatures
\begin{equation}	\label{Eq141}
	{{\max}^k}\{{T^k}\}\to 0
\end{equation}
current in $k$-th branch is calculated as
\begin{equation}	\label{Eq142}
	{J^k}\to \frac{e}{\pi \hbar}\sum\nolimits_{mn}^l{\int{dE[E\in (E_{\text{F}}^k,E_{\text{F}}^l)]P_{mn}^{kl}\left( E \right)}}
\end{equation}
where we use agreement \hyperref[agr:XII]{XII}.

\end{prop}
\begin{proof}[\textsc{Proof~\ref{prop:10}}]
 See Appendix \ref{sec:E_3}.	
\end{proof}

Taking into account the definition (\ref{Eq125}), expression (\ref{Eq142}) has a simple physical interpretation: the electric current in the $k$-th branch causes the branch-channels whose energies $\{E_{\bot m}^k\}_m^k$ are located between the Fermi level of the $k$-th branch and the Fermi levels of all other branches.

\begin{prop}	\label{prop:11}
 When reservoir voltages are low
\begin{equation}	\label{Eq143}
	{\{{U^k}\to 0\}}^{k\in \mathbb{E}}
\end{equation}
and reservoirs have same temperature
\begin{equation}	\label{Eq144}
	{\{{{\square}^k} = \ {{\square}^= }|\ \square \ = T,\mu \}}^{k\in \mathbb{E}}
\end{equation}
\pagebreak

\noindent
current in $k$-th branch is calculated as
\begin{equation}	\label{Eq145}
	{J^k}\to {{\tilde{J}}^k} := \sum\nolimits_{}^l{{{{\tilde{\sigma}}}^{kl}}{U^{lk}}}
\end{equation}
\begin{equation}	\label{Eq146}
	\begin{aligned}
		{\tilde{\sigma}}^{kl} := \frac{{e^2}}{\pi \hbar}\sum\nolimits_{mn}\frac{1}{4{\mu}^=} \int_{-\infty}^{+\infty}&{d\varepsilon [\lambda _m^k<\varepsilon ]{{| C_{mn}^{kl}\left( \varepsilon \right) |}^2}[\varepsilon >\lambda _n^l]}	\\
		&\times{\cosh}^{-2}\left( [{{\varepsilon}_{\text{F}}}-\varepsilon ]/[2{{\mu}^= }] \right)
	\end{aligned}
\end{equation}
\begin{equation}	\label{Eq147}
	{U^{lk}} := \left( E_{\text{F}}^l-E_{\text{F}}^k \right)/e
\end{equation}
where $\tilde{\sigma} := {{\{{{\tilde{\sigma}}^{kl}}\}}^{kl}}$ is the approximate conductivity of the network, $U^{lk}$ is the voltage between the reservoirs $l$ and $k$ (\textit{bias voltage}).

\end{prop}
\begin{proof}[\textsc{Proof~\ref{prop:11}}]
 See Appendix \ref{sec:E_2}.	
\end{proof}

We can see from definition (\ref{Eq146}), that the approximate conductivity of branched semiconductor nanostructures does not depend on bias voltages. This is a convenient quantity for analyzing their volt-ampere characteristics at low bias voltages.

\begin{prop}	\label{prop:12}
 At low voltages on reservoirs (\ref{Eq143}) with the same (\ref{Eq144}) low (\ref{Eq141}) temperature, the approximate conductivity is written as
\begin{equation}	\label{Eq148}
	{{\tilde{\sigma}}^{kl}}\to \frac{{e^2}}{\pi \hbar}\sum\nolimits_{mn}{[\lambda _m^k<{{\varepsilon}_{\text{F}}}]{{| C_{mn}^{kl}\left( {{\varepsilon}_{\text{F}}} \right) |}^2}[{{\varepsilon}_{\text{F}}}>\lambda _n^l]}
\end{equation}
\end{prop}
\begin{proof}[\textsc{Proof~\ref{prop:12}}]
 See Appendix \ref{sec:E_3}.	
\end{proof}

Thus, the material of section \ref{sec:3_2} is the last part of the scheme for calculating electronic transport in branched semiconductor nanostructures using a quantum network model.

\section{Two-dimensional quantum network of Q-, I- and Y-junctions}	\label{sec:4}

Due to the high level of development of planar technology, nanoelectronic devices based on two-dimensional electron gas are of interest. Their simplest model is a two-dimensional quantum network. In section \ref{sec:4}, we will demonstrate the approach proposed in sections \ref{sec:2} and \ref{sec:3} by calculations for a two-dimensional quantum network of smooth junctions with one, two and three adjacent branches.

\subsection{Arbitrary network}	\label{sec:4_1}

\subsubsection{Geometry}	\label{sec:4_1_1}

Let us consider a two-dimensional quantum network of smooth Q-, I- and Y-junctions~--- \textit{QIY-network} (Fig.~\ref{fig:4}). The absence of corners in the geometry of the junctions (their smoothness) corresponds to the typical nano-devices studied in experiment \cite{Bib014}. Due to its structure, the QIY network is suitable for modeling many two-dimensional nanostructures, as well as for demonstrating calculations using the scheme proposed in sections \ref{sec:2} and \ref{sec:3}.

\begin{figure}[htb]\center
	\hspace{30pt}\includegraphics{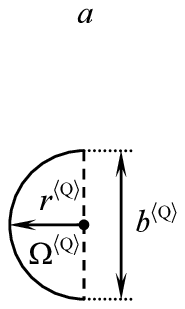}\hspace{50pt}	
	\includegraphics{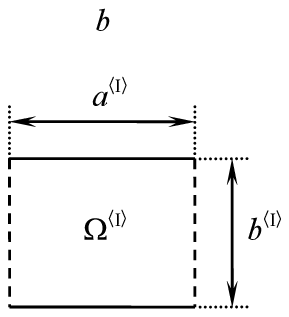}	\\
	\vspace{10pt}\includegraphics{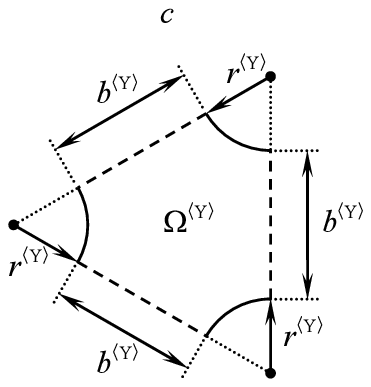}
	\caption{Junctions of smooth QIY-network: \textit{a}~--- Q-junction, \textit{b}~--- I-junction, \textit{c}~--- symmetric Y-junction.}\label{fig:4}
\end{figure}

During calculating we will use NCF (\ref{Eq092}), which is based on the combining formula (\ref{Eq071}). Formula (\ref{Eq071}) contains operators relates to changing LCFs on the branch boundaries: $U^{[\mathbb{J},\mathbb{K}]}$ and $U^{[\mathbb{K},\mathbb{L}]}$ (\ref{Eq072}). We choose such an LCF location when they are relates to each other by rotation and translation (Fig.~\ref{fig:5}). It is universal, since it is convenient during calculating the S-matrix of any two-dimensional network using the formula (\ref{Eq092}).

\begin{figure}[htb]\center
	\includegraphics{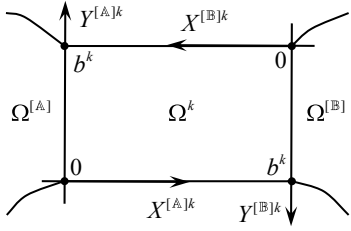}
	\caption{Universal location of local coordinate frames ${[XY]}^{[\mathbb{A}]k}$ and ${[XY]}^{[\mathbb{B}]k}$ in internal branch ${\Omega}^k$ of two-dimensional network.}\label{fig:5}
\end{figure}

In this work, we consider a QIY-network with branches of the same width $B^= $:
\begin{equation}	\label{Eq149}
	{\{{B^k} = {B^= }\}}^{k\in \mathbb{I}\bigcup \mathbb{E}}
\end{equation}
The quantity $B^= $ is assumed to be the characteristic length during the transition (\ref{Eq013}) and (\ref{Eq014}) to the dimensionless Schr\"odinger equation (\ref{Eq012}):
\begin{equation}	\label{Eq150}
	L = {B^= }
\end{equation}
Since dimensionless widths are defined as $\{b^k := {B^k}/L\}^{k\in \mathbb{I}\bigcup \mathbb{E}}$, taking into account equalities (\ref{Eq149}), assuming zero potential in all branches, we have:
\begin{equation}	\label{Eq151}
	{\{{{\Omega}^k}|{b^k} = 1,\ {u^k} = 0\}}^{k\in \mathbb{I}\bigcup \mathbb{E}}
\end{equation}
Solving the two-dimensional analogue of the problem (\ref{Eq025}) assuming ${{\beta}^k} = (0,{b^k})$ (Fig.~\ref{fig:5}), in case (\ref{Eq151}) we get
\begin{equation}	\label{Eq152}
	\{h_m^k\left( y \right) = \sqrt{2}\sin \left( \pi my \right),\ \lambda _m^k = {{(\pi m)}^2}\}_{m\in \mathbb{N}}^{k\in \mathbb{I}\bigcup \mathbb{E}}
\end{equation}
Since according expressions (\ref{Eq152}) the transverse modes $\{h_m^k\}_m^k$ and energies of the channels $\{\lambda _m^k\}_m^k$ are the same in all branches of the network, by analogy with equalities (\ref{Eq149}) we also have
\begin{equation}	\label{Eq153}
	\{\square _m^k = \ \square _m^= |\ \square \ = h,\lambda,\kappa \}_{m\in \mathbb{N}}^{k\in \mathbb{I}\bigcup \mathbb{E}}
\end{equation}

According to the chosen LCFs location (Fig.~\ref{fig:5}), for operators providing their changing in the combining formula (\ref{Eq071}), from expressions (\ref{Eq072}), (\ref{Eq030}) and (\ref{Eq152}) we have
\begin{equation}	\label{Eq154}
	\{U_{mn}^{[\mathbb{J},\mathbb{K}]kl} = {{(-1)}^{m+1}}I_{mn}^{kl} = U_{mn}^{[\mathbb{K},\mathbb{L}]kl}\}_{m,n\in \mathbb{N}}^{k,l\in \mathbb{I}\bigcup \mathbb{E}}
\end{equation}
Using the expression (\ref{Eq154}) in calculations according to NCF (\ref{Eq092}), one can find the S-matrix of the QIY-network based on the S-matrices of its junctions (Fig.~\ref{fig:4}) and information about its structure.

\subsubsection{Parameters}	\label{sec:4_1_2}

The structure of arbitrary QIY-network has the form
\begin{equation}	\label{Eq155}
	\mathcal{N} = {\textstyle \mathcal{Q}\bigcup \mathcal{I}\bigcup \mathcal{Y}}
\end{equation}
where $\mathcal{Q}$, $\mathcal{I}$, $\mathcal{Y}$ are tuples of Q-, I-, and Y-junction structural identifiers, respectively. For smoothness of the network, there must be no corners on the boundaries of the junctions with branches. Therefore, according to equalities (\ref{Eq151}) ${b^{\langle \text{Q} \rangle}} = {b^{\langle \text{I} \rangle}} = {b^{\langle \text{Y} \rangle}} = 1$ and the arch radius ${r^{\langle \text{Q} \rangle}} = 1/2$. We also assume ${r^{\langle \text{Y} \rangle}} = 1/2$ that provides the network small size. Without loss of generality of the approach, the electric field in the I-junction $\epsilon _{\bot}^{\langle \text{I} \rangle}$ will be considered isotropic and transverse. As a result, based on example (\ref{Eq010}), we have the following structure specificity of the QIY-network:
\begin{equation}	\label{Eq156}
	\begin{aligned}
		\{&{\Omega}^{[\mathbb{K}]} = {{\Omega}^{\langle \text{Q} \rangle}} \!\!\!\!\!&|& \ \ {r^{\langle \text{Q} \rangle}} = 1/2, & {b^{\langle \text{Q} \rangle}} & = 1 & {} & {\}^{\mathbb{K}\in \mathcal{Q}}} \\
		\{&{\Omega}^{[\mathbb{K}]} = {{\Omega}^{\langle \text{I} \rangle}} \!\!\!\!\!&|& \ \ {a^{\langle \text{I} \rangle}} = {a^{[\mathbb{K}]}}, & {b^{\langle \text{I} \rangle}} & = 1, & \epsilon _{\bot}^{\langle \text{I} \rangle} = \epsilon _{\bot}^{[\mathbb{K}]} & {\}^{\mathbb{K}\in \mathcal{I}}} \\
		\{&{\Omega}^{[\mathbb{K}]} = {{\Omega}^{\langle \text{Y} \rangle}} \!\!\!\!\!&|& \ \ {r^{\langle \text{Y} \rangle}} = 1/2, & {b^{\langle \text{Y} \rangle}} & = 1 & {} & {\}^{\mathbb{K}\in \mathcal{Y}}}
	\end{aligned}
\end{equation}

For a two-dimensional electron gas in all formulas $m = {m_{\text{e}}}$, $e = -{e_0}$, where $m_{\text{e}}$ is the effective mass of an electron, $e_0$ is an elementary charge. Then the dimensionless network parameters are written as a table~\ref{tab:2}.

\begin{table}
\caption{\label{tab:2}Dimensionless parameters of QIY-network}
\begin{tabularx}{\linewidth}{>{\raggedright}l>{\raggedright}X}
\toprule
definition	&	name \tabularnewline
\midrule
${{\{{{a}^{k}}={{A}^{k}}/L\}}^{k\in \mathbb{I}}}$	&	lengths of branches	\tabularnewline
${{\{{{b}^{k}}={{B}^{k}}/L\}}^{k\in \mathbb{I}\bigcup \mathbb{E}}}$	&	widths of branches	\tabularnewline
${{\{{{r}^{[\mathbb{K}]}}={{R}^{[\mathbb{K}]}}/L\}}^{\mathbb{K}\in \mathcal{Q}}}$	&	radiuses of Q-junctions arches	\tabularnewline
${{\{{{a}^{[\mathbb{K}]}}={{A}^{[\mathbb{K}]}}/L\}}^{\mathbb{K}\in \mathcal{I}}}$	&	lengths of I-junctions	\tabularnewline
${{\{{{b}^{[\mathbb{K}]}}={{B}^{[\mathbb{K}]}}/L\}}^{\mathbb{K}\in \mathcal{I}}}$	&	widths of I-junctions	\tabularnewline
${{\{\epsilon _{\bot}^{[\mathbb{K}]}=2{{m}_{{\text{e}}}}{{L}^{3}}{{e}_{0}}{{\hbar}^{-2}}\mathcal{E}_{\bot}^{[\mathbb{K}]}\}}^{\mathbb{K}\in \mathcal{I}}}$	&	electric field intensities in I-junctions	\tabularnewline
${{\{{{r}^{[\mathbb{K}]}}={{R}^{[\mathbb{K}]}}/L\}}^{\mathbb{K}\in \mathcal{Y}}}$	&	radiuses of Y-junctions arches	\tabularnewline
${{\{{{\mu}^{k}}=2{{m}_{{\text{e}}}}{{L}^{2}}{{\hbar}^{-2}}{{k}_{0}}{{T}^{k}}\}}^{k\in \mathbb{E}}}$	&	temperatures of reservoirs	\tabularnewline
${{\{\varepsilon _{{\text{F}}}^{k}=2{{m}_{{\text{e}}}}{{L}^{2}}{{\hbar}^{-2}}E_{{\text{F}}}^{k}\}}^{k\in \mathbb{E}}}$	&	Fermi levels of reservoirs	\tabularnewline
$\varepsilon _{{\text{F}}}^{\left\langle \text{QIY} \right\rangle}=2{{m}_{{\text{e}}}}{{L}^{2}}{{\hbar}^{-2}}E_{{\text{F}}}^{\left\langle \text{QIY} \right\rangle}$	&	Fermi level of network	\tabularnewline
\bottomrule
\end{tabularx}
\end{table}

The values and ranges of dimensionless geometric parameters of the network are determined by their dimensional analogues (Table~\ref{tab:2}), which are expressed using a characteristic length $L = {B^= }$. Let us choose $L = 10\ \text{nm}$. Then, according to equalities (\ref{Eq149}) and table~\ref{tab:2}, in terms of the dimensional parameters of the network we have
\begin{equation}	\label{Eq157}
	{B^= } = {B^{\langle \text{Q} \rangle}} = {B^{\langle \text{I} \rangle}} = {B^{\langle \text{Y} \rangle}} = 10\ \text{nm}, \ \ {R^{\langle \text{Q} \rangle}} = {R^{\langle \text{Y} \rangle}} = 5\ \text{nm}
\end{equation}
These values are much larger than the lattice constant of III-V semiconductors \cite{Bib030}. Therefore, when modeling such a nanostructure, one can use the zone theory of a solid. Due to small sizes, the effects of size-quantization here will be especially significant.

Ballistic electron transport will dominate in a structure that is smaller than the free length of electrons ${\ell}_e$. It has been experimentally shown that for a two-dimensional electron gas formed in heterostructures based on III-V semiconductors at room temperatures ${{\ell}_e}\sim 100\ \text{nm}$ \cite{Bib002}.

As the maximum permissible intensity of electric field in the I-junctions of the network, we assume a value ${{10}^8}\ \text{V/m}$. This is an order of magnitude less than autoemission intensities: $\sim {{10}^9}\ \text{V/m}$ \cite{Bib031}. Typical for calculations in the work are bias voltages $\sim 10\ \text{mV}$ \cite{Bib003,Bib004}.

\subsubsection{S-matrices of junctions}	\label{sec:4_1_3}

\begin{equation}	\label{Eq158}
	\{\ {{\square}^{\langle \text{I} \rangle}}\ \mapsto \ \square \ |\ \square \ = \Psi,{{\epsilon}_{\bot}},...\}\quad \operatorname{in}\ \text{\ref{sec:4_1_3}}
\end{equation}

To calculate electric currents through the QIY-network in the framework of the Landauer--B\"uttiker formalism (\ref{Eq137}), it is necessary to find its S-matrix using NCF (\ref{Eq092}). One can do it based on S-matrices of network junctions (Fig.~\ref{fig:4}) taking into account the specificity of its structure (\ref{Eq156}). At the same time, we assume that local coordinate frame at junctions are related to each other by rotation and translation (Fig.~\ref{fig:5}).

We find the S-matrices of Q- and Y-junctions by direct numerical calculation using the ND-map (\ref{Eq063}), (\ref{Eq066}) method based on their triangulation. In comparison with DN-map, the ND-map method is more preferable, since in this case the unitary condition (\ref{Eq118}) is better met for the current scattering matrix $C_{++}$.

Let us find the S-matrix of the I-junction in analytical form using SBC. Structurally, the I-junction is a section of a two-dimensional electron waveguide with applied electric field (Fig.~\ref{fig:6}). Without loss of generality, consider the I-junction with the potential
\begin{equation}	\label{Eq159}
	V\left( x,y \right) = {e_0}{{\mathcal{E}}_{\bot}}y
\end{equation}

\begin{figure}[htb]\center
	\includegraphics{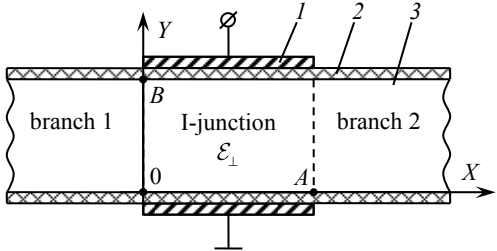}
	\caption{Structure of I-junction: \textit{1}~--- metal, \textit{2}~--- dielectric, \textit{3}~--- semiconductor; $A$~--- length of junction, $B$~--- width of junction, ${\mathcal{E}}_{\bot}$~--- intensity of transverse electric field.}\label{fig:6}
\end{figure}

The transverse linear potential (\ref{Eq159}) is a rough approximation to the real one. However, it implements the main feature of the control element: the presence of inhomogeneity in it in the form of a variable electric field. At the same time, in experiments for three-terminal ballistic junctions, it was found that their high-temperature electrical properties are qualitatively insensitive to the details of their structure \cite{Bib004}. Therefore, the potential (\ref{Eq159}) can be used to search for the features of the transport properties of the network that will be present in the experiment. In addition, in this case, one can organize the calculation of the S-matrix of the I-junction, exceeding in speed by 2-3 orders of magnitude the general algorithm proposed in the Appendix \ref{sec:B_2}.

We find the S-matrix of the I-junction based on the matrix $G^{\Diamond}$ (\ref{Eq056}), (\ref{Eq058}). Taking into account expressions (\ref{Eq013}), (\ref{Eq159}), (\ref{Eq156}) and $L = B$, we will use the following dimensionless quantities:
\begin{equation}	\label{Eq160}
	\begin{aligned}
		&a = A/B,\quad &&b = 1,	\\
		&{\epsilon}_{\bot} := 2{m_{\text{e}}}{B^3}{e_0}{{\hbar}^{-2}}{{\mathcal{E}}_{\bot}},\quad &&\varepsilon = 2{m_{\text{e}}}{B^2}{{\hbar}^{-2}}E
	\end{aligned}
\end{equation}
Then the electron scattering problem (\ref{Eq057}) in the I-junction $\Omega = (0,a)\times (0,1)$ when $\{{\beta}^k = (0,1)\}^{k = 1,2}$ takes the form
\begin{equation}	\label{Eq161}
	\left\{ \begin{aligned}
		&[ -\Delta +{\epsilon}_{\bot}y ]\Psi (x,y) = \varepsilon \Psi (x,y), &&\! \{x,y\} \! \in (0,a) \!\times\! (0,1) \\
		&\Psi (x,0) = \Psi (x,1) = 0, &&\! x \in (0,a) \\
		&[K + i{\partial}_1]W\Psi(0,y) = 2K{\psi}^{\triangleleft}(0,y), &&\! y\in (0,1)
	\end{aligned} \right.
\end{equation}
\begin{prop}	\label{prop:13}
 The operator $G^{\Diamond}$ for the scattering problem in the I-junction (\ref{Eq161}) is written as
\begin{equation}	\label{Eq162}
	G_{nm}^{\Diamond}\left( x \right) = \langle h_n^= | {h_m} \rangle \left[ \begin{matrix}
	{{\left( -1 \right)}^{m+1}}g_m^{\Diamond 1}\left( -x \right) & {{\left( -1 \right)}^{m+1}}g_m^{\Diamond 2}\left( -x \right) \\
		g_m^{\Diamond 1}\left( x+a \right) & g_m^{\Diamond 2}\left( x+a \right) \\
		\end{matrix} \right]
\end{equation}
\begin{equation}	\label{Eq163}
	\begin{aligned}
		{h_m}(y) =& c_{\bot m}^1\operatorname{Ai}\left( [{\epsilon}_{\bot}y-{\lambda}_m]/\epsilon _{\bot}^{2/3} \right)	\\
		&+ c_{\bot m}^2\operatorname{Bi}\left( [{\epsilon}_{\bot}y-{\lambda}_m]/\epsilon _{\bot}^{2/3} \right)
	\end{aligned}
\end{equation}
where coefficients $c_{\bot m}^1$ and $c_{\bot m}^2$ are the solution of the system of equations
\begin{equation}	\label{Eq164}
	\left[ \begin{matrix}
	\operatorname{Ai}\left( -{{\lambda}_m}/\epsilon _{\bot}^{2/3} \right) & \operatorname{Bi}\left( -{{\lambda}_m}/\epsilon _{\bot}^{2/3} \right) \\
		\operatorname{Ai}\left( [{{\epsilon}_{\bot}}-{{\lambda}_m}]/\epsilon _{\bot}^{2/3} \right) & \operatorname{Bi}\left( [{{\epsilon}_{\bot}}-{{\lambda}_m}]/\epsilon _{\bot}^{2/3} \right) \\
		\end{matrix} \right]\left[ \begin{matrix}
	c_{\bot m}^1 \\
		c_{\bot m}^2 \\
		\end{matrix} \right] = 0
\end{equation}
${\lambda}_m$ is solution of equation
\begin{equation}	\label{Eq165}
	\begin{aligned}
		\operatorname{Ai}&\left( -{{\lambda}_m}/\epsilon _{\bot}^{2/3} \right)\operatorname{Bi}\left( [{{\epsilon}_{\bot}}-{{\lambda}_m}]/\epsilon _{\bot}^{2/3} \right)	\\
		&= \operatorname{Bi}\left( -{{\lambda}_m}/\epsilon _{\bot}^{2/3} \right)\operatorname{Ai}\left( [{{\epsilon}_{\bot}}-{{\lambda}_m}]/\epsilon _{\bot}^{2/3} \right)
	\end{aligned}
\end{equation}
\vspace{15pt}
\begin{equation}	\label{Eq166}
	\begin{aligned}
		g_m^{\Diamond 1}(x) &= \exp \left( +i{{\kappa}_m}x \right),	\\
		g_m^{\Diamond 2}(x) &= \exp \left( -i{{\kappa}_m}x \right), \quad {\kappa}_m := \sqrt{\varepsilon -{\lambda}_m}
	\end{aligned}
\end{equation}
\end{prop}
\begin{proof}[\textsc{Proof~\ref{prop:13}}]
 We are looking for a function $\Psi $ in the form:
\begin{equation}	\label{Eq167}
	\Psi \left( x,y \right) = \sum\nolimits_m{{g_m}\left( x \right){h_m}\left( y \right)},\quad {g_m} := \langle {h_m} | \Psi \rangle 
\end{equation}
where $h_m$ is solution of a problem on eigenfunctions and eigenvalues
\begin{equation}	\label{Eq168}
	\left\{ \begin{aligned}
		&[ -\partial _y^2+{{\epsilon}_{\bot}}y]{h_m}(y) = {{\lambda}_m}{h_m}(y), && y\in (0,1) \\
		&{h_m}(y) = 0, && y\in \{0,1\}
	\end{aligned} \right.
\end{equation}
The solution of the problem (\ref{Eq168}) is a linear combination of Airy functions (\ref{Eq163}). Using the boundary conditions in the problem (\ref{Eq168}), from the expression (\ref{Eq163}) we obtain the system of equations (\ref{Eq164}). The system (\ref{Eq164}) allows to find unknown coefficients in the expression (\ref{Eq163}), and the condition of its solvability (equality to zero of the determinant of the matrix) will give the equation for ${\lambda}_m$ (\ref{Eq165}).

Substituting function (\ref{Eq167}) to the problem (\ref{Eq161}), we get
\begin{equation}	\label{Eq169}
	-\partial _1^2{g_m} = \left( \varepsilon -{{\lambda}_m} \right){g_m}
\end{equation}
The solution to the problem (\ref{Eq169}) is
\begin{equation}	\label{Eq170}
	{g_m}\left( x \right) = g_m^{\Diamond 1}\left( x \right)c_m^{\Diamond 1}+g_m^{\Diamond 2}\left( x \right)c_m^{\Diamond 2}
\end{equation}
where functions $g_m^{\Diamond 1}$ and $g_m^{\Diamond 2}$ are defined according to expressions (\ref{Eq166}). From expression (\ref{Eq167}) taking into account equalities (\ref{Eq170}) and (\ref{Eq153}) for the function in the junction one can obtain an expression of the form (\ref{Eq058}), where
\begin{equation}	\label{Eq171}
	G_{nm}^{\Diamond kl} := \langle h_n^= |{W^k}g_m^{\Diamond l}| {h_m} \rangle 
\end{equation}

The definition (\ref{Eq171}) is written as
\begin{equation}	\label{Eq172}
	G_{nm}^{\Diamond kl} = \left[ {W^k}g_m^{\Diamond l} \right]\langle h_n^= |{W^k}| {h_m} \rangle = \left[ {W^k}g_m^{\Diamond l} \right]\langle {{{\bar{W}}}^k}h_n^= | {h_m} \rangle 
\end{equation}
By choosing the universal location of the LCF in the network (Fig.~\ref{fig:5}) for the I-junction we get
\begin{equation}	\label{Eq173}
	{W^1}g_m^{\Diamond}\left( x \right) = g_m^{\Diamond}\left( -x \right),\quad {W^2}g_m^{\Diamond}\left( x \right) = g_m^{\Diamond}\left( x+a \right)
\end{equation}
Based on expressions (\ref{Eq172}) and (\ref{Eq173}) we write the operator $G^{\Diamond}$ in matrix form
\begin{equation}	\label{Eq174}
	\begin{aligned}
		G_{nm}^{\Diamond}&(x) =	\\
		& \left[ \begin{matrix}
			g_m^{\Diamond 1}\left( -x \right)\langle {{{\bar{W}}}^1}h_n^= | {h_m} \rangle & g_m^{\Diamond 2}\left( -x \right)\langle {{{\bar{W}}}^1}h_n^= | {h_m} \rangle \\
			g_m^{\Diamond 1}\left( x+a \right)\langle {{{\bar{W}}}^2}h_n^= | {h_m} \rangle & g_m^{\Diamond 2}\left( x+a \right)\langle {{{\bar{W}}}^2}h_n^= | {h_m} \rangle \\
		\end{matrix} \right]
	\end{aligned}
\end{equation}
When LCFs are related to each other by rotation and translation, transverse functions in the branches when transitioning to GCF (Fig.~\ref{fig:6}) are transformed as
\begin{equation}	\label{Eq175}
	{{\bar{W}}^1}h_n^= \left( y \right) = h_n^= \left( b-y \right),\quad {{\bar{W}}^2}h_n^= \left( y \right) = h_n^= \left( y \right)
\end{equation}
Taking into account equalities (\ref{Eq152}) and (\ref{Eq153}) from formulas (\ref{Eq175}) we get
\begin{equation}	\label{Eq176}
	{{\bar{W}}^1}h_n^= = {{\left( -1 \right)}^{n+1}}h_n^= 
\end{equation}
From expressions (\ref{Eq174})--(\ref{Eq176}) we have statement (\ref{Eq162}).	
\end{proof}

The formulas (\ref{Eq162}) and (\ref{Eq056}) allow a high speed calculation of the S-matrix of the I-junction with transverse potential (\ref{Eq161}).

\subsection{Network of four junctions}	\label{sec:4_2}

\subsubsection{Geometry and parameters}	\label{sec:4_2_1}

As an example, we model a smooth branched nanostructure using a QIY-network (Fig.~\ref{fig:7}) ${{\Omega}^{\langle \text{ex} \rangle}} := {{\Omega}^{[0,4,5]}}$. Its structure (\ref{Eq155}) is written in the form
\begin{equation}	\label{Eq177}
	\mathcal{Q} = \{\{2\}\},\quad \mathcal{I} = \{\{1,3\}\},\quad \mathcal{Y} = \{\{0,1,2\},\{3,4,5\}\}
\end{equation}
We consider a system based on a two-dimensional electron gas formed in a heterostructure based on GaAs (${m_{\text{e}}} = 0.063{m_0}$) \cite{Bib030}. In addition to the dimensional parameters of the network (\ref{Eq157}), we assume for the nanostructure
\begin{equation}	\label{Eq178}
	\begin{aligned}
		&A^1 = 7\ \text{nm},\quad A^2 = 5\ \text{nm},\quad A^3 = 10\ \text{nm},	\\
		&A^{[1,3]} = A^{\langle \text{I} \rangle} = 15\ \text{nm}
	\end{aligned}
\end{equation}
According to equalities (\ref{Eq157}) and (\ref{Eq178}) linear sizes of the nanostructure $<60\ \text{nm}$ (Fig.~\ref{fig:7}), and the free length of electrons ${\ell}_e$ at room temperature for it $\sim 100\ \text{nm}$ \cite{Bib002}. Therefore, up to room temperatures, ballistic electron transport dominates in the nanostructure, and the proposed scheme can be used to calculate it.

According to the assumed width of the branches (\ref{Eq149}) and (\ref{Eq157}), taking into account the equalities (\ref{Eq152}) and (\ref{Eq127}) for the dimensional energies of the channels we have:
\begin{equation}	\label{Eq179}
	E_{\bot 1}^= = 0.0\text{6}0\ \text{eV},\quad E_{\bot 2}^= = 0.\text{239}\ \text{eV},\quad E_{\bot 3}^= = 0.\text{537}\ \text{eV}
\end{equation}
The distance between energies is more than thermal broadening of levels ${k_0}T$ at room temperature: 0.026 eV, which indicates their good resolution.

\begin{figure}[htb]\center
	\includegraphics{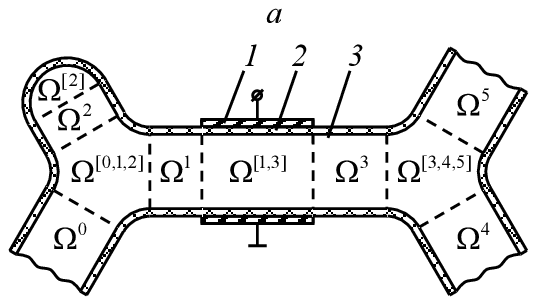}
	\includegraphics{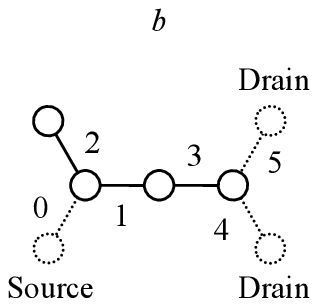}
	\caption{Smooth branched nanostructure in terms of four-junctions QIY-network. \textit{a}~--- construction: \textit{1}~--- metal, \textit{2}~--- dielectric, \textit{3}~--- semiconductor; dashed lines are boundaries between junctions and branches. \textit{b}~--- scheme: solid lines are internal junctions and branches, dotted lines are external junctions (source and drains) and branches.}\label{fig:7}
\end{figure}

For the range of values of the Fermi level of the nanostructure we choose $E_{\text{F}}^{\langle \text{ex} \rangle}\in (0,10E_{\bot 1}^= )$. Such a range, on the one hand, allows to consider boundary cases $E_{\text{F}}^{\langle \text{ex} \rangle} = E_{\bot 1}^= $, $E_{\text{F}}^{\langle \text{ex} \rangle} = E_{\bot 2}^= = 4E_{\bot 1}^= $, and $E_{\text{F}}^{\langle \text{ex} \rangle} = E_{\bot 3}^= = 9E_{\bot 1}^= $, on the other hand, excludes from consideration higher channels that are difficult to observe in a real experiment. Fermi levels in reservoirs (external junctions) are determined based on external voltages applied in push-pull fashion:
\begin{equation}	\label{Eq180}
	\begin{aligned}
		&E_{\text{F}}^0 = E_{\text{F}}^{\langle \text{ex} \rangle}-{e_0}{U_{\parallel}}/2,	\\
		&E_{\text{F}}^4 = E_{\text{F}}^{\langle \text{ex} \rangle}+{e_0}{U_{\parallel}}/2,\quad E_{\text{F}}^5 = E_{\text{F}}^{\langle \text{ex} \rangle}+{e_0}{U_{\parallel}}/2
	\end{aligned}
\end{equation}
where ${U_{\parallel}} := {U^{04}} = {U^{05}}$ is the bias voltage (\ref{Eq147}). When solving problems about electron transport in low-dimensional systems, typical values of bias voltages $\sim 10\ \text{mV}$ \cite{Bib003,Bib004}. We choose a range for calculations ${U_{\parallel}} = 0\div 100\ \text{mV}$.

We assume that the temperatures of all reservoirs are the same (\ref{Eq144}). We consider the range ${T^= } = 0\div 300\ \text{K}$. Finally, for the intensity of the control electric field in the I-junction, we choose a value $\mathcal{E}_{\bot}^{[1,3]} = \mathcal{E}_{\bot}^{\langle \text{I} \rangle} = -{{10}^7}\ \text{V/m}$.

Defining the dimensionless bias voltage as
\begin{equation}	\label{Eq181}
	{{\upsilon}_{\parallel}} := 2{m_{\text{e}}}{L^2}{{\hbar}^{-2}}{e_0}{U_{\parallel}}
\end{equation}
taking into account table~\ref{tab:2}, based on the values and ranges of the dimensional parameters we obtain the following values and ranges of dimensionless parameters of the nanostructure:
\begin{equation}	\label{Eq182}
	\begin{aligned}
		&{a^1} = 0.7, &&{a^2} = 0.5, &&{a^3} = 1, &&{a^{\langle \text{I} \rangle}} = 1.5, \\
		&{b^= } = 1, &&{b^{\langle \text{I} \rangle}} = 1, &&{b^{\langle \text{Y} \rangle}} = 1, &&{r^{\langle \text{Y} \rangle}} = 0.5,\quad \epsilon _{\bot}^{\langle \text{I} \rangle} = 16.5
	\end{aligned}
\end{equation}
\vspace{-20pt}
\begin{equation}	\label{Eq183}
	\varepsilon _{\text{F}}^{\langle \text{ex} \rangle}\in (0,10{{\pi}^2}),\quad {{\mu}^= }\in (0,4.27),\quad {{\upsilon}_{\parallel}}\in (0,\text{16}\text{.5})
\end{equation}
\pagebreak

\subsubsection{Expression for S-matrix}	\label{sec:4_2_2}

According to the approach of sections \ref{sec:2} and \ref{sec:3}, the electrical properties of the nanostructure (Fig.~\ref{fig:7}) are calculated based on its S-matrix. The S-matrix of the nanostructure is calculated using NCF (\ref{Eq092}). Let us write this procedure step by step.

For a QIY network (\ref{Eq155}) of structure (\ref{Eq177}), the procedure for combining of junctions' S-matrices according to NCF (\ref{Eq092}) takes the form
\begin{equation}	\label{Eq184}
	\begin{aligned}
		{S^{[\mathbb{E}]}} =& {S^{[2]}} \!\circledast\! {S^{[1,3]}} \!\circledast\! {S^{[0,1,2]}} \!\circledast\! {S^{[3,4,5]}}	\\
		=& {S^{[2,1,3]}} \!\circledast\! {S^{[0,1,2]}} \!\circledast\! {S^{[3,4,5]}} \!=\! {S^{[0,3]}} \!\circledast\! {S^{[3,4,5]}} \!=\! {S^{[0,4,5]}}
	\end{aligned}
\end{equation}
From the procedure (\ref{Eq184}) it follows that a tuple of numbers of external branches $\mathbb{E} = \{0,4,5\}$. The procedure (\ref{Eq184}) consists of three steps. We write them using expressions (\ref{Eq089}) and (\ref{Eq088}).

1. ${S^{[2,1,3]}} = {S^{[2]}}\circledast {S^{[1,3]}}$, $\mathbb{J} = \{2\}$, $\mathbb{K} = \varnothing $, $\mathbb{L} = \{1,3\}$; (\ref{Eq071})~$\quad \Rightarrow \quad$
\begin{equation}	\label{Eq185}
	\begin{aligned}
		{S^{[2,1,3]}} &=
		\left[ \begin{matrix}
			{S^{[2]22}} & {O^{2\{1,3\}}} \\
			{O^{\{1,3\}2}} & {S^{[1,3]\{1,3\}\{1,3\}}} \\
		\end{matrix} \right] 	\\
		&= \left[ \begin{matrix}
			{S^{[2]22}} & {O^{21}} & {O^{23}} \\
			{O^{12}} & {S^{[1,3]11}} & {S^{[1,3]13}} \\
			{O^{32}} & {S^{[1,3]31}} & {S^{[1,3]33}} \\
		\end{matrix} \right]
	\end{aligned}
\end{equation}

2. ${S^{[0,3]}} = {S^{[2,1,3]}}\circledast {S^{[0,1,2]}}$; $\mathbb{J} = \{3\}$, $\mathbb{K} = \{2,1\}$, $\mathbb{L} = \{0\}$; (\ref{Eq071}) $\quad \Rightarrow \quad$

\begin{strip}
\vspace{5pt}
$\begin{aligned}
	{S^{[3,0]}} =&
	\left[ \begin{matrix}
		{S^{[3,2,1]33}} & {O^{30}} \\
		{O^{03}} & {S^{[2,1,0]00}} \\
	\end{matrix} \right]+
	\left[ \begin{matrix}
		{S^{[3,2,1]32}} & {S^{[3,2,1]31}} & {O^{32}} & {O^{31}} \\
		{O^{02}} & {O^{01}} & {S^{[2,1,0]02}} & {S^{[2,1,0]01}} \\
	\end{matrix} \right]	\\
	&\times {{\left[ \begin{matrix}
		-{S^{[3,2,1]\{2,1\}\{2,1\}}}  & {U^{[3,2,1]\{2,1\}\{2,1\}}}\exp \left( -i{K^{\{2,1\}\{2,1\}}}{A^{\{2,1\}\{2,1\}}} \right) \\
		{U^{[2,1,0]\{2,1\}\{2,1\}}}\exp \left( -i{K^{\{2,1\}\{2,1\}}}{A^{\{2,1\}\{2,1\}}} \right)  & -{S^{[2,1,0]\{2,1\}\{2,1\}}} \\
	\end{matrix} \right]}^{-1}}	\\
	&\times \left[ \begin{matrix}
		{S^{[3,2,1]23}} & {O^{20}} \\
		{S^{[3,2,1]13}} & {O^{10}} \\
		{O^{23}} & {S^{[2,1,0]20}} \\
		{O^{13}} & {S^{[2,1,0]10}} \\
	\end{matrix} \right]
\end{aligned}$
\vspace{5pt}

In this expression, objects with a tuple index $\{2,1\}$ can also be expanded using example (\ref{Eq002}). Here we do not do it for short.

3. ${S^{[0,4,5]}} = {S^{[0,3]}}\circledast {S^{[3,4,5]}}$; $\mathbb{J} = \{0\}$, $\mathbb{K} = \{3\}$, $\mathbb{L} = \{4,5\}$; (\ref{Eq071}) $\quad \Rightarrow \quad$

\vspace{10pt}
$\begin{aligned}
	{S^{[0,4,5]}} =&
	\left[ \begin{matrix}
		{S^{[0,3]00}} & {O^{04}} & {O^{05}} \\
		{O^{40}} & {S^{[3,4,5]44}} & {S^{[3,4,5]45}} \\
		{O^{50}} & {S^{[3,4,5]54}} & {S^{[3,4,5]55}} \\
	\end{matrix} \right] +
	\left[ \begin{matrix}
		{S^{[0,3]03}} & {O^{03}} \\
		{O^{43}} & {S^{[3,4,5]43}} \\
		{O^{53}} & {S^{[3,4,5]53}} \\
	\end{matrix} \right] \\
	&\times {{\left[ \begin{matrix}
		-{S^{[0,3]33}} & {U^{[0,3]33}}\exp \left( -i{K^{33}}{A^{33}} \right) \\
		{U^{[3,4,5]33}}\exp \left( -i{K^{33}}{A^{33}} \right) & -{S^{[3,4,5]33}} \\
	\end{matrix} \right]}^{-1}}
	\left[ \begin{matrix}
		{S^{[0,3]30}} & {O^{34}} & {O^{35}} \\
		{O^{30}} & {S^{[3,4,5]34}} & {S^{[3,4,5]35}} \\
	\end{matrix} \right]
\end{aligned}$
\vspace{10pt}
\end{strip}

Using this procedure, one can find an extended scattering matrix $S^{[0,4,5]}$ of nanostructure (Fig.~\ref{fig:7}). In section \ref{sec:4_2_2}, during matrix operations we consider 6 channels, since their further increase has little effect on the electrical properties calculated in section \ref{sec:4_2_3}. Therefore all submatrices in expressions of section \ref{sec:4_2_2} have the size $6\times 6$: ${S^{[2]22}} = {{\{S_{mn}^{[2]22}\}}_{m,n = 1..6}}$, ${O^{30}} = {{\{O_{mn}^{30}\}}_{m,n = 1..6}}$, ${K^{33}} = {{\{K_{mn}^{33}\}}_{m,n = 1..6}}$ etc.

In our numerical calculation, the written procedure is implemented automatically. It can be programmed for an arbitrary quantum network due to the agreements and notation (section \ref{sec:2_1}) used in writing NCF (\ref{Eq092}).

\subsubsection{Electrical properties}	\label{sec:4_2_3}

Based on the S-matrix of the nanostructure (Fig.~\ref{fig:7}) let us study its electrical properties (section \ref{sec:3_2_2}). Calculations showed that in section \ref{sec:4_2_3} in the sums (\ref{Eq137}) and (\ref{Eq146}) it is enough to take into account only the first 4 of the channel, and one can take the energy of the 5-th channel for an upper limit of integration:
\begin{equation}	\label{Eq186}
	\sum\nolimits_{mn}^{}{\int_{-\infty}^{+\infty}{d\varepsilon \{...\}}}\mapsto \sum\nolimits_{m,n = 1..4}^{}{\int_{-\infty}^{+\infty}{d\varepsilon [\varepsilon <\lambda _5^= ]\{...\}}}
\end{equation}
At the same time, the lower limit of integration has the values of energies of the first 4 channels ${\{\lambda _m^= \}}_{m = 1..4}$.

We consider the approximate conductivity of the nanostructure (Fig.~\ref{fig:7}) at low bias voltages (\ref{Eq146}):
\begin{equation}	\label{Eq187}
	{{\tilde{\sigma}}^{\langle \text{ex} \rangle}} = {{\tilde{\sigma}}^{[0,4,5]}} = \left[ \begin{matrix}
	{{{\tilde{\sigma}}}^{[0,4,5]00}} & {{{\tilde{\sigma}}}^{[0,4,5]04}} & {{{\tilde{\sigma}}}^{[0,4,5]05}} \\
		{{{\tilde{\sigma}}}^{[0,4,5]40}} & {{{\tilde{\sigma}}}^{[0,4,5]44}} & {{{\tilde{\sigma}}}^{[0,4,5]45}} \\
		{{{\tilde{\sigma}}}^{[0,4,5]50}} & {{{\tilde{\sigma}}}^{[0,4,5]54}} & {{{\tilde{\sigma}}}^{[0,4,5]55}} \\
		\end{matrix} \right]
\end{equation}
Taking into account the problem definition (\ref{Eq180}) here we are only interested in non-diagonal components ${\tilde{\sigma}}^{[0,4,5]40}$ and ${\tilde{\sigma}}^{[0,4,5]50}$.

Let us plot graphs of components ${\tilde{\sigma}}^{[0,4,5]40}$ and ${\tilde{\sigma}}^{[0,4,5]50}$ as functions of the dimensionless Fermi level of the structure $\varepsilon _{\text{F}}^{\langle \text{ex} \rangle}$ and its temperature $T^= $ (Fig.~\ref{fig:8}). Both graphs show that at ${T^= } = 0\ \text{K}$ elements ${\tilde{\sigma}}^{[0,4,5]40}$ and ${\tilde{\sigma}}^{[0,4,5]50}$ have many local extremes. According to expression (\ref{Eq148}), such conductivity behavior reveals the complex resonant structure of the system (Fig.~\ref{fig:7}). With rising temperatures, extremes are smoothed out. This is due to the fact that according to the equations (\ref{Eq140}) and (\ref{Eq144}), the parameter ${\mu}^= $ grows, and in the expression (\ref{Eq146}), the peak of the integrand $\{{{\cosh}^{-2}}([\varepsilon _{\text{F}}^{\langle \text{ex} \rangle}-\varepsilon ]/[2{{\mu}^= }])|\varepsilon \in \mathbb{R}\}$ widens. Physically, this means that with an increase in temperature, the resonant structure of the system (Fig.~\ref{fig:7}) on conduction graphs (Fig.~\ref{fig:8}) smoothed due to a blur in the energies of electrons involved in conduction.

\begin{figure}[htb]\flushleft
	\hspace{30pt}\includegraphics{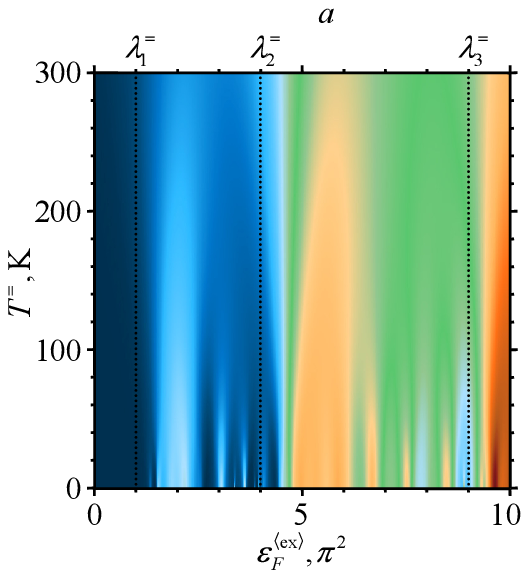}	\\
	\vspace{10pt}\hspace{30pt}\includegraphics{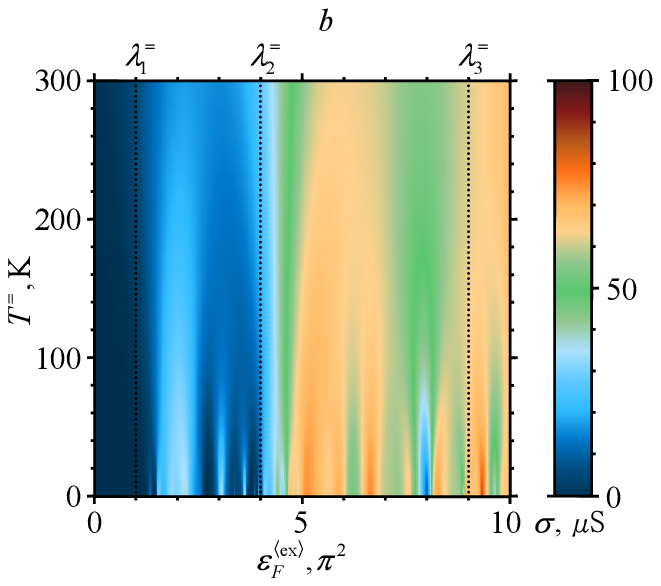}
	\caption{Elements of approximate conductivity matrix (\ref{Eq187}) of smooth branched nanostructure (Fig.~\ref{fig:7}) with dimensional parameters (\ref{Eq157}) and (\ref{Eq178}): \textit{a}~--- ${\tilde{\sigma}}^{[0,4,5]40}$, \textit{b}~--- ${\tilde{\sigma}}^{[0,4,5]50}$.}\label{fig:8}
\end{figure}

In graphs (Fig.~\ref{fig:8}) it is also seen that as the number of “open” channels increases, the conductivity grows. However, this growth is not monotonous. Note also that at ${T^= }\ne 0\ \text{K}$ the concept of “open” channels is conditional. This is due to a blur in the energies of electrons involved in conduction. In particular, therefore, at ${T^= }\ne 0\ \text{K}$ a non-zero conductivity is observed in the “zero-channel” mode ($\varepsilon _{\text{F}}^{\langle \text{ex} \rangle}\le \lambda _1^= $). This effect is clearly visible in two-dimensional graphs of conductivity (Fig.~\ref{fig:9}). They also show the resonant structure of the system (Fig.~\ref{fig:7}), which is smoothed with increasing temperature.

\begin{figure}[htb]\center
	\includegraphics{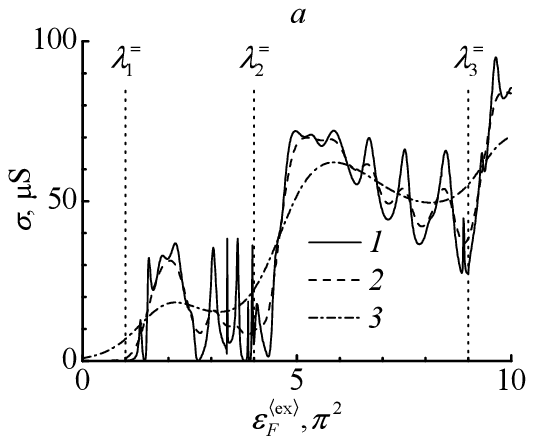}	\\
	\vspace{10pt}\includegraphics{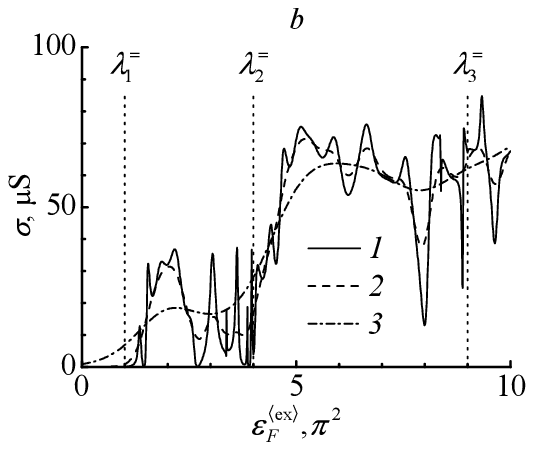}
	\caption{Elements of approximate conductivity matrix (\ref{Eq187}) of smooth branched nanostructure (Fig.~\ref{fig:7}) with dimensional parameters (\ref{Eq157}) and (\ref{Eq178}): \textit{a}~--- ${\tilde{\sigma}}^{[0,4,5]40}$, \textit{b}~--- ${\tilde{\sigma}}^{[0,4,5]50}$; \textit{1}~--- ${T^= } = 0\ \text{K}$, \textit{2}~--- ${T^= } = 77\ \text{K}$, \textit{3}~--- ${T^= } = 300\ \text{K}$.}\label{fig:9}
\end{figure}

Graphs of approximate conductivity (Fig.~\ref{fig:8} and Fig.~\ref{fig:9}) describe the electrical properties of the nanostructure (Fig.~\ref{fig:7}) only at low bias voltages. We plot its volt-ampere characteristics (VAC) for the four positions of the dimensionless Fermi level of structure $\varepsilon _{\text{F}}^{\langle \text{ex} \rangle}$ corresponding to conditionally different modes (Fig.~\ref{fig:10}): “one-channel” ($\varepsilon _{\text{F}}^{\langle \text{ex} \rangle} = (\lambda _1^= +\lambda _2^= )/2 = 2.5{{\pi}^2}$), second boundary ($\varepsilon _{\text{F}}^{\langle \text{ex} \rangle} = \lambda _2^= = 4{{\pi}^2}$), “two-channel” ($\varepsilon _{\text{F}}^{\langle \text{ex} \rangle} = (\lambda _2^= +\lambda _3^= )/2 = 6.5{{\pi}^2}$), third boundary ($\varepsilon _{\text{F}}^{\langle \text{ex} \rangle} = \lambda _3^= = 9{{\pi}^2}$). We do not consider the “zero-channel” and the first boundary modes ($\varepsilon _{\text{F}}^{\langle \text{ex} \rangle}\le \lambda _1^= $), since the currents in them are small ($<1\ \text{}\!\!\mu\!\!\text{}\ \text{A}$).

\begin{figure}[htb]\center
	\includegraphics{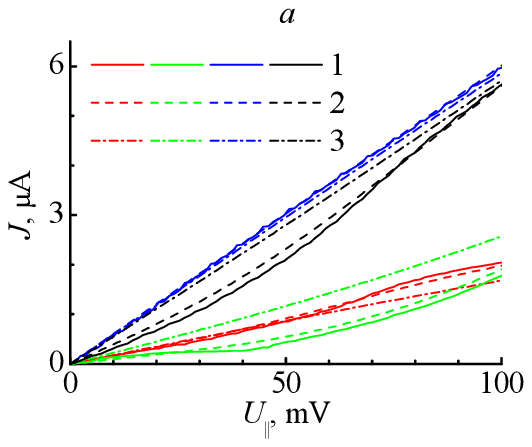}	\\
	\vspace{10pt}\includegraphics{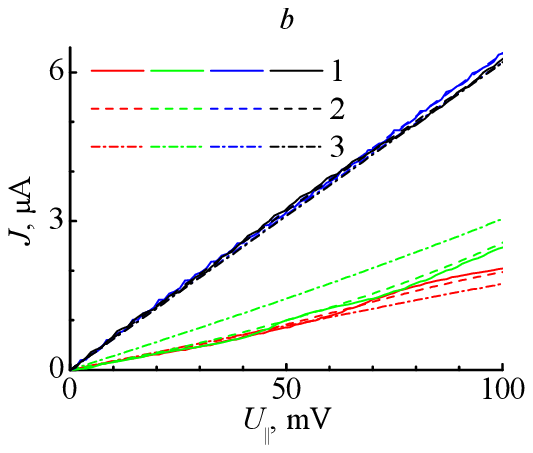}
	\caption{Volt-ampere characteristics of smooth branched nanostructure (Fig.~\ref{fig:7}) with dimensional parameters (\ref{Eq157}) and (\ref{Eq178}): \textit{a}~--- $J^4$, \textit{b}~--- $J^5$; \textit{1}~--- ${T^= } = 0\ \text{K}$, \textit{2}~--- ${T^= } = 77\ \text{K}$, \textit{3}~--- ${T^= } = 300\ \text{K}$; red lines~--- $\varepsilon _{\text{F}}^{\langle \text{ex} \rangle} = (\lambda _1^= +\lambda _2^= )/2$, green lines~--- $\varepsilon _{\text{F}}^{\langle \text{ex} \rangle} = \lambda _2^= $, blue lines~--- $\varepsilon _{\text{F}}^{\langle \text{ex} \rangle} = (\lambda _2^= +\lambda _3^= )/2$, black lines~--- $\varepsilon _{\text{F}}^{\langle \text{ex} \rangle} = \lambda _3^= $.}\label{fig:10}
\end{figure}

Before analysis of the VAC of nanostructure (Fig.~\ref{fig:10}) we estimate by formula (\ref{Eq136}) the error of electric currents calculation. In all cases in considered range of voltages $|{J^0}+{J^4}+{J^5}|<5\cdot {{10}^{-11}}\text{A}$, that can be assumes as an acceptable value.

On VAC of nanostructures (Fig.~\ref{fig:10}) we see the grouping of curves in the “one-channel” and second boundary modes, in the “two-channel” and third boundary modes. Calculations have shown that this is not a threshold phenomenon: with an increase of Fermi level, curves, bending, gradually shift. Such an offset is not monotonous: at the same time, the curve can both rise and fall. This means that with an increase of Fermi levels, electric current, having a global upward trend, can locally decrease. Such behavior of VAC is observed at all temperatures within the range ${T^= } = 0\div 300\ \text{K}$.

The dependency of VAC of the nanostructure (Fig.~\ref{fig:10}) from temperature has the following features. At ${T^= } = 0\ \text{K}$ curves have many small corners that reveal the complex resonant structure of the system (Fig.~\ref{fig:7}). At the same time, the curves are strongly deformed. As the temperature increases, the corners smooth out and the VAC approaches linear. This suggests that as the temperature increases, the applicability of the expression for currents at low bias voltages (\ref{Eq145}) and, therefore, the expression for approximate conductivity (\ref{Eq146}) increases. At the same time, it becomes apparent that the expression for conductivity based only on the probability of transmission (\ref{Eq148}) (low temperatures and small bias voltages) has a very small domain of applicability.

Finally, we note that in section \ref{sec:4_2_3}, the discussion of the electrical properties of the nanostructure at ${T^= }\approx 0\ \text{K}$ is given purely to explain their general trends. To obtain in this case adequate results for the real system, an accurate calculation of the potential in it is necessary. This is due to the fact that even for the simplest systems, only their high-temperature electrical properties are qualitatively insensitive to the details of their structure \cite{Bib004}.

Thus, from the study of approximate conductivity (Fig.~\ref{fig:8} and Fig.~\ref{fig:9}) and VAC (Fig.~\ref{fig:10}) of four-junctions nanostructure (Fig.~\ref{fig:7}) in section \ref{sec:4_2_3}, we have the following conclusion. To correctly describe the electrical properties of a nanostructure only on the basis of its S-matrix, two conditions must be met: extremely low temperatures (${T^= }\approx 0\ \text{K}$) and high-precision calculation of the potential in it.

\section{Conclusion}	\label{sec:5}

In this work, we considered the calculation of electron transport in branched semiconductor nanostructures. For this purpose, we developed and applied a model of a quantum network. The proposed calculation scheme is based on a special notation system that simplifies numerical implementation. In the calculation of the S-matrix of the junction, preference is given to the method of scattering boundary conditions. For them, a clear integro-differential form was proposed, and the method of calculating the S-matrix of a junction in an explicit form has also been developed. To calculate the S-matrix of the entire network in terms of its junctions' S-matrices, a network combining formula was proposed. It explicitly takes into account all connections and is a universal calculation algorithm suitable for networks of an arbitrary structure. The calculation of electrical currents through the network performed using the Landauer--B\"uttiker formalism, adapted to the designation system used.

We demonstrated the calculation scheme proposed in the work by modeling a nanostructure based on two-dimensional electron gas. For this purpose, we proposed a model of a network of smooth junctions with one, two and three adjacent branches. We calculated the electrical properties of such a network (using the example of GaAs), consisting of four junctions, depending on the temperature. We have found that the resonance structure of such a system manifests itself in its electrical properties only at extremely low temperatures ($\approx 0\ \text{K}$). Under such conditions, it is possible to reduce their calculation to only the S-matrix of the system with a high-precision calculation of the potential in it. With increasing temperature, electron statistics should be taken into account. At the same time, in electrical properties, a significant smoothing of the resonant structure of the system occurs, and its effect on transport properties is averaged.

The prospects for the theme of this work are follows:

\begin{itemize}
	\item application of scattering boundary conditions for self-consistent calculation of currents;

	\item generalization in the case of semiconductors with the relativistic dispersion law of charge carriers, taking into account spin-orbital interaction;

	\item modeling of semiconductor nanostructures with predefined transport properties.
\end{itemize}

\appendix
\normalsize

\section{Resonant and bound states}	\label{sec:A}

\begin{equation}	\label{Eq188}
	\{\ {{\square}^{[\mathbb{A}]}}\ \mapsto \ \square \ |\ \square \ = S,\psi,...\}\quad \operatorname{in}\ \text{\ref{sec:A}}
\end{equation}

In section \ref{sec:2_3_1}, we justified the clearness of the scattering boundary conditions, relying on a problem about resonant and bound states. In Appendix \ref{sec:A} we will show how it is formulated in terms of SBC.

Resonant and antiresonant, bound and antibond states are determined based on the poles and zeros of the S-matrix \cite{Bib032} (Table~\ref{tab:3}).

\begin{table}
\caption{\label{tab:3}Poles and zeros of S-matrix}
\begin{tabularx}{\linewidth}{c*{3}{>{\centering}X}c}
\toprule
	\multirow{3}*{symbol}	&\multicolumn{3}{c}{name}	\tabularnewline
\noalign{\smallskip}\cline{2-4}\noalign{\smallskip}
	&	$\forall \operatorname{Im}{{\varepsilon}^{\pm}}$	&	$\operatorname{Im}{{\varepsilon}^{\pm}}\ne 0$	&	$\operatorname{Im}{{\varepsilon}^{\pm}}=0$	\tabularnewline
\midrule
${{\varepsilon}^{+}}$	&	pole of S-matrix	&	energy of resonant state	&	energy of bound state	\tabularnewline
${{\varepsilon}^{-}}$	&	zero of S-matrix	&	energy of \\ antiresonant~state	&	energy of antibound state	\tabularnewline
\bottomrule
\end{tabularx}
\end{table}

\vspace{-10pt}
\begin{equation}	\label{Eq189}
	{{\varepsilon}^+}:\quad | S({{\varepsilon}^+}) | = \infty 
\end{equation}
\vspace{-20pt}
\begin{equation}	\label{Eq190}
	{{\varepsilon}^-}:\quad | S({{\varepsilon}^-}) | = 0
\end{equation}
From definitions (\ref{Eq036})--(\ref{Eq039}) we get
\begin{equation}	\label{Eq191}
	{{\{{{\psi}^{\triangleleft k}} = 0\}}^k}\quad \Leftrightarrow \quad {S^{-1}}{c^{\triangleright}} = 0
\end{equation}
\begin{equation}	\label{Eq192}
	{{\{{{\psi}^{\triangleright k}} = 0\}}^k}\quad \Leftrightarrow \quad S{c^{\triangleleft}} = 0
\end{equation}
Due to the criterion for the existence of non-zero solutions of homogeneous system of linear algebraic equations from statements (\ref{Eq191}) and (\ref{Eq192}) in GCF, it follows
\begin{equation}	\label{Eq193}
	| S({{\varepsilon}^+}) | = \infty \quad \Leftrightarrow \quad {{\Psi}^{\triangleleft}} = 0
\end{equation}
\begin{equation}	\label{Eq194}
	| S({{\varepsilon}^-}) | = 0 \hspace{12pt} \Leftrightarrow \quad {{\Psi}^{\triangleright}} = 0
\end{equation}
Then according to statements (\ref{Eq193}) and (\ref{Eq194}) taking into account expressions (\ref{Eq016}), (\ref{Eq049}) and (\ref{Eq050}) we have
\begin{equation}	\label{Eq195}
	\left\{ \begin{aligned}
		&[-\Delta +\upsilon ]{\Psi}^{\pm} = {\varepsilon}^{\pm}{\Psi}^{\pm} && \operatorname{in}\ \ \Omega \\
		&{\Psi}^{\pm} = 0 && \operatorname{on}\ \partial \Omega \backslash \Gamma \\
		&[K^{\pm} \pm i{\partial}_n]{\Psi}^{\pm} = 0 && \operatorname{on}\ \Gamma
	\end{aligned} \right.
\end{equation}
It means that the poles of the S-matrix are the eigenvalues of the problem with zero incoming SBC, the zeroes of the S-matrix are the eigenvalues of the problem with zero outgoing SBC.

The problem (\ref{Eq195}) about resonant and bound states ${\Psi}^+$ describes a “emitter”: a junction can emit electrons in the absence of incident ones on it. The problem (\ref{Eq195}) about antiresonant and antibound states ${\Psi}^-$ describes an “absorber”: a junction can absorb electrons in the absence of scattered ones by it. Here, the terms “incident” and “scattered” are used conventionally for waves in branch-channels with complex ${\kappa}^{\pm}$ (\ref{Eq028}).

The type of ${\kappa}^{\pm}$ for the waves in the branch-channels of problems (\ref{Eq195}) depends on the type of eigenvalue ${\varepsilon}^{\pm}$ corresponding to it. In turn, the type of ${\varepsilon}^{\pm}$ is determined by the amount of isolation of the junction. Let us explain it as follows. The energy structure of the channels in the branches $\{\lambda _m^k\}_m^k$ creates potential barriers of heights $\{{\min}_m \lambda _m^k\}^k$ on the boundaries with the junction. The barrier with the lowest height $\min _m^k\lambda _m^k$ provides the best connection with the environment. Depending on the type of junction and the value $\min _m^k\lambda _m^k$, among the eigenvalues of the problems (\ref{Eq195}) there may be real numbers ${{\varepsilon}^{\pm}}\le \min _m^k\lambda _m^k$~--- energies of bound and antibound states. Other eigenvalues with $\operatorname{Im}{{\varepsilon}^{\pm}}\ne 0$ are interpreted as energies of resonant and antiresonant states (Table~\ref{tab:3}).

The $\min _m^k\lambda _m^k$-dependent isolation of a junction is easily interpreted in terms of SBC. With the growth of channels' energies (for example, due to a decrease in the cross section of the branches), the isolation of the junction grows. In particular, when $\min _m^k\lambda _m^k\to \infty $ the problem (\ref{Eq195}) takes the form
\begin{equation}	\label{Eq196}
	\left\{ \begin{aligned}
		&[-\Delta +\upsilon ]{\Psi}^{\times} = {\varepsilon}^{\times}{\Psi}^{\times} && \operatorname{in}\ \ \Omega \\
		&{\Psi}^{\times} = 0 && \operatorname{on}\ \partial \Omega
	\end{aligned} \right.
\end{equation}
Problem (\ref{Eq196}) describes a completely isolated junction. This limit transition is equivalent to the case when, in comparison with the operator $K^{\pm}$, the operator ${\partial}_n$ becomes negligible. Therefore, in terms of SBC, the physical meaning of the operator $K$ is isolation of junction.

\section{Scattering boundary conditions in quantum wire}	\label{sec:B}

\begin{equation}	\label{Eq197}
	\{\ {{\square}^{\langle \text{I} \rangle}}\ \mapsto \ \square \ |\ \square \ = \Psi,\upsilon,...\}\quad \operatorname{in}\ \text{\ref{sec:B}}
\end{equation}

In section \ref{sec:2_3_2}, we indicated that in some cases, using the SBC, the S-matrix of the quantum network junction can be written in terms of the operator $G^{\Diamond}$ (\ref{Eq056}). This is possible at least in the problem of an electron scattering in a quantum wire with a compact potential (a consequence of a defect, an external electric field, etc.). In the quantum network model, this problem is equivalent to the scattering problem in the I-junction.

In Appendix \ref{sec:B}, we will show how to use the scattering boundary conditions (\ref{Eq042}) to write an expression for the S-matrix of the I-junction in terms of the operator $G^{\Diamond}$ (\ref{Eq058}). For this purpose, we will consider three-, two-, and one-dimensional mathematical models of quantum wire.

\subsection{Three-dimensional I-junction}	\label{sec:B_1}

The specificity of the quantum wire is that both branches (\ref{Eq019}) have cross sections $\beta = \{{\beta}^k\}^{k = 1,2}$ equal to the cross section of the wire in GCF ${{\rm B}}^= $:
\begin{equation}	\label{Eq198}
	{\{{{\bar{W}}^k}{{\beta}^k} = {{{\rm B}}^= }\}}^{k = 1,2}
\end{equation}
Taking into account the property (\ref{Eq198}), the electron scattering problem (\ref{Eq057}) in the support of the compact potential (I-junction) $\Omega = (-a,+a)\times {{{\rm B}}^= }$ takes the form:
\begin{equation}	\label{Eq199}
	\left\{ \begin{aligned}
		&[ -\Delta + \upsilon (x,y,z) ]\Psi (x,&&\!\!\!\!\!\!y,z) = \varepsilon \Psi (x,y,z),	\\
		&&& \ \ \{x,y,z\}\in \left( -a,+a \right)\times {{{\rm B}}^= } \\
		&\Psi (x,y,z) = 0, && \ \ \{x,y,z\}\in (-a,+a)\times \partial {\rm B}^=  \\
		&[ K + i{\partial}_1 ]W\Psi(0,y,z)&&\!\!\!\!\!\! = 2K{{\psi}^{\triangleleft}}(0,y,z),	\quad \{y,z\} \in \beta
	\end{aligned} \right.
\end{equation}
\begin{prop}	\label{prop:14}
 The operator $G^{\Diamond}$ for the scattering problem in the I-junction (\ref{Eq199}) is written as
\begin{equation}	\label{Eq200}
	G_{nm}^{\Diamond}(x) = \left[ \begin{matrix}
	g_{nm}^{\Diamond 1}(-x-a) & \ g_{nm}^{\Diamond 2}(-x-a) \\
		g_{nm}^{\Diamond 1}(+x+a) & \ g_{nm}^{\Diamond 2}(+x+a) \\
		\end{matrix} \right]
\end{equation}
\begin{equation}	\label{Eq201}
	\begin{aligned}
		&\sum\nolimits_n{\left( -{I_{mn}}\partial _1^2+{V_{mn}} \right)g_{nk}^{\Diamond}} = \sum\nolimits_n{{E_{mn}}g_{nk}^{\Diamond}},	\\
		&\forall x\in (-a,+a)\quad
		\left| \begin{matrix}
			g^{\Diamond 1} & g^{\Diamond 2} \\
			{\dot{g}}^{\Diamond 1} & {\dot{g}}^{\Diamond 2} \\
		\end{matrix} \right| (x)\ne 0
	\end{aligned}
\end{equation}
\begin{equation}	\label{Eq202}
	{V_{mn}} := \langle {h_m} |\upsilon | {h_n} \rangle,\quad {E_{mn}} := {I_{mn}}\left( \varepsilon -{{\lambda}_m} \right)
\end{equation}
where orthonormal system of functions ${\{{h_m}\}}_m$ and eigenvalues ${\{{{\lambda}_m}\}}_m$ are solution of the problem
\begin{equation}	\label{Eq203}
	\left\{ \begin{aligned}
		&[-\partial _y^2-\partial _z^2]{h_m} = {{\lambda}_m}{h_m} && \operatorname{in}\ \ {{{\rm B}}^= } \\
		&{h_m} = 0 && \operatorname{on}\ \partial {{{\rm B}}^= }
	\end{aligned} \right.
\end{equation}
\end{prop}
\begin{proof}[\textsc{Proof~\ref{prop:14}}]
 We are looking for a solution to the problem (\ref{Eq199}) in the form of decomposition according to the orthonormal system ${\{{h_m}\}}_m$ (\ref{Eq203}):
\begin{equation}	\label{Eq204}
	\Psi \left( x,y,z \right) = \sum\nolimits_m{{g_m}\left( x \right){h_m}\left( y,z \right)},\quad {g_m} := \langle {h_m} | \Psi \rangle 
\end{equation}
$\begin{aligned}
	&(\ref{Eq199}) \quad \Rightarrow \quad 0 = [ -\Delta -\varepsilon +\upsilon ]| \Psi \rangle	\\
	&= \sum\nolimits_{mn}{| {h_m} \rangle \langle {h_m} |\left[ -\partial _1^2-\partial _2^2-\partial _3^2-\varepsilon +\upsilon \right]| {h_n} \rangle \langle {h_n} | \Psi \rangle};	\\
	&(\ref{Eq204}), (\ref{Eq203}), (\ref{Eq202}) \quad \Rightarrow
\end{aligned}$
\begin{equation}	\label{Eq205}
	\sum\nolimits_n{\left( -{I_{mn}}\partial _1^2+{V_{mn}} \right){g_n}} = \sum\nolimits_n{{E_{mn}}{g_n}}
\end{equation}

The solution of system (\ref{Eq205}) can be written as:
\begin{equation}	\label{Eq206}
	\begin{aligned}
		{g_m}(x) = \sum\nolimits_n{g_{mn}^{\Diamond 1}(x)c_n^{\Diamond 1}} &+ \sum\nolimits_n{g_{mn}^{\Diamond 2}(x)c_n^{\Diamond 2}}	\\
		\Leftrightarrow \quad g =\,& {g^{\Diamond}}{c^{\Diamond}}
	\end{aligned}
\end{equation}
where $g^{\Diamond 1}$, $g^{\Diamond 2}$ are matrices of the \textit{fundamental system of solutions} (\textit{FSS}) of problem (\ref{Eq205}), defined as (\ref{Eq201}).

\hspace{-17pt}
$\begin{aligned}
	&(\ref{Eq204}), (\ref{Eq206}) \quad \Rightarrow \quad {W^k}| \Psi \rangle = \sum\nolimits_n{| h_n^k \rangle \langle h_n^k |{W^k}| \Psi \rangle} = \\
	&\sum\nolimits_{nmp}^l{| h_n^k \rangle \langle h_n^k |{W^k}g_{pm}^{\Diamond l}c_m^{\Diamond l}| {h_p} \rangle} = :\sum\nolimits_{nm}^l{G_{nm}^{\Diamond kl}c_m^{\Diamond l}h_n^k}\quad \Rightarrow
\end{aligned}$
\begin{equation}	\label{Eq207}
	{W^k}\Psi = \sum\nolimits_n{[{G^{\Diamond}}{c^{\Diamond}}]_n^kh_n^k},\quad G_{nm}^{\Diamond kl} = \sum\nolimits_p{\langle h_n^k |{W^k}g_{pm}^{\Diamond l}| {h_p} \rangle}
\end{equation}

The operator $G^{\Diamond}$ in expression (\ref{Eq207}) completely complies with the definition (\ref{Eq058}). It can also be written in matrix form by the branches:

\pagebreak
\begin{strip}
$\begin{aligned}
	G_{nm}^{\Diamond kl} = \sum\nolimits_p{\left[ {W^k}g_{pm}^{\Diamond l} \right]\langle h_n^k |{W^k}| {h_p} \rangle}
	= \sum\nolimits_p{\left[ {W^k}g_{pm}^{\Diamond l} \right]}\langle {{{\bar{W}}}^k}h_n^k | {h_p} \rangle ; \quad	
	\left\{ \begin{aligned}
		&{W^1}g_{pm}^{\Diamond}(x) = g_{pm}^{\Diamond}(-x-a) \\
		&{W^2}g_{pm}^{\Diamond}(x) = g_{pm}^{\Diamond}(+x+a)
	\end{aligned} \right.\quad \Rightarrow
\end{aligned}$
\begin{equation}	\label{Eq208}
	G_{nm}^{\Diamond}(x) = \sum\nolimits_p{\left[ \begin{matrix}
	g_{pm}^{\Diamond 1}(-x-a)\langle {{{\bar{W}}}^1}h_n^1 | {h_p} \rangle & \ g_{pm}^{\Diamond 2}(-x-a)\langle {{{\bar{W}}}^1}h_n^1 | {h_p} \rangle \\
		g_{pm}^{\Diamond 1}(+x+a)\langle {{{\bar{W}}}^2}h_n^2 | {h_p} \rangle & \ g_{pm}^{\Diamond 2}(+x+a)\langle {{{\bar{W}}}^2}h_n^2 | {h_p} \rangle \\
		\end{matrix} \right]}
\end{equation}
\vspace{-15pt}
\end{strip}
The scalar products in expression (\ref{Eq208}) are determined by the cross section of the wire and the choice of LCFs in a particular problem. In the simplest case, when LCFs relates to each other by reflection relative to the plane $YZ$, we have ${{\bar{W}}^1}h_n^1 = {h_n} = {{\bar{W}}^2}h_n^2$, and the expression (\ref{Eq208}) takes the form (\ref{Eq200}).	
\end{proof}

FSS matrices $\{g^{\Diamond k}\}^{k = 1,2}$ can be found numerically, for example, by solving Cauchy problems:
\begin{equation}	\label{Eq209}
	\left\{ \begin{aligned}
		&\sum\nolimits_p{\left[ -{I_{mp}}\partial _x^2+{V_{mp}}(x) \right] g_{pn}^{\Diamond k}(x)} = \sum\nolimits_p&&\!\!\!\!\!\!{{E_{mp}}g_{pn}^{\Diamond k}(x)},	\\
		&&& x\in \left( -a,+a \right) \\
		&g_{mn}^{\Diamond k}\left( -a \right) = I_{mn}^{1k},\quad {{\partial}_1}g_{mn}^{\Diamond k}\left( -a \right) = I_{mn}^{2k}&&
	\end{aligned} \right.
\end{equation}

Thus, the three-dimensional scattering problem on heterogeneity in quantum wire (\ref{Eq199}) is reduced to a system of ordinary differential equations. Having calculated the matrices of the fundamental system of solutions based on problems (\ref{Eq209}), one can write the operator $G^{\Diamond}$ (\ref{Eq207}). With it, it is possible to find the decomposition coefficients for the function $\Psi $ in expressions (\ref{Eq204}), (\ref{Eq206}) and the S-matrix of the I-junction using formulas (\ref{Eq060}) and (\ref{Eq056}), respectively.

\subsection{Two-dimensional I-junction}	\label{sec:B_2}

The two-dimensional mathematical model of quantum wire is relevant for structures formed on the basis of two-dimensional electron gas. As in the three-dimensional case, the specificity (\ref{Eq198}) for branches with cross sections in the form of intervals $\{{{\beta}^k} = (0,b)\}^{k = 1,2}$ is preserved here. Also we choose GCF such that ${{{\rm B}}^= } = (0,b)$. Then, taking into account the property (\ref{Eq198}), the electron scattering problem (\ref{Eq057}) in the support of the compact potential (I-junction) $\Omega = (-a,+a)\times (0,b)$ takes the form:
\begin{equation}	\label{Eq210}
	\left\{ \begin{aligned}
		&[ -\Delta + \upsilon (x,y) ]\Psi (x,y)  &&\!\!\!\!\!\!\!\!= \varepsilon \Psi (x,y),	\\
		&&&\!\!\!\!\!\!\!\!\{x,y\} \in (-a,+a) \times (0,b) \\
		&\Psi (x,0) = \Psi (x,b) = 0,  &&\qquad\qquad\!\!  x \in (-a,+a) \\
		&[ K+i{\partial}_1]W\Psi (0,y) = &&\!\!\!\!\!\! 2K{{\psi}^{\triangleleft}}(0,y),  \quad y \in (0,b)
	\end{aligned} \right.
\end{equation}
\begin{prop}	\label{prop:15}
 The operator $G^{\Diamond}$ for the scattering problem in the I-junction (\ref{Eq210}) is written as
\begin{equation}	\label{Eq211}
	G_{nm}^{\Diamond}(x) = \left[ \begin{matrix}
	{{\left( -1 \right)}^{n+1}}g_{nm}^{\Diamond 1}(-x-a) & \ {{(-1)}^{n+1}}g_{nm}^{\Diamond 2}(-x-a) \\
		g_{nm}^{\Diamond 1}(+x+a) & \ g_{nm}^{\Diamond 2}(+x+a) \\
		\end{matrix} \right]
\end{equation}
FSS matrices $\{g^{\Diamond k}\}^{k = 1,2}$ are defined as (\ref{Eq201}) and (\ref{Eq202}), where orthonormal system of functions $\{{h_m}\}_m$ and eigenvalues $\{{\lambda}_m\}_m$ are solution of the problem
\begin{equation}	\label{Eq212}
	\left\{ \begin{aligned}
		&-\partial _y^2{h_m}(y) = {\lambda}_m{h_m}(y), && y\in (0,b) \\
		&{h_m}(y) = 0, && y\in \{0,b\}
	\end{aligned} \right.
\end{equation}
\end{prop}
\begin{proof}[\textsc{Proof~\ref{prop:15}}]
 We are looking for a solution to the problem (\ref{Eq210}) in the form of decomposition according to the orthonormal system ${\{{h_m}\}}_m$ (\ref{Eq212}):
\begin{equation}	\label{Eq213}
	\Psi \left( x,y \right) = \sum\nolimits_m{{g_m}\left( x \right){h_m}\left( y \right)},\quad {g_m} := \langle {h_m} | \Psi \rangle 
\end{equation}
Subsequent computations dimensionally accurate are similar to those following decomposition (\ref{Eq204}). Therefore, the expressions (\ref{Eq205})--(\ref{Eq209}) are also valid here.

For two-dimensional systems, it is convenient to use LCFs relates to each other by rotation and translation (Fig.~\ref{fig:5}). Let us find the operator $G^{\Diamond}$ (\ref{Eq208}) in this case.

\begin{strip}
\noindent
$\left. \begin{aligned}
	&\left\{ \begin{aligned}
		&{w^1}(x,y) = (-x-a,b-y)\quad &\Rightarrow& \quad {W^1}g_m^{\Diamond}(x){h_m}(y) = g_m^{\Diamond}(-x-a){h_m}(b-y) \\
		 &{w^2}(x,y) = (+x+a,y)\quad &\Rightarrow& \quad {W^2}g_m^{\Diamond}(x){h_m}(y) = g_m^{\Diamond}(+x+a){h_m}(y)
	\end{aligned} \right.\quad \\
	&{{{\bar{W}}}^2}h_m^2 = h_m^2 = h_m^1,\quad \quad {{{\bar{W}}}^1} = {W^1} \\
	&(\ref{Eq212}) \quad \Rightarrow \quad h_m^k\left( y \right) = \sqrt{\tfrac{2}{b}}\sin \left( \tfrac{\pi m}{b}y \right)\quad \Rightarrow  \quad {{{\bar{W}}}^1}h_m^1(y) = \sqrt{\tfrac{2}{b}}\sin \left( \tfrac{\pi m}{b}[b-y] \right)	 \quad \\
	&\quad = \sqrt{\tfrac{2}{b}}\left[ \sin \left( \tfrac{\pi m}{b}b \right)\cos \left( \tfrac{\pi m}{b}y \right)-\cos \left( \tfrac{\pi m}{b}b \right)\sin \left( \tfrac{\pi m}{b}y \right) \right] = {{\left( -1 \right)}^{m+1}}h_m^1\left( y \right) \quad
\end{aligned}
\right\}\Rightarrow $
\vspace{-20pt}
\end{strip}
\begin{equation}	\label{Eq214}
	{{\bar{W}}^1}h_n^1 = {{\left( -1 \right)}^{n+1}}h_n^1,\quad {{\bar{W}}^2}h_n^2 = h_n^1
\end{equation}
We choose the GCF so that ${h_n} = h_n^1$. Then (\ref{Eq208}), (\ref{Eq214})~$\quad \Rightarrow \quad$~(\ref{Eq211}).	
\end{proof}
\newpage

\mbox{}\\
Thus, the two-dimensional problem of scattering on inhomogeneity in the quantum wire (\ref{Eq210}) is reduced to a system of ordinary differential equations. Having calculated the matrices of the fundamental system of solutions based on problems (\ref{Eq209}), one can write the operator $G^{\Diamond}$ (\ref{Eq207}). With it, it is possible to find the decomposition coefficients for the function $\Psi $ in expressions (\ref{Eq213}), (\ref{Eq206}) and the S-matrix of the I-junction using formulas (\ref{Eq060}) and (\ref{Eq056}), respectively.

\subsection{One-dimensional I-junction}	\label{sec:B_3}

The simplest mathematical model of quantum wire is a one-dimensional model. It describes the motion of an electron in a single channel whose energy is assumed to be zero. Due to its specificity, this model has a small field of application and gives an idea of the transport properties of quantum wire at a qualitative level. For one-dimensional case
\begin{equation}	\label{Eq215}
	\{h_n^k\}_n^{k = 1,2} = \{1,1\},\quad {{\{{K^{kl}} = {I^{kl}}{{\kappa}^l},{{\kappa}^l} = \sqrt{\varepsilon}\}}^{k,l = 1,2}}
\end{equation}
Then, taking into account expressions (\ref{Eq037}), the electron scattering problem (\ref{Eq057}) in the support of the compact potential (I-junction) $\Omega = (-a,+a)$ takes the form:
\begin{equation}	\label{Eq216}
	\left\{ \begin{aligned}
		&[ -\partial _x^2+\upsilon (x)]\Psi (x) = \varepsilon \Psi (x), \quad x\in (-a,+a) \\
		&[ K+i{\partial}_1 ]W\Psi (0) = 2K{c^{\triangleleft}}
	\end{aligned} \right.
\end{equation}
\begin{prop}	\label{prop:16}
 The operator $G^{\Diamond}$ for the scattering problem in the I-junction (\ref{Eq216}) is written as
\begin{equation}	\label{Eq217}
	{G^{\Diamond}}(x) = \left[ \begin{matrix}
	{g^{\Diamond 1}}(-x-a) & \ {g^{\Diamond 2}}(-x-a) \\
		{g^{\Diamond 1}}(+x+a) & \ {g^{\Diamond 2}}(+x+a) \\
		\end{matrix} \right]
\end{equation}
where $g^{\Diamond 1}$ and $g^{\Diamond 2}$ are linearly independent solutions of the equation for $\Psi $ in the problem (\ref{Eq216}).

\end{prop}
\begin{proof}[\textsc{Proof~\ref{prop:16}}]
 The solution of problem (\ref{Eq216}) is a linear combination of two linearly independent solutions of the equation for $\Psi $ in problem (\ref{Eq216}):
\begin{equation}	\label{Eq218}
	\Psi \left( x \right) = {g^{\Diamond 1}}\left( x \right){c^{\Diamond 1}}+{g^{\Diamond 2}}\left( x \right){c^{\Diamond 2}} = \sum\nolimits_{}^k{{g^{\Diamond k}}\left( x \right){c^{\Diamond k}}}
\end{equation}
$\begin{aligned}
(\ref{Eq218}) \quad \Rightarrow \quad {W^k}\Psi = {W^k}\sum\nolimits_{}^l{{g^{\Diamond l}}{c^{\Diamond l}}} = \sum\nolimits_{}^l{{W^k}{g^{\Diamond l}}{c^{\Diamond l}}} \quad \Rightarrow
\end{aligned}$
\begin{equation}	\label{Eq219}
	{W^k}\Psi = {{[{G^{\Diamond}}{c^{\Diamond}}]}^k},\quad {G^{\Diamond kl}} := {W^k}{g^{\Diamond l}}
\end{equation}

Taking into account the equalities (\ref{Eq215}), the operator $G^{\Diamond}$ in expression (\ref{Eq219}) completely meets the definition (\ref{Eq058}). It can also be written in matrix form:

\renewcommand{\qedsymbol}{}
\begin{strip}
$\left\{ \begin{aligned}
	&{G^{\Diamond kl}} = {W^k}{g^{\Diamond l}}\quad \Rightarrow \quad {G^{\Diamond}}\left( x \right) =
	\left[ \begin{matrix}
		{W^1}{g^{\Diamond 1}}\left( x \right)  & \ {W^1}{g^{\Diamond 2}}\left( x \right) \\
		{W^2}{g^{\Diamond 1}}\left( x \right)  & \ {W^2}{g^{\Diamond 2}}\left( x \right) \\
	\end{matrix} \right] \\
	&{{\left( {w^1} \right)}^{-1}}x = -x-a\quad \Rightarrow \quad {w^1}x = -x-a\quad \Rightarrow \quad {W^1}{g^{\Diamond}}\left( x \right) = {g^{\Diamond}}\left( {w^1}x \right) = {g^{\Diamond}}\left( -x-a \right) \\
	&{{\left( {w^2} \right)}^{-1}}x = +x-a\quad \Rightarrow \quad {w^2}x = +x+a\quad \Rightarrow \quad {W^2}{g^{\Diamond}}\left( x \right) = {g^{\Diamond}}\left( {w^2}x \right) = {g^{\Diamond}}\left( +x+a \right)
\end{aligned} \right.\quad \Rightarrow \quad $ (\ref{Eq217})	~\hfill$\blacksquare$
\end{strip}
\end{proof}

\vspace{-22pt}
Functions $g^{\Diamond k}$ can be found numerically, for example, by solving two Cauchy problems:
\begin{equation}	\label{Eq220}
	\left\{ \begin{aligned}
		&[ -\partial _x^2+\upsilon (x) ]{g^{\Diamond k}}(x) = \varepsilon {g^{\Diamond k}}(x),\quad x\in (-a,+a) \\
		&{g^{\Diamond k}}(-a) = {I^{1k}},\quad {{\partial}_1}{g^{\Diamond k}}(-a) = {I^{2k}}
	\end{aligned} \right.
\end{equation}

Thus, the one-dimensional scattering problem on heterogeneity in the quantum wire (\ref{Eq216}) is reduced to a linear combination of two functions. Having calculated two linearly independent solutions based on problems (\ref{Eq220}), one can write the operator $G^{\Diamond}$ (\ref{Eq217}). With it, it is possible to find the decomposition coefficients for the function $\Psi $ in expression (\ref{Eq218}) and the S-matrix of the I-junction using formulas (\ref{Eq060}) and (\ref{Eq056}), respectively.

\section{S-matrix of quantum network junction in terms of DN- and ND-map}	\label{sec:C}

\begin{equation}	\label{Eq221}
	\{\ {{\square}^{[\mathbb{A}]}}\ \mapsto \ \square \ |\ \square \ = D,N,...\}\quad \operatorname{in}\ \text{\ref{sec:C}}
\end{equation}

In section \ref{sec:2_3_3}, we have written the DN- and ND-map methods within the notation system used (section \ref{sec:2_1_2}). In Appendix \ref{sec:C}, we will prove statements 4 and 5.

\begin{proof}[\textsc{Proof~\ref{prop:4}}]
 (\ref{Eq062}), $\chi = \{{\chi}^k\}^k$, (\ref{Eq030}) $\quad \Rightarrow \quad$
\begin{equation}	\label{Eq222}
	\sum\nolimits_n{D_{mn}^{kl}\langle h_n^l | {{\chi}^l}\left( 0,... \right) \rangle} = \langle h_m^k | {{\partial}_1}{W^k}\Psi \left( 0,... \right) \rangle 
\end{equation}
(\ref{Eq062}), (\ref{Eq044}) $\quad \Rightarrow \quad$
\begin{equation}	\label{Eq223}
	\sum\nolimits_{}^l{{D^{kl}}{{\psi}^l}} = {{\partial}_1}{{\psi}^k}\quad \operatorname{on}\ \gamma 
\end{equation}
(\ref{Eq222}), (\ref{Eq223}), (\ref{Eq035})--(\ref{Eq037}) $\quad \Rightarrow \quad$

\noindent
$\begin{aligned}
	0 =& \sum\nolimits_{}^l{{D^{kl}}| {{\psi}^l}(0,...) \rangle}-| {{\partial}_1}{{\psi}^k}(0,...) \rangle	\\
	=& \sum\nolimits_{mn}^l{| h_m^l \rangle \langle h_m^l |{D^{kl}}| h_n^l \rangle \langle h_n^l | {{\psi}^l}(0,...) \rangle}	\\
	&-\sum\nolimits_m^{}{| h_m^l \rangle \langle h_m^l | {{\partial}_1}{{\psi}^k}(0,...) \rangle}	\\
	 =& \sum\nolimits_m{| h_m^l \rangle}\left[ \sum\nolimits_n^l{D_{mn}^{kl}\langle h_n^l | {{\psi}^l}(0,...) \rangle } \right.	\\
	 &\qquad\qquad\qquad\; \left. - \langle h_m^l | {{\partial}_1}{{\psi}^k}(0,...) \rangle \right] \quad \Rightarrow
\end{aligned}$

\noindent
$\begin{aligned}
	0 =& \sum\nolimits_n^l{D_{mn}^{kl}\langle h_n^l | {{\psi}^l}(0,...) \rangle}-\langle h_m^k | {{\partial}_1}{{\psi}^k}(0,...) \rangle	\\
	=& \sum\nolimits_n^l{D_{mn}^{kl}\langle h_n^l | \sum\nolimits_p{\left[ {c^{\triangleleft}}+S{c^{\triangleleft}} \right]_p^l}h_p^l \rangle}	\\
	&- \langle h_m^k | \sum\nolimits_n{\left[ -iK{c^{\triangleleft}}+iKS{c^{\triangleleft}} \right]_n^kh_n^k} \rangle	\\
	 =& \sum\nolimits_n^l{D_{mn}^{kl}\left[ {c^{\triangleleft}}+S{c^{\triangleleft}} \right]_n^l}-\left[ -iK{c^{\triangleleft}}+iKS{c^{\triangleleft}} \right]_m^k\quad \Rightarrow	\\
	 O =& D\left[ {c^{\triangleleft}}+S{c^{\triangleleft}} \right]+iK{c^{\triangleleft}}-iKS{c^{\triangleleft}}
\end{aligned}$

\noindent
$\begin{aligned}
	\quad =& \left[ D+iK \right]{c^{\triangleleft}}+\left[ D-iK \right]S{c^{\triangleleft}}\quad \Rightarrow \quad (\ref{Eq061})
\end{aligned}$ 	
\end{proof}

\begin{proof}[\textsc{Proof~\ref{prop:5}}]
  (\ref{Eq064}), $\dot{\chi} = \{ \dot{\chi}^k \}^k$, (\ref{Eq030}) $\quad \Rightarrow \quad$
\begin{equation}	\label{Eq224}
	\sum\nolimits_n{N_{mn}^{kl}\langle h_n^l | {{{\dot{\chi}}}^l}\left( 0,... \right) \rangle} = \langle h_m^k | {W^k}\Psi \left( 0,... \right) \rangle 
\end{equation}
(\ref{Eq064}), (\ref{Eq044}) $\quad \Rightarrow \quad$
\begin{equation}	\label{Eq225}
	\sum\nolimits_{}^l{{N^{kl}}{{\partial}_1}{{\psi}^l}} = {{\psi}^k}\quad \operatorname{on}\ \gamma 
\end{equation}
(\ref{Eq224}), (\ref{Eq225}), (\ref{Eq035})--(\ref{Eq037}) $\quad \Rightarrow \quad$

\noindent
$\begin{aligned}
	0 =& \sum\nolimits_{}^l{{N^{kl}}| {{\partial}_1}{{\psi}^l}(0,...) \rangle}-| {{\psi}^k}(0,...) \rangle	\\
	=& \sum\nolimits_{mn}^l{| h_m^l \rangle \langle h_m^l |{N^{kl}}| h_n^l \rangle \langle h_n^l | {{\partial}_1}{{\psi}^l}(0,...) \rangle}	\\
	&-\sum\nolimits_m^{}{| h_m^l \rangle \langle h_m^l | {{\psi}^k}(0,...) \rangle}	\\
	=& \sum\nolimits_m^{}{| h_m^l \rangle}\left[ \sum\nolimits_n^l{N_{mn}^{kl}\langle h_n^l | {{\partial}_1}{{\psi}^l}(0,...) \rangle} \right.	\\
	&\qquad\qquad\qquad\; \left. - | h_m^l \rangle \langle h_m^l | {{\psi}^k}(0,...) \rangle \right]	\quad \Rightarrow
\end{aligned}$

\noindent
$\begin{aligned}
	0 =& \sum\nolimits_n^l{N_{mn}^{kl}\langle h_n^l | {{\partial}_1}{{\psi}^l}(0,...) \rangle}-\langle h_m^k | {{\psi}^k}(0,...) \rangle	\\
	=& \sum\nolimits_n^l{N_{mn}^{kl}\langle h_n^l | \sum\nolimits_p{\left[ -iK{c^{\triangleleft}}+iKS{c^{\triangleleft}} \right]_p^lh_p^l} \rangle}	\\
	&-\langle h_m^k | \sum\nolimits_n{\left[ {c^{\triangleleft}}+S{c^{\triangleleft}} \right]_n^kh_n^k} \rangle	\\
	=& \sum\nolimits_n^l{N_{mn}^{kl}\left[ -iK{c^{\triangleleft}}+iKS{c^{\triangleleft}} \right]_n^l}-\left[ {c^{\triangleleft}}+S{c^{\triangleleft}} \right]_m^k\quad \Rightarrow	\\
	O =& N\left[ -iK{c^{\triangleleft}}+iKS{c^{\triangleleft}} \right]-{c^{\triangleleft}}-S{c^{\triangleleft}}
\end{aligned}$

\noindent
$\begin{aligned}
	\quad = -\left[ NiK+I \right]{c^{\triangleleft}}+\left[ NiK-I \right]S{c^{\triangleleft}} \quad \Rightarrow \quad (\ref{Eq063})
\end{aligned}$ 	
\end{proof}

\section{Network combining formula for series networks}	\label{sec:D}

In section \ref{sec:2_4_3}, we indicated that in the particular case, NCF (\ref{Eq092}) gives a result that coincides with what was previously obtained in the literature. In Appendix \ref{sec:D}, we will prove this statement.

\subsection{Network of arbitrary dimension}	\label{sec:D_1}

Let us consider sections of a quantum network consisting of two internal junctions connected by one branch. We assume the geometry of the junctions and branches arbitrary in the framework of the problem solved in section \ref{sec:2}. Without loss of generality, for simplicity, we consider the case where the LCFs at the ends of each internal branch are related to each other by reflection: ${U^{[\mathbb{J},\mathbb{K}]\mathbb{K}\mathbb{K}}} = {I^{\mathbb{K}\mathbb{K}}} = {U^{[\mathbb{J},\mathbb{K}]\mathbb{K}\mathbb{K}}}$.

For the closed-end section of network $\mathcal{N} = \{\{1,2\},\{2\}\}$. NCF (\ref{Eq092}) takes the form:
\begin{equation}	\label{Eq226}
	{S^{[1]}} = {S^{[1,2]}}\circledast {S^{[2]}}
\end{equation}
from formula (\ref{Eq226}), it follows that a tuple of external branch identifiers $\mathbb{E} = \{1\}$. Expanding expression (\ref{Eq226}) by formula (\ref{Eq071}) taking into account $\mathbb{J} = \{1\}$, $\mathbb{K} = \{2\}$, $\mathbb{L} = \varnothing $, we have
\begin{equation}	\label{Eq227}
	\begin{aligned}
		{S^{[1]}} =& {S^{[1,2]11}}+
		\left[ \begin{matrix}
		{S^{[1,2]12}} & {O^{12}} \\
		\end{matrix} \right]	\\
		&\times \left[ \begin{matrix}
			-{S^{[1,2]22}} &\!\! \exp \left( -i{K^{22}}{A^{22}} \right) \\
			\exp \left( -i{K^{22}}{A^{22}} \right) &\!\!  -{S^{[2]22}} \\
		\end{matrix} \right]^{-1}
		\left[ \begin{matrix}
			{S^{[1,2]21}} \\
			{O^{21}} \\
		\end{matrix} \right]
	\end{aligned}
\end{equation}

For the open-end section of network $\mathcal{N} = \{\{1,2\},\{2,3\}\}$. NCF (\ref{Eq092}) takes the form:
\begin{equation}	\label{Eq228}
	{S^{[1,3]}} = {S^{[1,2]}}\circledast {S^{[2,3]}}
\end{equation}
from the formula (\ref{Eq228}), it follows that a tuple of external branch identifiers $\mathbb{E} = \{1,3\}$. Expanding expression (\ref{Eq228}) by formula (\ref{Eq071}) taking into account $\mathbb{J} = \{1\}$, $\mathbb{K} = \{2\}$, $\mathbb{L} = \{3\}$, we have
\begin{equation}	\label{Eq229}
	\begin{aligned}
		{S^{[1,3]}} =&
		\left[ \begin{matrix}
			{S^{[1,2]11}} & {O^{13}} \\
			{O^{31}} & {S^{[2,3]33}} \\
		\end{matrix} \right]+
		\left[ \begin{matrix}
			{S^{[1,2]12}} & {O^{12}} \\
			{O^{32}} & {S^{[2,3]32}} \\
		\end{matrix} \right]	\\
		&\times  \left[ \begin{matrix}
			-{S^{[1,2]22}} & \exp \left( -i{K^{22}}{A^{22}} \right) \\
			\exp \left( -i{K^{22}}{A^{22}} \right) & -{S^{[2,3]22}} \\
		\end{matrix} \right]^{-1}	\\
		&\times \left[ \begin{matrix}
			{S^{[1,2]21}} & {O^{23}} \\
			{O^{21}} & {S^{[2,3]23}} \\
		\end{matrix} \right]
	\end{aligned}
\end{equation}

\subsection{One-dimensional network}	\label{sec:D_2}

One of the varieties of quantum networks is a one-dimensional network. It is a sequence of compact potentials and consists of closed-end and/or open-end sections. We write their S-matrices (\ref{Eq227}) and (\ref{Eq229}) using type identifiers (section \ref{sec:2_1_2}) based on junction numbering. For closed-end and open-end sections, we introduce the symbols:
\begin{equation}	\label{Eq230}
	{S^{\langle 3 \rangle}} := {S^{[1]}},\hspace{16pt} {S^{\langle 1 \rangle}} := {S^{[1,2]}},\quad {S^{\langle 2 \rangle}} := {S^{[2]}}
\end{equation}
\begin{equation}	\label{Eq231}
	{S^{\langle 3 \rangle}} := {S^{[1,3]}},\quad {S^{\langle 1 \rangle}} := {S^{[1,2]}},\quad {S^{\langle 2 \rangle}} := {S^{[2,3]}}
\end{equation}
respectively. In both cases, we will shorten notation: ${{\kappa}^2} = \sqrt{\varepsilon} = :\kappa $, ${a^2} = :a$. Since the values in definitions (\ref{Eq230}) and (\ref{Eq231}) are numbers, they commute among themselves. Therefore, formulas (\ref{Eq227}) and (\ref{Eq229}) taking into account definitions (\ref{Eq073}) are simplified:
\begin{equation}	\label{Eq232}
	{S^{\langle 3 \rangle}} = {S^{\langle 1 \rangle 11}}+\frac{{S^{\langle 1 \rangle 12}}{S^{\langle 2 \rangle 22}}{S^{\langle 1 \rangle 21}}{e^{i2\kappa a}}}{1-{S^{\langle 1 \rangle 22}}{S^{\langle 2 \rangle 22}}{e^{i2\kappa a}}}
\end{equation}
\begin{strip}
\begin{equation}	\label{Eq233}
	{S^{\langle 3 \rangle}} = \left[ \displaystyle\begin{matrix}
	{S^{\langle 1 \rangle 11}}+\displaystyle{
\frac{{S^{\langle 1 \rangle 12}}{S^{\langle 2 \rangle 22}}{S^{\langle 1 \rangle 21}}{e^{i2\kappa a}}}{1-{S^{\langle 1 \rangle 22}}{S^{\langle 2 \rangle 22}}{e^{i2\kappa a}}}} & \displaystyle{
\frac{{S^{\langle 1 \rangle 12}}{S^{\langle 2 \rangle 23}}{e^{i\kappa a}}}{1-{S^{\langle 1 \rangle 22}}{S^{\langle 2 \rangle 22}}{e^{i2\kappa a}}}} \\
		\displaystyle{
\frac{{S^{\langle 2 \rangle 32}}{S^{\langle 1 \rangle 21}}{e^{i\kappa a}}}{1-{S^{\langle 1 \rangle 22}}{S^{\langle 2 \rangle 22}}{e^{i2\kappa a}}}} & {S^{\langle 2 \rangle 33}}+\displaystyle{
\frac{{S^{\langle 2 \rangle 32}}{S^{\langle 1 \rangle 22}}{S^{\langle 2 \rangle 23}}{e^{i2\kappa a}}}{1-{S^{\langle 1 \rangle 22}}{S^{\langle 2 \rangle 22}}{e^{i2\kappa a}}}} \\
		\end{matrix} \right]
\end{equation}
\end{strip}
respectively. For a “branchless” network (\ref{Eq090}), taking into account the specifics of the notation, the expression (\ref{Eq233}) coincides with the expression obtained earlier \cite[p.~49]{Bib027}.

\section{Landauer--B\"uttiker formalism}	\label{sec:E}

\begin{equation}	\label{Eq234}
	\{\ {{\square}^{[\mathbb{E}]}}\ \mapsto \ \square \ |\ \square \ = J,P,...\}\quad \operatorname{in}\ \text{\ref{sec:E}}
\end{equation}

In section \ref{sec:3_2}, we write the Landauer--B\"uttiker formalism in framework of the notation system used (section \ref{sec:2_1_2}). In Appendix \ref{sec:E} we will give the mathematical justification for it.

\subsection{Expression for currents}	\label{sec:E_1}

\begin{proof}[\textsc{Proof~\ref{prop:9}}]
 (\ref{Eq120}), (\ref{Eq121}), (\ref{Eq122}), (\ref{Eq123}) $\quad \Rightarrow \quad$
\begin{equation}	\label{Eq235}
	{J^k} = \int_{-\infty}^{+\infty}{dE{j^k}\left( E \right)}
\end{equation}
\begin{equation}	\label{Eq236}
	{j^k} := \sum\nolimits_m{j_m^k},\quad j_m^k := j_m^{\triangleleft k}+j_m^{\triangleright k}
\end{equation}
(\ref{Eq125}) $\quad \Rightarrow \quad$
\begin{equation}	\label{Eq237}
	P_{mn}^{kl}\left( E \right) = [E_{\bot m}^k<E]P_{mn}^{kl}\left( E \right)[E>E_{\bot n}^l]
\end{equation}
(\ref{Eq122}), (\ref{Eq123}), (\ref{Eq235}), (\ref{Eq236}), (\ref{Eq237}) $\quad \Rightarrow \quad$

\renewcommand{\qedsymbol}{}
\begin{strip}
$\left. \begin{aligned}
	J^k =& \sum\nolimits_m{\int_{-\infty}^{+\infty}{dE\left\{ -\frac{e}{\pi \hbar}{f^k}\left( E \right)[E_{\bot m}^k<E]+[E_{\bot m}^k<E]\sum\nolimits_n^l{P_{mn}^{kl}\left( E \right)\frac{e}{\pi \hbar}{f^l}\left( E \right)[E_{\bot n}^l<E]} \right\}}}	\\
	=& \frac{e}{\pi \hbar}\sum\nolimits_m{\int_{-\infty}^{+\infty}{dE[E_{\bot m}^k<E]\left\{ -{f^k}\left( E \right) +\sum\nolimits_n^l{[E_{\bot n}^l<E]P_{mn}^{kl}\left( E \right){f^l}\left( E \right)}\right\}}} \\
	&\{...\} = -{f^k}\left( E \right)\sum\nolimits_n^{}{[E_{\bot n}^k<E]I_{mn}^{kk}}	
	+ \sum\nolimits_n^{}{[E_{\bot n}^k<E]P_{mn}^{kk}\left( E \right){f^k}\left( E \right)}
	+ \sum\nolimits_n^{l\ne k}{[E_{\bot n}^l<E]P_{mn}^{kl}\left( E \right){f^l}\left( E \right)} \quad \\
	&\quad \ \ \ = \sum\nolimits_n^{l\ne k}{[E_{\bot n}^l<E]P_{mn}^{kl}\left( E \right){f^l}\left( E \right)} - {f^k}\left( E \right)\sum\nolimits_n^{}{[E_{\bot n}^k<E]\left[ I_{mn}^{kk}-P_{mn}^{kk}\left( E \right) \right]}
\end{aligned} \right\}\Rightarrow $
\begin{equation}	\label{Eq238}
	\begin{aligned}
		{J^k} =& \frac{e}{\pi \hbar}\sum\nolimits_m{\int_{-\infty}^{+\infty}{dE[E_{\bot m}^k<E]\left\{ \sum\nolimits_n^{l\ne k}{[E_{\bot n}^l<E]P_{mn}^{kl}\left( E \right){f^l}\left( E \right)}
		-{f^k}\left( E \right)\sum\nolimits_n^{}{[E_{\bot n}^k<E]\left[ I_{mn}^{kk}-P_{mn}^{kk}\left( E \right) \right]} \right\}}} \\
	\end{aligned}
\end{equation}
(\ref{Eq238}), (\ref{Eq118}) $\quad \Rightarrow \quad$

$[_m^k\in \mathbb{O}]\sum\nolimits_n^l{[_n^l\in \mathbb{O}]C_{mn}^{kl}\bar{C}_{nm}^{lk}} = [_m^k\in \mathbb{O}]I_{mm}^{kk}$; (\ref{Eq125}) $\quad \Rightarrow \quad$

$\begin{aligned}
	&[_m^k,_q^p\in \mathbb{O}]\sum\nolimits_n^l{[_n^l\in \mathbb{O}]C_{mn}^{kl}\bar{C}_{nq}^{lp}} = [_m^k,_q^p\in \mathbb{O}]I_{mq}^{kp};\quad \sphericalangle \quad _m^k = _q^p \quad \Rightarrow \quad [_m^k\in \mathbb{O}]\sum\nolimits_n^l{[_n^l\in \mathbb{O}]C_{mn}^{kl}\bar{C}_{nm}^{lk}} = [_m^k\in \mathbb{O}]I_{mm}^{kk}; \quad (\ref{Eq125}) \quad \Rightarrow	\\
	&[_m^k\in \mathbb{O}]I_{mm}^{kk} = [_m^k\in \mathbb{O}]\sum\nolimits_n^l{[_n^l\in \mathbb{O}]P_{mn}^{kl}} = [_m^k\in \mathbb{O}]\sum\nolimits_n^l{\left( [_n^l = _n^k]+[_n^l\ne _n^k] \right)[_n^l\in \mathbb{O}]P_{mn}^{kl}}	\\
	&\hspace{152pt} = [_m^k\in \mathbb{O}]\sum\nolimits_n^{}{[_n^k\in \mathbb{O}]P_{mn}^{kk}}+[_m^k\in \mathbb{O}]\sum\nolimits_n^{l\ne k}{[_n^l\in \mathbb{O}]P_{mn}^{kl}}\quad \Rightarrow \\
	 &[_m^k\in \mathbb{O}]\sum\nolimits_n^{l\ne k}{[_n^l\in \mathbb{O}]P_{mn}^{kl}} = [_m^k\in \mathbb{O}]\left\{ I_{mm}^{kk}-\sum\nolimits_n^{}{[_n^k\in \mathbb{O}]P_{mn}^{kk}} \right\}
	 = [_m^k\in \mathbb{O}]\left\{ \sum\nolimits_n^{}{[_n^k\in \mathbb{O}]I_{mn}^{kk}}-\sum\nolimits_n^{}{[_n^k\in \mathbb{O}]P_{mn}^{kk}} \right\}\quad \Rightarrow \\
	 &[_m^k\in \mathbb{O}]\sum\nolimits_n^{l\ne k}{[_n^l\in \mathbb{O}]P_{mn}^{kl}} = [_m^k\in \mathbb{O}]\sum\nolimits_n^{}{[_n^k\in \mathbb{O}]\left\{ I_{mn}^{kk}-P_{mn}^{kk} \right\}}\quad \Leftrightarrow
\end{aligned}$
\begin{equation}	\label{Eq239}
	[E_{\bot m}^k<E]\sum\nolimits_n^{l\ne k}{[E_{\bot n}^l<E]P_{mn}^{kl}\left( E \right)} = [E_{\bot m}^k<E]\sum\nolimits_n^{}{[E_{\bot n}^k<E]\left\{ I_{mn}^{kk}-P_{mn}^{kk}\left( E \right) \right\}}
\end{equation}

(\ref{Eq238}), (\ref{Eq239}) $\quad \Rightarrow \quad$

$\displaystyle \left\{... \right\} = \sum\nolimits_n^{l\ne k}{[E_{\bot n}^l<E]P_{mn}^{kl}\left( E \right){f^l}\left( E \right)}-{f^k}\left( E \right)\sum\nolimits_n^{l\ne k}{[E_{\bot n}^l<E]P_{mn}^{kl}\left( E \right)}\quad \Rightarrow $

$\displaystyle \left\{... \right\} = \sum\nolimits_n^{l\ne k}{[E_{\bot n}^l<E]P_{mn}^{kl}\left( E \right)\left\{ {f^l}\left( E \right)-{f^k}\left( E \right) \right\}}; \quad (\ref{Eq238}) \quad \Rightarrow $

$\displaystyle {J^k} = \frac{e}{\pi \hbar}\sum\nolimits_{mn}^l{\int_{-\infty}^{+\infty}{dE[E_{\bot m}^k<E][E_{\bot n}^l<E]P_{mn}^{kl}\left( E \right)\left\{ {f^l}\left( E \right)-{f^k}\left( E \right) \right\}}}; \quad (\ref{Eq237}) \quad \Rightarrow \quad (\ref{Eq135})$	~\hfill$\blacksquare$
\end{strip}	
\end{proof}

\vspace{-22pt}
Formula (\ref{Eq135}) in case of two-terminal structure coincides with the known result \cite[(49)]{Bib018}.

Higher in sums and sets with a index $l\ne k$, convolution is only on the first character: $l$ is dummy index, $k$ is free index. For example
\begin{equation}	\label{Eq240}
	\sum\nolimits_{}^{l\ne k}{\left\{... \right\}} = \sum\nolimits_{}^l{[l\ne k]\left\{... \right\}}\ne \sum\nolimits_{}^{kl}{[l\ne k]\left\{... \right\}}
\end{equation}

\newpage
\subsection{Low bias voltages}	\label{sec:E_2}

\begin{proof}[\textsc{Proof~\ref{prop:11}}]
 (\ref{Eq143}), (\ref{Eq133}), (\ref{Eq013}) $\quad \Rightarrow \quad$
\begin{equation}	\label{Eq241}
	{\{\varepsilon _{\text{F}}^k\to {{\varepsilon}_{\text{F}}}\}}^{k\in \mathbb{E}}
\end{equation}
\begin{equation}	\label{Eq242}
	f\left( \varepsilon,\varepsilon _{\text{F}}^k \right) := {F_{-1}}\left( [\varepsilon _{\text{F}}^k-\varepsilon ]/{{\mu}^k} \right)
\end{equation}

\newpage
\begin{strip}
$\left. \begin{aligned}
	&(\ref{Eq241}) &\Rightarrow& \quad f\left( \varepsilon,\varepsilon _{\text{F}}^k \right)\to f\left( \varepsilon,{{\varepsilon}_{\text{F}}} \right)+\left( \varepsilon _{\text{F}}^k-{{\varepsilon}_{\text{F}}} \right){{\partial}_2}f\left( \varepsilon,{{\varepsilon}_{\text{F}}} \right)	\\
	&(\ref{Eq242}), (\ref{Eq144}) &\Rightarrow& \quad {{\partial}_2}f\left( \varepsilon,{{\varepsilon}_{\text{F}}} \right) = {{\dot{F}}_{-1}}\left( [{{\varepsilon}_{\text{F}}}-\varepsilon ]/{{\mu}^= } \right)/{{\mu}^= } \quad \Rightarrow \\
	&(\ref{Eq138}) &\Rightarrow& \quad {{\dot{F}}_{-1}}\left( \eta \right) = {{\partial}_{\eta}}{{(1+{e^{-\eta}})}^{-1}} = {e^{-\eta}}{{(1+{e^{-\eta}})}^{-2}} = {{({e^{+\eta /2}}+{e^{-\eta /2}})}^{-2}} = \tfrac{1}{4}{{\cosh}^{-2}}(\eta /2) \quad
\end{aligned}\right\}\Rightarrow$
\begin{equation}	\label{Eq243}
	f\left( \varepsilon,\varepsilon _{\text{F}}^l \right)-f\left( \varepsilon,\varepsilon _{\text{F}}^k \right)\to \left( \varepsilon _{\text{F}}^l-\varepsilon _{\text{F}}^k \right)\frac{1}{4{{\mu}^= }}{{\cosh}^{-2}}\left( [{{\varepsilon}_{\text{F}}}-\varepsilon ]/[2{{\mu}^= }] \right)
\end{equation}
(\ref{Eq137}), (\ref{Eq242}), (\ref{Eq243}) $\quad \Rightarrow \quad$
\begin{equation}	\label{Eq244}
	{{{\rm I}}^k}\to \sum\nolimits_{mn}^l{\left( \varepsilon _{\text{F}}^l-\varepsilon _{\text{F}}^k \right)\frac{1}{4{{\mu}^= }}\int_{-\infty}^{+\infty}{d\varepsilon [\lambda _m^k<\varepsilon ]{{| C_{mn}^{kl}\left( \varepsilon \right) |}^2}[\varepsilon >\lambda _n^l]{{\cosh}^{-2}}\left( [{{\varepsilon}_{\text{F}}}-\varepsilon ]/[2{{\mu}^= }] \right)}}
\end{equation}
\end{strip}
\noindent
(\ref{Eq139}), (\ref{Eq244}), (\ref{Eq013}), (\ref{Eq147}), (\ref{Eq146}) $\quad \Rightarrow \quad$ (\ref{Eq145})	
\end{proof}

\subsection{Low temperatures}	\label{sec:E_3}

\newpage
\begin{proof}[\textsc{Proof~\ref{prop:10}}]
 (\ref{Eq124}), (\ref{Eq141}) $\quad \Rightarrow \quad$
\begin{equation}	\label{Eq245}
	{f^k}\left( E \right)\to [E<E_{\text{F}}^k]
\end{equation}
(\ref{Eq135}), (\ref{Eq245}) $\quad \Rightarrow \quad$

\renewcommand{\qedsymbol}{}
\begin{strip}
$\left. \begin{aligned}
	&{J^k}\to \frac{e}{\pi \hbar}\sum\nolimits_{mn}^l{\int_{-\infty}^{+\infty}{dEP_{mn}^{kl}\left( E \right)\left\{ [E<E_{\text{F}}^l]-[E<E_{\text{F}}^k] \right\}}} \\
	&\quad [E<E_{\text{F}}^l]-[E<E_{\text{F}}^k] = \left\{ \begin{aligned}
		0  \quad &\Leftarrow \quad   E<E_{\text{F}}^l, \quad E<E_{\text{F}}^k  {}  {} \\
		+1  \quad &\Leftarrow \quad   E<E_{\text{F}}^l, \quad E\ge E_{\text{F}}^k \quad  \Rightarrow \quad  E_{\text{F}}^k<E_{\text{F}}^l \\
		-1  \quad &\Leftarrow \quad   E\ge E_{\text{F}}^l, \quad  E<E_{\text{F}}^k \quad  \Rightarrow \quad E_{\text{F}}^k>E_{\text{F}}^l \\
		0  \quad &\Leftarrow \quad   E\ge E_{\text{F}}^l, \quad E\ge E_{\text{F}}^k  {}  {}
	\end{aligned} \right.\quad \Rightarrow \quad \\
	&\quad [E<E_{\text{F}}^l]-[E<E_{\text{F}}^k] = \operatorname{sgn} \left( E_{\text{F}}^l-E_{\text{F}}^k \right)[\min (E_{\text{F}}^k,E_{\text{F}}^l)\le E<\max (E_{\text{F}}^k,E_{\text{F}}^l]
\end{aligned} \right\}\Rightarrow \quad (\ref{Eq142}) $	~\hfill$\blacksquare$
\end{strip}
\end{proof}

\begin{proof}[\textsc{Proof~\ref{prop:12}}]
 (\ref{Eq141}), (\ref{Eq140}), (\ref{Eq144}) $\quad \Rightarrow \quad$ ${{\mu}^= }\to 0$ $\quad \Rightarrow \quad$
\begin{equation}	\label{Eq246}
	\frac{1}{4{{\mu}^= }}{{\cosh}^{-2}}\left( [{{\varepsilon}_{\text{F}}}-\varepsilon ]/[2{{\mu}^= }] \right)\to \delta \left( {{\varepsilon}_{\text{F}}}-\varepsilon \right)
\end{equation}
(\ref{Eq146}), (\ref{Eq246}) $\quad \Rightarrow \quad$ (\ref{Eq148})	
\end{proof}

\pdfbookmark[1]{Acknowledgments}{Acknowledgments}

\section*{Acknowledgments}

The author is grateful to A. M. Yafyasov for scientific advising that contributed to the writing of this paper. The author acknowledges Saint-Petersburg State University for a research grant 73031758.

\pdfbookmark[1]{References}{References}
\bibliographystyle{spmpsci+}
\bibliography{Article}

\end{document}